\documentclass[11pt]{article}

\usepackage{amsmath,amsfonts, amssymb, amsthm, eurosym,geometry,ulem,graphicx,caption,color,setspace,sectsty,comment,footmisc,caption,pdflscape,subfigure,array,hyperref}
\usepackage{blkarray,multirow} 

\usepackage{geometry}
\usepackage{enumerate}
\usepackage{thmtools} 
\usepackage{thm-restate}

\usepackage{mathtools}

\usepackage{accents}

\usepackage{bbm}

\usepackage{algorithm2e}

\usepackage{tikz}

\usepackage{hyperref}
\hypersetup{
    colorlinks=true,
    linkcolor= blue,
    filecolor= magenta,      
    citecolor = green,
    linktocpage = true
}

\usepackage{array} 

\usepackage[flushleft]{threeparttable}

\usepackage[skip=0.5\baselineskip]{caption}

\usepackage[makeroom]{cancel}

\theoremstyle{plain}

\newtheorem{lemma}{Lemma}
\newtheorem*{cor}{Corollary}
\newtheorem{claim}{Claim}
\newtheorem{example}{Example}
\newtheorem{definition}{Definition}
\newtheorem{assumption}{Assumption}

\theoremstyle{definition}

\theoremstyle{remark}

\newcommand{\R}{\mathbb{R}}

\newcommand{\block}[1]{
  \overbrace{\begin{matrix} T\boldsymbol{I}_T - \boldsymbol{1}_{T\times T} & \quad \quad T\boldsymbol{I}_T - \boldsymbol{1}_{T\times T} & \cdots & T\boldsymbol{I}_T - \boldsymbol{1}_{T\times T} \end{matrix}}^{#1}
}
\newcommand{\blockk}[1]{
  \overbrace{\begin{matrix} T\boldsymbol{I}_T - \boldsymbol{1}_{T\times T}  & \cdots & T\boldsymbol{I}_T - \boldsymbol{1}_{T\times T} \end{matrix}}^{#1}
}

\DeclareMathOperator*{\plim}{plim}

\newcolumntype{L}[1]{>{\raggedright\let\newline\\arraybackslash\hspace{0pt}}m{#1}}
\newcolumntype{C}[1]{>{\centering\let\newline\\arraybackslash\hspace{0pt}}m{#1}}
\newcolumntype{R}[1]{>{\raggedleft\let\newline\\arraybackslash\hspace{0pt}}m{#1}}

\usepackage[round]{natbib}
\usepackage{doi}
\usepackage{graphicx}
\usepackage{textgreek}
\usepackage{amsmath}
\usepackage{booktabs,caption}
\usepackage[flushleft]{threeparttable}
\usepackage{longtable}

\begin{document}

\begin{titlepage}

\title{Correlated Synthetic Controls }  
\author{Tzvetan Moev\thanks{ I would like to thank Frank DiTraglia and Max Kasy for excellent supervision. I am grateful for the invaluable comments of Otso Hao, Xiyu Jiao, Anders Kock, Barbara Petrongolo, Stanislav Slavov, Yi Ying Tan  and Frank Windmeijer.}} 

\maketitle
\begin{abstract}
\noindent \normalsize{Synthetic Control methods have recently gained considerable attention in applications with only one treated unit.  Their popularity is partly based on the key insight that we can \textit{predict} good synthetic counterfactuals for our treated unit. However, this insight of predicting counterfactuals is generalisable to microeconometric settings where we often observe many treated units. We propose the Correlated Synthetic Controls (CSC) estimator for such situations: intuitively, it creates synthetic controls that are correlated across individuals with similar observables. When treatment assignment is correlated with unobservables, we show that the CSC estimator has more desirable theoretical properties than the difference-in-differences estimator. We also utilise CSC in practice to obtain heterogeneous treatment effects in the well-known Mariel Boatlift study, leveraging additional information from the PSID.}
\vspace{0.3in}

\bigskip

\end{abstract}
\pagebreak \newpage

\thispagestyle{empty}
\setcounter{page}{0}
\end{titlepage}
\pagebreak \newpage

\doublespacing

\section{ Introduction}

The Fundamental Problem of Causal Inference states that we cannot directly infer the causal effect of some intervention for a single individual because we do not simultaneously observe what happens to them with and without treatment. In other words, we cannot do a \textit{within-person} comparison of the two potential outcomes. What if we can construct a surrogate for the potential outcome without treatment for every treated individual? Achieving this would allow us to get around the Fundamental Problem of Causal Inference by approximating the ideal \textit{within-person} comparison. In the case of one treated unit observed for a long period, the Synthetic Control\footnote{Main abbreviations used in the paper: ATT -- Average Treatment Effect on the Treated; CSC -- Correlated Synthetic Controls; DGP -- Data Generation Process; DiD -- Difference-in-Differences; fDiD -- feasible Difference-in-Differences;  PSC -- Penalised Synthetic Control; PSID -- Panel Study of Income Dynamics; SC -- Synthetic Control;   iDiD -- infeasible Difference-in-Differences} (SC) Method \citep{aba10} method constructs such a synthetic counterfactual by taking a weighted average of the control units. This insight of constructing counterfactuals can also be generalised to the setting considered in this thesis: a binary treatment affects many individuals which we observe for a relatively short time period. Similarly to \citet{lho20}, this paper develops an estimator that uses the control individuals to create SCs for all treated individuals. We call the estimator Correlated Synthetic Controls (CSC) because it builds counterfactuals that are similar across treated individuals with similar observable characteristics.

To explain how CSC works, we should note that generalising SC to panels with many treated units observed for a short period that are common in applied microeconomics is not trivial. In particular, there are two extreme approaches that we can take to achieve this. One is to construct a separate SC for every treated individual. The other is to aggregate the time-series for all treated individuals and for all control individuals, because often they belong to well-defined groups such as cities. For instance, in the Mariel Boatlift study of \citet{car90}, we can group all treated individuals to Miami and all control individuals to other US cities. Then, we can construct a SC for Miami by combining other cities' time series. The CSC estimator is a compromise between these two extremes. Similarly, a very recent paper by \citet{ben21} that appeared while this thesis was being written builds an estimator that balances a related but different tradeoff between two extremes.\footnote{ In Section \ref{subsec_man}, we show how the two extremes that we consider are different from their distinction between \textit{pooled} and \textit{separate} SCs.}

The two extremes that we consider are related to the distinction between pooled (or homogeneous) coefficients and fixed (or heterogeneous) coefficients models in the panel data literature \citep{pes08}. However, there is a middle ground: correlated random coefficients models \citep{woo03, hsi08}. The CSC approach postulates such a model for the weights allocated to different donors. In particular, the weights used for contructing the SCs are allowed to differ across treated individuals in a deterministic way based on their observables, similarly to coefficients in correlated random coefficients models. Thus, individuals with similar observables will have similar (or correlated) SCs. This is an attractive approach in our setting because CSC overcomes the challenge of short $T$ by using information from comparable treated observations when constructing the SCs.

Beyond proposing a novel estimator, we make three contributions to the literature. Firstly, we compare the theoretical properties of CSC to those of Difference-in-Differences (DiD). When treatment is strongly correlated with unobservable characteristics, e.g., we suspect selection on unobservables, CSC should be used because its estimate of the treatment effect has a smaller estimation error than DiD. This provides one good reason for empirical researchers to choose CSC in the context of a panel with many treated units, even though DiD is often considered the default option in such cases \citep{ark20}. Importantly, this theoretical result does not depend on assumptions specific to CSC but holds more generally for many estimators from the SC family when used in a setting with many treated units.

Secondly, via a simulation study we can compare CSC not only to DiD, but also to Penalised Synthetic Control (PSC) \citep{lho20}, our estimator's closest sibling from the family of SC estimators. We provide an infeasible estimator that estimates consistently the parameter of interest (the Average Treatment Effect on the Treated, or ATT) and study the conditions, under which the three feasible estimators approximate its behaviour. The simulation also confirms our main theoretical result. 

Thirdly, we illustrate how researchers can use CSC in practice via studying the effect of immigration on wages and labour supply in the context of the Mariel Boatlift \citep{car90}. Intuitively, what we do is construct a synthetic doppelganger for every treated worker based on workers in other states. Next, we show that CSC performs slightly better than PSC with real data in terms of predicting good counterfactuals in the pre-treatment period. In contrast to previous studies, our empirical application uses an alternative data source, namely the Panel Study of Income and Dynamics (PSID), because the credibility of typical data sources used for evaluations of the Mariel Boatlift\footnote{Mostly variants of CPS. See \citet[p.6-12]{per18}.} have been recently questioned \citep{cle19}. We conclude by showing how CSC can be used to estimate the heterogeneous treatment effects of immigration. 

So, CSC provides a useful addition to the toolkit of empirical researchers for conducting causal inference and the rest of this paper attempts to illustrate this point. Section \ref{sec_mot} presents a motivating example, in which the estimator seems more appropriate than standard causal inference techniques.  The formal setup of the problem and the construction of the CSC estimator is detailed in Section \ref{sec_est_cons}. Next, its theoretical properties are explored in Section \ref{sec_the} whereas Section \ref{sec_sim} provides the result from our simulation exercise. Section \ref{sec_emp} applies the estimator to the Mariel Boatlift.  Some avenues for further research and limitations are discussed in Section \ref{sec_concl}. Supplementary examples and further clarifications are contained in Appendix \ref{ap_oth_res} whereas Appendix \ref{sec_proof} contains the proofs of the main results in the thesis. The coding for the simulation and the empirical application can be found in \href{https://github.com/tzvetanmoev/Correlated-Synthetic-Controls}{this GitHub repository}.

\section{Motivating Example} \label{sec_mot}
The key purpose of this section is to provide an example, in which the CSC estimator seems more appropriate than other causal inference techniques. A key implication of models in economic geography postulates that market access is an important determinant of the spatial distribution of economic activity \citep{kru91, dav02}. \citet{red08} examine this causal lint by exploiting the German division after the Second World War as an exogenous shock to the market access enjoyed by cities. In particular, they speculate that West German cities close to the East-West border experienced a bigger decline in market access relative to other West German cities. The decline should additionally be more pronounced for small cities close to the border relative to big cites. According to economic geography, the reduction in market access in certain cities would lead to a decline in economic development in these cities.

So, \citet{red08} would like to test this mechanism by measuring the possibly heterogeneous treatment effect (depending on city size) of the East-West border on the population of cities close to the border.\footnote{In the paper, they argue that population growth is a good proxy for economic development.} Suppose that we have data on German cities for two pre-treatment years (1922, 1937) and one post-treatment year (1952). A key decision problem which they are facing is which causal inference technique to use where the two main candidates are DiD and SC. To estimate via DiD, we specify the two-way fixed effects model:
\begin{gather}
    Population_{it} = \rho + \gamma_i + \delta_t + \tau D_{it} + e_{it}
\label{mot_DiD}
\end{gather}
where $i$ is city, $t$ is time, $\gamma_i$ are city-level fixed effects, $\delta_t$ are time fixed effects, $e_{it}$ are idiosyncratic shocks and $D_{it}$ is a treatment indicator which equals to $1$ only for the treated cities in the post-treatment indicator. The parameter of interest is $\tau$, which can be interpreted as identifying the ATT and which we estimate via DiD. However, we may be concerned that with just two pre-treatment periods we cannot evaluate if the parallel trends assumption which is necessary for identifying ATT via DiD holds.  Moreover, as suggested by \citet{red08}, the treatment effects may be heterogeneous, i.e., they are bigger for smaller cities. Recent work on DiD with heterogeneous effects has shown that in such cases $\tau$ from (\ref{mot_DiD}) will not identify the ATT, even when the parallel trends assumption holds \citep{cha20}.\footnote{With few discrete covariates, we can consider doing DiD separately for each group of cities (e.g.\ small and big cities). However, this approach becomes infeasible if we want to add continuous covariates or if we have many discrete covariates.} Thus,  DiD does not seem optimal.

Alternatively, we can consider creating a separate SC for every treated city in our sample. The benefit of this approach is that we can get the individual treatment effects for every city with small estimation error under certain conditions \citep{aba10}. This will allow us to capture the heterogeneity among different cities without having to rely on parallel trends. A SC approach would postulate a model where the population of a treated city is a weighted average of populations of untreated cities:
\begin{gather*}
    Population_{it} = \sum_{j=1}^{n_0} w_{ij} Population_{jt} + \tau_i D_{it} + e_{it}
\end{gather*}
where $w_{ij}$ is the weight of untreated city $j$ for treated city $i$, $n_0$ is the total number of untreated cities and $D_{it}$ and $e_{it}$ are defined as above. The parameter of interest is the individual treatment effect $\tau_i$. Note that $w_{ij}$ which we can estimate via SC are constrained to be non-negative and to sum up to 1.  There are two related problems with using SC in this case, especially with a lot of control cities and short $T$ as in our case. Firstly, there might be more than one combination of donors that perfectly matches the time series of a certain treated city pre-treatment, i.e., a multiplicity of solutions for $w_{ij}$. For example, consider some treated city $H$ with pre-treatment values $\{420, \ 480\}$ and suppose that there are three other cities with the following populations in the pre-treatment period (1922 and 1937):
$$ A=\{400, \ 450 \} \quad \quad \quad B = \{440, \ 510\} \quad \quad \quad C =\{500, \  600 \}  $$
where using control cities $A$ and $B$ with weights $0.5$ and $0.5$ or using control cities $A$ and $C$ with weights $0.8$ and $0.2$ both match exactly the time series for $H$. Secondly, if we try to estimate separate SC for every treated city with just two pre-treatment periods, we risk over-fitting, even when multiplicity of solutions is not an issue. 

If neither DiD, nor SC is appropriate, we would have to look out for another techniques such as CSC. Essentially, CSC modifies the weights in such a way that it tackles overfitting and multiplicity of solutions simultaneously without relying on parallel trends. We let the weights depend on observable characteristics and allow for an individual fixed effect $\gamma_i$:  
\begin{gather*}
    Population_{it} = \gamma_i + \sum_{j=1}^{n_0} w_{ij} Population_{jt} + \tau_i D_{it} + e_{it} \\ 
    w_{ij} = w^{small}_j Small_i + w^{big}_j Big_i
\end{gather*}
where $w_{ij}$ is the weight of donor $j$ on treated unit $i$, $D_{it}$ is treatment indicator and $\tau_i$ captures the fact that we get an individual estimate of the treatment effect. $Small_i $ and $Big_i$ indicator functions for being a small or big city. While the weights $w_{ij}$ are still constrained to be non-negative and sum up to 1, the most significant difference is that we constrain them to be the same for all small cities and the same for all big  cities. As a result, we do not get multiplicity of solutions because for small cities the weights are simultaneously balancing the time-series of several cities. Analogically, over-fitting is less of a concern because we have fewer free parameters and each set of weights is exploiting information from different treated cities belonging to the same group. 

So, CSC should be preferred to standard SC and DiD in this case. In contrast to DiD, it does not rely on assumptions like parallel trends and can naturally accommodate heterogeneity of treatment effects. On the other hand, CSC does not suffer from multiplicity of solutions and overfitting as opposed to SC.

\section{Estimator Construction}\label{sec_est_cons}

\subsection{Related work}
In this section we will briefly review the literature and place the CSC in a wider context. In a seminal paper, \citet{aba10} introduced the SC method for constructing synthetic counterfactuals for a single treated unit. This was followed by a series of empirical papers, using SC to estimate the treatment effects of various macro interventions such as the effect of Brexit \citep{bor17b} or the effect of the German unification \citep{aba15}. Despite the wave of applied papers using the method, little was known about its theoretical properties until \citet{dou18} and the work of Bruno Ferman and coauthors \citep{bot19, fer20, fer19}. \citet{dou18} present a general framework for estimators which nests difference-in-difference and SC as special cases and illustrate the connections between the two estimators. Regarding whether DiD or SC is preferable in applications, \citet{fer19} provide conditions under which SC has better theoretical properties. In a related article, \citet{fer20} shows that when our data is generated from an interactive fixed effects models, the SC can yield an asymptotically unbiased estimate of the treatment effect under certain conditions. Lastly, \citet{bot19} show that even if the SC does not match perfectly the true time series in the pretreatment period, the SC method can still yield a meaningful estimate of the treatment effect.

The papers discussed so far have considered the case of one treated unit. However, the idea of constructing counterfactuals is generalisable to settings with many treated individuals: we can construct SC for every individual that has been treated in our dataset. Several recent papers have proposed estimators that work in this context such as the synthetic DiD \citep{ark20} and matrix completion \citep{ath21}. However, the papers closest in spirit for our contribution are \citet{lho20} and \citet{ben21}. Firstly, \citet{lho20} propose the PSC which tackles a key problem when trying to generalise the original SC to a setting with many treated units, namely the multiplicity of solutions. Secondly, while primarily aimed at building an estimator inspired by SC for staggered adoption, the estimator that \citet{ben21} propose also works for the case of many treated units and has a similar motivation to CSC. The CSC estimator that we propose is explicitly aimed at the many treated unit setting without staggered adoption and so its closest sibling is \citet{lho20}. CSC improves on existing techniques by bringing insights from the panel data literature in order to exploit information from treated individuals that are similar in terms of observables.

\subsection{Set-up} \label{subsec_setup}
We introduce in this subsection a formal framework for thinking about the family of SC estimators with many treated units based on \citet{dou18} and more generally Imbens' Sargan Lecture \citeyearpar{imb21}. This formulation of the problem is useful because it allows us to see how the different estimators in the literature relate to each other and to compare the optimisation problem that they solve. Moreover, it brings home the main insight from the different SC methods: we can translate the causal inference problem into a prediction problem (under certain assumptions). 

Suppose that we observe many treated units for a relatively short time period. We have available a panel dataset with the outcome variable of $N$ individuals who are observed for $T$ periods. A policy is implemented \textit{once}\footnote{Thereby ruling out staggered adoption.} at time $T_0 + 1$ and we are interested in estimating its ATT time $t$ (denoted by $\tau_t$ throughout the paper) and the individual treatment effects (denoted by $\tau_{it}$). The intervention affects a total of $n_1$ people which forms our treatment group, whereas the other $n_0 \equiv N – n_1$ people form our control group, which we refer to interchangeably as the donors. So, we can observe the outcome variable $y_{it}(D)$ for both donors $i \in \{1,2, \dots, n_0\}$ and treated individuals $i \in \{n_0+1, n_0 + 2, \dots n_0 + n_1\}$ in the pre-treatment period $t \in \{1, 2, \dots, T_0 \}$ and the post-treatment period $t \in \{T_0+1, \dots T\}$. $D$ indicates if an individual is treated, so that we only have $D=1$ for treated units after $T_0$. In addition to the outcomes $y_{it}(D)$, there is data on $K$ time-invariant covariates in the $(K \times 1)$ vector $\boldsymbol{x}_i = (x_i^{(1)}, x_i^{(2)}, \dots x_i^{(k)})'$. We can define formally the unobserved individual treatment effects as $\tau_{it} = y_{it}(1) – y_{it}(0)$ and ATT at time $t$ is: 
$$\tau_t = \frac{\sum_{i=n_0+1}^N \tau_{it} }{n_1} = \frac{\sum_{i=n_0+1}^N [y_{it}(1) -  y_{it}(0)]}{n_1} $$
where we use ATT at time $t$ and $\tau_t$ interchangeably.

In the empirical application to the $1980$ Mariel Boatlift with PSID data, for instance, the outcome variable is the wages of individuals for the period $1974-1984$, i.e.,  $T = 11$. On the other hand, the treated individuals are people who live in Miami whereas the donors are people who live elsewhere in the US. More generally, we can neatly  illustrate this set-up via the matrix of outcomes $y_{it}(D)$, called $\boldsymbol{\Theta}$: 

\begin{align}
\boldsymbol{\Theta} \equiv 
\begin{blockarray}{cccc|ccc}
    \begin{block}{(cccc|ccc@{\hspace*{7pt}})}
        y_{1,1}(0) & y_{1,2}(0) & \dots & y_{1, T_0}(0) & y_{1,T_0+1}(0) & \dots & y_{1,T}(0) \\
        y_{2,1}(0) & y_{1,2}(0) & \dots & y_{2, T_0}(0) & y_{2, T_0+1}(0) & \dots & y_{2,T}(0) \\
        \vdots & \ddots  &  \ddots & \vdots  &  \vdots & \dots  & \vdots  \\
        y_{n_0, 1}(0) & \dots  & \dots  & y_{n_0, T_0}(0) & y_{n_0, T_0+1}(0) & \dots & y_{n_0,T}(0) \\
         &  &   &  & \\
        \cline{1-7}
         &  &   &  & \\
        y_{n_0+1,1}(0) & y_{n_0+1, 2}(0)  &  \dots  & y_{n_0+1, T_0}(0) & y_{n_0+1,T_0+1}(1) &  \dots & y_{n_0+1,T}(1)\\
        y_{n_0+2,1}(0) & y_{n_0+2, 2}(0)  &  \dots  & y_{n_0+2, T_0}(0) & y_{n_0+2,T_0+1}(1)  & \dots & y_{n_0+2,T}(1) \\
        \vdots & \ddots  &  \ddots  & \vdots & \vdots & \dots  & \vdots  \\
        y_{n_0+n_1,1}(0) & y_{n_0+n_1, 2}(0)  &  \dots  & y_{n_0+n_1, T_0}(0) & y_{n_0+n_1,T_0+1}(1) & \dots & y_{n_0+n_1,T}(1)  \\
    \end{block} 
\end{blockarray}
\label{full_theta_matrix}
\end{align}
Note that $\boldsymbol{\Theta}$ has four panels. The top-left panel contains the outcome variables for the donors in the pre-treatment period ($\boldsymbol{Y}_{n_0}^{pre}$) whereas the top-right panel is filled with the same values for the post-treatment period  ($\boldsymbol{Y}_{n_0}^{post}$). The pre-treatment outcome variables for the treated observations are in the bottom-left panel ($\boldsymbol{Y}_{n_1}^{pre}$) and the post-treatment values for the treated group are in the bottom-right panel ($\boldsymbol{Y}_{n_1}^{post}$). Thus, we can rewrite $\boldsymbol{\Theta}$ as a block matrix: 
\begin{align}
   \boldsymbol{\Theta} = \begin{blockarray}{c|c}
\begin{block}{(c|c@{\hspace*{3pt}})}
        \boldsymbol{Y}_{n_0}^{pre}(0) & \boldsymbol{Y}_{n_0}^{post}(0) \\
        \cline{1-2}
         \boldsymbol{Y}_{n_1}^{pre}(0) & \boldsymbol{Y}_{n_1}^{post}(1) \\
    \end{block}
\end{blockarray}
\label{block_theta_matrix}
\end{align}
Our objective is to estimate the ATT of the policy occurring at time $T_0+1$. In the post-treatment period $ t> T_0$, we only observe the treated outcomes for the treatment group $\boldsymbol{Y}_{n_1}^{post}(1)$ whereas the untreated outcomes for the treatment group $\boldsymbol{Y}_{n_1}^{post}(0)$ are unobserved. In order to calculate $\tau_i$, we need to observe both. The SC method tackles this data limitation by estimating the unobserved $\widehat{\boldsymbol{Y}}_{n_1}^{post}(0)$ using the rest of the information in $\boldsymbol{\Theta}$ and then utilises these estimates to calculate individual treatment effect $\tau_{it}$ for the post-treatment period. In this sense, it turns the causal inference problem into a prediction problem as we are trying to predict the missing values of $\boldsymbol{Y}_{n_1}^{post}(0)$. This is an extremely powerful insight, as it allows us to augment causal inference with cutting-edge techniques from statistics and machine learning that are great at predicting out-of-sample values. Thus, what we are actually trying to estimate is the unobserved matrix with untreated values for all individuals: 
\begin{align}
\widehat{\boldsymbol{\Theta}}(0) =  
\begin{blockarray}{c|c}
\begin{block}{(c|c@{\hspace*{3pt}})}
        \boldsymbol{Y}_{n_0}^{pre}(0) & \boldsymbol{Y}_{n_0}^{post}(0) \\
        \cline{1-2}
         \boldsymbol{Y}_{n_1}^{pre}(0) & \widehat{\boldsymbol{Y}_{n_1}^{post}}(0) \\
    \end{block}
\end{blockarray}
\label{block_untr_theta_matrix}
\end{align}
and the key question is how to use the information from the observed panels, namely $\boldsymbol{Y}_{n_0}^{pre}(0)$, $\boldsymbol{Y}_{n_0}^{post}(0)$ and $\boldsymbol{Y}_{n_1}^{pre}(0) $ to predict $\widehat{\boldsymbol{Y}_{n_1}^{post}}(0)$.\footnote{ Given this general framework, one may wonder why we are bothering with SC, given that we are faced with a problem that requires good prediction of missing values in a matrix and we have an abundance of techniques for such situations in the matrix completion literature \citep{ath21}. We discuss this question further in Section \ref{subsec_dgp_the} after we introduce a potential DGP for $\boldsymbol{\Theta}$.}

\subsection{The SC Method}
The set-up in the previous subsection raises the question how we can impute the missing counterfactuals and so tackle the Fundamental Problem of Causal Inference via approximating the within-person comparison. In the context of $n_1 = 1$ and long $T_0$, \citet{aba10} introduced the SC method for cases with a single treated unit. As discussed in Section \ref{sec_mot}, their method assumes that the outcome vector of the treated unit $y_{n_0+1, t}$ can be represented as a weighted average of $n_0$ donors' outcomes:
\begin{gather}
    y_{n_0+1,t} = \sum_{j=1}^{n_0} w_j y_{jt} + \epsilon_{n_0+1,t}
    \label{scm_wei}
\end{gather}
where $n_0+1$ indicates the single treated observation as it would be shown in $\boldsymbol{\Theta}$ and the weight on donor $j$ is given by $w_j$. In addition, $\epsilon_{n_0+1,t}$ are idiosyncratic shocks that are independent and have mean 0. Since we are taking a weighted average, the weights are constrained to be between 0 and 1 and to sum up to 1. As such, we can interpret the weights as probabilities.\footnote{To give an example, \citet{bor19} construct a synthetic Britain before 2016 to study the effect of Brexit. Their chosen synthetic Britain is made up of 51\% US, 17\% Italy, 14\% New Zealand, 11\% Hungary, 5\% Germany and other countries with a weight of  1\% or less  and matches pretty closely the true Britain. See \href{https://academic.oup.com/view-large/figure/186550277/uez020fig2.jpg}{Figure 2} in \citet{bor17}.}

We may then wonder how the weights on donors are estimated. Equation (\ref{scm_wei}) points towards the idea that we are essentially regressing the time-series for the treated unit on the time-series for the donors in the pre-treatment period, except that the coefficients are constrained to be non-negative and sum up to 1. Thus, one possibility would be to rewrite the SC model in (\ref{scm_wei}) as a constrained optimisation problem:
\begin{gather}
    \min_{w_j} \sum_{t=1}^{T_0} \left(y_{n_0+1,t} - \sum_{j=1}^{n_0} w_j y_{jt} \right)^2  \quad \quad s.t. \quad \quad \sum_{j=1}^{n_0} w_j = 1 \quad w_j \geq 0
    \label{sc_est}
\end{gather}
where the two constraints ensure that the weights can be interpreted as probabilities. So, the formulation amounts to a constrained regression problem.\footnote{Note that the literature has disagreed on whether there are better  ways to estimate $w_j$ which also take account of covariates \citep{bot19}. See Appendix \ref{ap_cov_vs_out} for details.}

\subsection{Many treated units}\label{subsec_man}
While the previous section considered the case of $n_1 = 1$ and implicitly assumed $T_0$ is not short, many interesting applications in empirical microeconomics involve having $n_1 > 1$ and short $T_0$. This section illustrates why implementing SC is not trivial in this case. Similarly to Section \ref{sec_mot}, the problems can be made clear via a particular example.\footnote{See Section \ref{sec_emp} for more details} Consider Card's study \citeyearpar{car90} of the effect of the  Mariel Boatlift, a massive way of Cuban immigration to Miami in 1980, on natives' wages. 

Suppose that we have panel data $\boldsymbol{\Theta}$ on many treated workers in Miami and on many control workers in other cities that did not experience the treatment of immigration. For the reasons outlined in Section \ref{sec_mot}, we would like to use a SC approach rather than DiD. We are faced with two possibilities: either create a \textit{separate SC} for every treated individual in Miami or construct a single \textit{pooled SC} for all treated individuals in Miami that matches well the average wage in Miami. 

The benefit of \textit{separate SC} is that we can obtain an individual treatment effect which allows us to explore how the causal effects vary across different groups, e.g., for low-skilled versus high-skilled workers. However, we can quickly run into two problem: multiplicity of solutions and overfitting. The first issue renders SC infeasible here \citep{lho20}. Overfitting will result from the fact that SC is essentially a constrained regression but we will have very few observations and many regressors given small $T_0$ and big $n_0$.

Fortunately, one can still overcome the problem of multiplicity of solutions via changing the objective function in such a way that it picks the optimal SC out of all perfect SCs based on some criteria. This is the approach taken by \citet{lho20} who consider a similar setting as in this paper. Their estimator selects the SC which has the closest values of the matched variable to the actual treated observation. The example in Appendix \ref{ap_multiplicity} provides an illustration. However, they are still solving a complicated quadratic problem for every treated observation separately. Thus, their estimator can still suffer from overfitting, depending on what variables are included. 

On the other hand, we may consider calculating a (single) \textit{pooled SC} for all treated individuals in Miami, i.e., calculate a single set of weights. This approach has the benefit of not running into multiplicity of solutions and overfitting: we will be solving a constrained quadratic optimisation problem with $n_1 \times T$ observations and $n_0$ regressors rather than just $n_1$ obsevations and $n_0$ regressors as in \textit{separate SC}. However, with many treated individuals, the single SC will not be matching too well the time series of some outlying observations. For example, individual who earns a very low wage will get the same SC as an individual who earns a very high wage. While this is a serious limitation, it illustrates that if we can reduce the heterogeneity across individuals by classifying them into groups, then we can create SCs within each group and tackle  multiplicity of solutions and overfitting.

Moreover, under certain conditions, the \textit{pooled SC} method nests as a special case one common approach for policy evaluation of the Mariel Boatlift \citep{per18}. We can call this approach \textit{city-level SC}: aggregate the individual-level data on a city-level and then run SC on the city-level. So, we are creating SCs for the average wage in Miami by combining other US cities. Appendix \ref{ap_equiv} provides a set of restrictions, under which the \textit{pooled SC} reduces to \textit{city-level SC}: essentially, the space of weights, from which \textit{city-level SC} selects, is a subset of the space of weights, from which which \textit{pooled SC} selects. Unfortunately, the \textit{city-level SC} also has certain limitations. In particular, inference remains a challenge in this case, as we observe just one estimate of the treatment effect. More broadly, inference with SCs is still work in progress \citep{che20}. Furthermore, when $T_0$ is short, using SC methods on aggregate units is not recommended by \citet{aba10}, as the estimate of the ATT can be very biased.\footnote{In addition, when aggregating, we may be losing the heterogeneity of treatment effects. In contrast to \textit{separate SC} and \textit{pooled SC}, we cannot explore how the effect varies for different groups. One solution to this problem is to estimate different SCs for every group of interest, e.g. low-skilled vs high-skilled workers. I thank Barbara Petrongolo for pointing this out to me. However, when we suspect that there is considerable heterogeneity across many categories or even across continuous covariates, then simply dividing our dataset into groups may quickly become infeasible. }

So, it seems that neither the \textit{pooled SC}, nor the \textit{separate SC} are without problems. However, they motivate our estimator CSC, as it balances between the two extremes. This insight for balancing \textit{pooled SC} and \textit{separate SC} has also been exploited in the \textit{partially pooled SC} proposed very recently by \citet{ben21} who made their paper public while we were working on CSC. However, the distinction which the authors draw between \textit{pooled SC} and \textit{separate SC} is slightly different from ours. In particular, we and \citet{ben21} understand \textit{separate SC} similarly: create a separate synthetic counterfactual for every treated individual. However, we differ in how we understand \textit{pooled SC}. For us, \textit{pooled SC} refers to a SC-type of estimator that gives every treated individuals the same set of weights. In contrast, \citet{ben21} understand \textit{pooled SC} to be creating separate SCs  for every treated individual. Crucially, instead of matching the individual time-series of a person, their \textit{pooled SC} is picking weights that match the average time-series for all treated individuals. In a certain sense, we are pooling the \textit{weights} across individuals whereas they are pooling the \textit{outcomes} of treated individuals. As a result of this and other differences,\footnote{ Another key conceptual difference with \citet{ben21} is that they create an estimator for the staggered adoption case with multiple but not too many treated unit whereas we are focused on an estimator with non-staggered adoption with many treated unit. As a result, \citet{ben21} do not consider how issues such as overfitting and multiplicity of solutions affect the properties of their estimator. So, empirical researchers should choose between CSC and \textit{partially pooled SC} estimator based on whether they are facing staggered adoption or many treated units. } the two estimators are related but remain different in important ways. 

\subsection{Correlated Synthetic Controls}\label{sec_CSC}
The key distinction between \textit{separate SC} for each unit and \textit{pooled SC} is analogous to a key distinction in the panel data literature between fixed (or heterogeneous) coefficients models and pooled (or homogeneous) coefficients. If we want to allow for individual-specific coefficients $\beta_i$ in a linear regression model $y_{it} = \beta_i x_{it} + \epsilon_{it}$ without an intercept, we can run separate regressions for every observation $i$. This is equivalent to adding an interaction between the covariates and dummies for every observation as $y_{it} = \sum_{j=1}^N \mathbbm{1}\{j=i\} \beta_j x_{it} + \epsilon_{it}$. This model is sometimes called fixed \textit{coefficients} model in analogy to the fixed \textit{effects} models \citep{bal08}. We can estimate it either by OLS on the last equation or by a separate regression for every unit $i$, i.e., we run $N$ regressions of the type $y_{it} = \beta x_{it} + \epsilon_{it}$. On the other extreme, we can impose homogeneity on the slopes $\beta = \beta_1 = \dots = \beta_N$ in the linear model $y_{it} = \beta x_{it} + \epsilon_{it}$, which we can estimate via pooled OLS.

However, there is a middle ground between the two extremes of fixed coefficients and pooled coefficients: (correlated) random coefficients models \citep{woo03, sur11, hsi08}. Basically, these  models allow us to capture the heterogeneity in $\beta_i$ (in contrast to pooled OLS) without requiring a lot of data\footnote{This is a result of the incidental parameter problem: in $y_{it} = \beta_i x_{it} + \epsilon_{it}$ we need $T \to \infty$ as well, if we would like to estimate  $\beta_i $ consistently. However, in correlated random coefficients if we let $\beta_i = \beta + \psi z_i$ and $z_i$ is a discrete variable with categories $\{1,2, \dots, K\}$, we only need  $T*n_k \to \infty$ for consistency of $\beta_i$ where $n_k$ is the number of people for which  $z_i=k$. Or in other words we can get consistency only with $n_k \to \infty$.} for consistency (in contrast to fixed coefficients). The idea is that we allow $\beta_i$ to  differ across $i$ in a deterministic way based on some time-invariant $z_i$: we specify $\beta_i = \beta + \psi z_i$. These type of models are a generalisation of Chamberlain's (correlated) random \textit{effects} models that allow only the intercept to depend on observables \citep{cha82, cre08}. 

The main contribution of this paper is to apply this idea to the weights that postulated treated group's potential outcome without treatment as weighted average of donors' outcomes. So, we shall assume that the weights for treated unit $i$ follow such a (correlated) random coefficient model:
\begin{gather}
  y_{it}(0) = \eta_i + \sum_{j=1}^{n_0} w_{ij} y_{jt}(0) + e_{it} \quad s.t.  \label{CSC_mod} \\
  w_{ij} = \underbrace{\omega_j}_{ind.-invariant} + \underbrace{\boldsymbol{x}_i \boldsymbol{\alpha}^{j}}_{ind.-specific} \quad \quad (Random \ Coef.) \nonumber \\
    \sum_{j=1}^{n_0} w_{ij} = 1 \quad \quad  \forall j: w_{ij} \geq 0 \quad \quad (SC \ Constraints)  \nonumber 
\end{gather}
where we also allow for an intercept $\eta_i$, following suggestions in \citet{fer19} and \citet{dou18}. In the $(Random \ Coef.)$ constraint, each weight $w_{ij}$ has two parts:  an individual invariant part $\omega_j$ and individual specific part $\boldsymbol{x}_i \boldsymbol{\alpha}^{j}$ where $\boldsymbol{x}_i$ is a $(1\times K)$ vector of covariates and $\boldsymbol{\alpha}^{j}$ is a $(K \times 1)$ vector of coefficients. The individual-invariant part ensures we are close to the \textit{pooled SC} approach discussed above whereas the individual-specific allows us to introduce heterogeneity across weights as in \textit{separate SC}. 

We call this estimator CSC. Intuitively, the SCs of two individuals are similar (or correlated) if  they are similar in terms of observables, as it was the case for small and big German cities in the Motivating Example. It is useful to consider another example: 
\begin{example}
    Suppose we have three treated individuals $i \in \{ Ed, \ Mihai, \ Yi \ Ying \}$ and two  covariates: i) years of education and ii) marital status for being \underline{married}, \underline{single} or \underline{other}. Firstly, holding education the same across them, CSC will yield \textit{separate} SCs with different weights on donor $j$ if our individuals  have three different values of marital status. Let Yi Ying be \underline{single}, Ed be \underline{married} and Mihai be \underline{other}:
    $$ w_{Yi \ Ying, j} = \omega_j + \alpha^{single} \neq w_{Ed, j} = \omega_j + \alpha^{married} \neq  w_{Mihai, j} = \omega_j + \alpha^{other}$$
    However, if they have the same martial status, CSC will result in a single set of weights for all of them (\textit{pooled} SC). Secondly, let education differ across the treated individuals: Mihai has 11 years, Ed - 12 and Yi Ying - 17. Suppose also that Ed and Mihai share the same marital status that is different from Yi Ying. Then, Ed and Mihai will get similar weights on donors, albeit not exactly the same, as they are similar in terms of observables. In that sense, their SCs will be correlated. In contrast, Yi Ying will get a set of weights which is considerably different, given her covariates.
    \label{ex_weights}
\end{example}

A few things should be remarked in light of Example \ref{ex_weights} in order to relate CSC to the discussion on panel data. Firstly, the estimator balances between estimating one set of weights for all treated individuals and separate sets of weights for treated individuals. The reason is that if \textit{Mihai} and \textit{Ed} have exactly the same covariates, we will give them the same weights and so the same SC that balances between fitting well both of their time series simultaneously. Note that we implicitly assume that their time-series will not be too different, if they are similar in terms of observables. 

Secondly, we can see how CSC solves the problem of multiplicity of solutions. Suppose that \textit{Ed} and \textit{Mihai} both have multiple exact SCs and have the same covariates but their time series are not exactly the same.  Then the set of potential donor combinations that exactly match Mihai's time series will not be intersecting with the set of potential donor combinations that exactly match Ed's time series. Since we can only pick one set of weights for them, we will pick the one set of weights that creates a donor which matches simultaneously both of their time series as closely as possible but not exactly. In that way, we tackle multiplicity of solutions: by construction, there does not exist a SC which exactly matches both of their time series at the same time.

Thirdly, the use of covariates in constructing the weights allows us to capture heterogeneity across treated units. If \textit{Ed} and \textit{Mihai} have similar covariates, they get correlated SCs which allows us to capture the heterogeneity between them and the other treated person, namely \textit{Yi Ying}, who might have very different covariates. In a sense, we end up with two different groups of treated individuals which is similar to recent work on group fixed effects \citep{bon15}.

\subsection{Estimation} \label{subsec_est}
Next, one may wonder how the $w_{ij}$ are estimated. We can rewrite the correlated random coefficients model of the weights as a constrained optimisation problem for the pre-treatment period. In particular,the parameters $\alpha_j^{(k)}$ and $\omega_j$ can be found by solving: 
     \begin{equation} 
        \begin{gathered}
            \max_{\alpha_j^{(k)}, \omega_j} \sum_{i=1}^{n_1} \sum_{t=1}^{T_0} \left( y_{it} - \eta_i - \sum_{j=1}^{n_0} y_{jt} \omega_j - \sum_{j=1}^{n_0} \sum_{k=1}^K \alpha^{k}_j y_{jt} x_i^{(k)}  \right)^2 \quad \quad s.t. \\
     \quad  \forall i: \  \sum_{j=1}^{n_0} \left(\omega_j + \sum_{k=1} x_{i}^{(k)} \alpha^k_j\right) = 1 \quad \quad \quad  \forall (i,j): \ \omega_j + \sum_{k=1}^K x_{i}^{(k)} \alpha^k_j \geq 0
            \label{opt_pr_scal_2}
        \end{gathered}
    \end{equation}  
where $y_{it}$ are the potential outcomes without treatment, $\eta_i$ is an individual fixed effect,  $\omega_j$ is the individual-invariant part of the weight on donor $j$,  $x_i^{(k)}$ is the value of the $k$-th covariate for the treated observation  $i$ and $\alpha_j^{(k)}$ is the coefficient of covariate $k$ in determining the weight on donor $j$. One intuition for estimating the SC model with random coefficients (weights) is that we are essentially regressing the vector of outcomes for the treatment group on the outcomes for the donors and an interaction between donors' outcomes and treatment group's covariates, given some constraints on the coefficients. More formally, this is as a constrained regression of $y_{it}$ for the treated group on donors' $y_{jt}$ and an interaction between donors' $y_{jt}$ and treated group's $x^{(k)}_i$.  See Appendix \ref{sec_est_det} for more details, including a formulation in terms of the block matrices of $\boldsymbol{\Theta}$ and the use of package \texttt{CVXR} \citep{fu19} to code CSC in \texttt{R}.

Before proceeding to CSC's theoretical properties, it is important to discuss one limitation of the estimator: it does \textit{not} allow for continuous covariates.\footnote{I would like to thank Anders Kock for encouraging me to pursue this line of thought.} The reason is a restriction imposed by the fact that the random coefficients weights sum up to 1. To gain some intuition, consider the example of $K=3$ covariates with weights:
$$ \sum_{j=1}^{n_0} \left( \omega_j + x_i^1 \alpha_j^1 +  x_i^2 \alpha_j^2 + x_i^3 \alpha_j^3 \right)=1 $$
which can be rewritten as:
\begin{gather}
    x_i^{(1)} \sum_{j=1}^{n_0} \alpha_j^{(1)} + x_i^{(2)} \sum_{j=1}^{n_0} \alpha_j^{(2)} +  x_i^{(3)} \sum_{j=1}^{n_0} \alpha_j^{(3)} = 1 - \sum_{j=1}^{n_0} \omega_j
    \label{anders_res}
\end{gather}
The last expression should hold for every treated $i$ but note that the right-hand side is independent of $i$. If we have continuous $x^{(k)}_{i}$ and $ i \geq 2$, then (\ref{anders_res}) will not hold in general for all $i$ and the optimisation problem will be infeasible, as it is impossible to satisfy the constraint exactly. Nevertheless, suppose that $x^k_{(i)}$  are dummies for a single categorical variable with three mutually exclusive categories (e.g., married, divorced, other), then for the restriction to hold it is sufficient to have: $\sum_{j=1}^{n_0} \alpha_j^{(1)} = \sum_{j=1}^{n_0} \alpha_j^{(2)} = \sum_{j=1}^{n_0} \alpha_j^{(3)} $. Therefore, while the \textit{sum} of the $\alpha^{(k)}_j$ over $j$ will be constrained to be the same across the three categories $k$, the coefficient on the same donor $j$ across two different $k$ and $q$ groups, respectively $\alpha_j^{(k)}$ and $\alpha_j^{(q)}$, can be different. This is good news, as we can then achieve our objective of having heterogeneity in the weights on different donors.\footnote{A similar restriction on the sums applies if $x_i^{(k)}$ are \textit{not} mutually exclusive categories}

Regarding continuous predictors, it is still possible to integrate them by recoding such a variable (e.g., wages) as a discrete predictor (e.g., income brackets). However, there are approaches, allowing us to handle continuous covariates in a more systematic way. For example, we can pre-process continuous covariates prior to estimating CSC, as done in coarse exact matching \citep{iac12}. The idea is that there is an extra step which allow us to balance the donors and the treatment groups in terms of some continuous observables. As a result, we do not need to worry about controlling for this particular predictor in CSC.

\section{Theoretical properties}\label{sec_the}
This section compares the estimation error of the true ATT $\tau$ from the estimated $\hat{\tau}^{DiD}$ from DiD and from the estimated $\hat{\tau}^{CSC}$ from CSC when the data is generated from an interactive fixed effects model.\footnote{We use sometimes interactive fixed effects models and factor models interchangeably, although we try to prioritise the former term.} This is important, because it illustrates in what circumstances CSC should be preferred to DiD. Specifically, in the case of many treated units, SC methods are not the default choice in empirical work, even though sometimes they can perform better.

The main takeaway from this section is that CSC should be preferred in the cases when we believe that the treatment is  correlated with unobservable characteristics, i.e., selection on unobservables, under the data generation process (DGP) we consider. For this condition to hold, we also need the SCs to do a good job at predicting the pretreatment outcomes of the treated group. On the other hand, if the treatment is not strongly correlated with unobserved characteristics, DiD might be a better choice. We assume that our data is generated from an interactive fixed effects models. More generally, to the best of our knowledge there have been just a few papers exploring the theory behind SC in the case of many treated units \citep{lho20, ben21}. Unfortunately, these studies do not provide much guidance on when SC methods should be preferred over DiD. This section fills this gap by making a small step towards providing such conditions.

\subsection{Data Generating Process}\label{subsec_dgp_the}
As common in the literature on SC \citep{aba10, fer19}, we shall assume that each $y_{it}$ in $\boldsymbol{\Theta}$ follows an interactive fixed effects model \citep{bai09, moo15, hsi18}:
\begin{align}
  y_{it} =& \ \boldsymbol{\theta}_t \boldsymbol{x'_i}  + D_{it} \tau + v_{it} \label{dgp_sim_1} \\ 
     =& \ \boldsymbol{\theta}_t \boldsymbol{x'_i}  + D_{it} \tau + \boldsymbol{\lambda_t}  \boldsymbol{\mu_i} + \epsilon_{it} \label{dgp_sim_2}
\end{align}
where $\boldsymbol{x}_i$ is a $(1 \times K)$ vector of time-invariant covariates, $ \boldsymbol{\theta}_t$ are $(1 \times K)$ time-varying coefficients on these covariates,  $D_{it}$ is the treatment assignment, $v_{it}$ is a composite error term such that $v_{it} = \boldsymbol{\lambda_t}  \boldsymbol{\mu_i} + \epsilon_{it}$, $\boldsymbol{\lambda_t} $ is a $(1 \times F)$ vector of common factors, $\boldsymbol{\mu_i}$ is a $(F \times 1)$ vector of factor loadings and $\epsilon_{it}$ are idiosyncratic shocks. The main quantity of interest is $\tau$.\footnote{The reason why $\tau$ is not the ATE is because we are not going to assume that treatment assignment is randomly assigned.} The \textit{interactive} fixed effect structure $\boldsymbol{\lambda_t}  \boldsymbol{\mu_i}$ is a generalisation of the traditional \textit{additive} fixed effects $\lambda_t + \mu_i$.\footnote{In fact, for $F=2$, using the interactive fixed effects with $\boldsymbol{\mu_i} = (1, \mu_i)'$ and $\boldsymbol{\lambda_t} = (\lambda_t, 1) $ reduces to the additive fixed effects.} We can interpret the interactive factor structure as saying that each individual $i$ has some unobserved characteristics $\boldsymbol{\mu}_i$ such as ability or motivation that determine their outcome $y_{it}$. However, at different points in time, different unobserved factors from $\boldsymbol{\mu}_i$ matter, implying that their effects are time-varying which justifies including time-varying common factors $\boldsymbol{\lambda_t} $. 

We will also impose some further structure on the DGP:

\begin{restatable}[\underline{DGP restrictions}]{assumption}{assumptiondgp}
    We assume that:
    \begin{enumerate}[(i)]
          \item Treatment-not-at-random: $Cov(D_{it}, \boldsymbol{\mu_i}) \neq \boldsymbol{0}_F \implies Cov(D_{it}, v_{it}) \neq 0 $
           \item Errors are iid with $E[\epsilon_{it}] = 0$ and $E[\epsilon^2_{it}] \leq \infty$. They are independent from all other random variables.
          \item DGP of $\boldsymbol{x}_i$: $\boldsymbol{x}_i$ is independent of $ \epsilon_{it}$ and $D_{it} $  but $Cov(\boldsymbol{\mu}'_i, \boldsymbol{x}_i ) \neq \boldsymbol{0}_{F\times K}$
          \item $\boldsymbol{\mu_i}$ are stochastic with $E[\boldsymbol{\mu_i}] =  \boldsymbol{\mu}$
          \item $\boldsymbol{\lambda_t}$ are fixed parameters with $\frac{1}{T}\sum_{t=1}^T \boldsymbol{\lambda_t} =  \boldsymbol{ \bar{\lambda} }$
          \item One post-treatment period $T=T_0+1$
    \end{enumerate}
    where $D_{it}$ is the treatment indicator, $\boldsymbol{\mu}_i$ is the column vector of factor loadings, $\boldsymbol{\lambda}_t$ is the row vector of the common factors, $\boldsymbol{0}_F$ is a $(F \times 1)$ column vector of zeros and  $\boldsymbol{0}_{F\times K}$ is a $(F \times K)$ matrix of zeros.
    \label{ass_th}
\end{restatable}
Let us detail the different subparts of \textbf{Assumption \ref{ass_th}}. Firstly, the treatment indicator $D_{it}$ is assumed to be correlated with the unobserved component $\boldsymbol{\mu}_i$ (\textbf{Assumption \ref{ass_th}}.i). This means that essentially  $D_{it}$ is endogenous ($Cov(D_{it}, v_{it}) \neq 0$) and so we cannot estimate $\tau$ consistently by simply ignoring the composite structure of the error term and fitting (\ref{dgp_sim_1}). Secondly, we assume that for all individuals $i$ and time periods $t$ the errors $\epsilon_{it}$ are iid and mean zero with a finite second moment (\textbf{Assumption \ref{ass_th}}.ii). So, we do not allow $\epsilon_{it}$ to follow more complicated autoregressive processes.  Thirdly, we assume that $\boldsymbol{x}_i$ are uncorrelated with treatment assignment \textbf{but} could be correlated with the composite error term $v_{it}$ via the factor loadings (\textbf{Assumption \ref{ass_th}}.iii). This means that $\boldsymbol{x}_i$ are endogenous and we cannot estimate their coefficients directly. Fourthly, we assume that $\boldsymbol{\mu}_i$ and $\boldsymbol{\lambda}_t$ are respectively stochastic and fixed with means $E[\boldsymbol{\mu}_i] = \boldsymbol{\mu}$ and $\bar{\boldsymbol{\lambda}}$ that are \textit{not} constrained to be 0 necessarily (\textbf{Assumptions \ref{ass_th}}.iv and \textbf{\ref{ass_th}}.v). The reason for this particular choice stems from a suggestion in Hsiao \citeyearpar[p.665]{hsi18} that in cases of short $T$ and long $N$ assuming fixed common factors and stochastic factor loadings is reasonable. Lastly, we assume only one post-treatment period, as it simplifies the algebra and allows us to abstract from considerations of dynamic ATT (\textbf{Assumptions \ref{ass_th}}.vi).

Given this DGP, a natural question is why we should use SC, given the abundance of estimators for interactive fixed effects models \citep{hsi18}. For instance, Xu \citeyearpar{xu17} proposes the Generalised SC which directly fits such a model to the data in $\boldsymbol{\Theta}$. More generally, we can simply model $\boldsymbol{\Theta}$ directly via matrix completion techniques, as in \citet{ath21}. 

There are at least three good reasons for not taking this approach. 
Firstly, in practice, we rarely know the true model generating $\boldsymbol{\Theta}$: it could be an interactive fixed effects model as in \citet{xu17} but it could also be a vector autoregressive process \citep{aba20b}. On the other hand,  SC methods can be shown to work well under other DGPs \citep{aba10, ben21}. While with interactive fixed effects models we risk misspecifying the DGP of $\boldsymbol{\Theta}$, SC allows more flexibility with respect to the true DGP. So, we will not risk fitting a factor model to a data that actually follows a vector autoregression. 

Secondly, \textbf{Assumption \ref{ass_th}}.i allows for a non-zero correlation between our only time-varying covariate $D_{it}$ and the factor loadings $\boldsymbol{\mu}_i$. We can then rewrite  model (\ref{dgp_sim_2}) as 
\begin{gather*}
    y_{it} = D_{it} \tau + \left( \boldsymbol{\theta}_t \quad \boldsymbol{\lambda}_t \right) \begin{pmatrix} \boldsymbol{x}'_t \\ \boldsymbol{\mu}_i
    \end{pmatrix}  + \epsilon_{it}
\end{gather*} 
where we treat $\boldsymbol{\theta}_t$ as common factors and $\boldsymbol{x}_i$ as factor loadings. This allows us to apply \textbf{Remark 5.9.} from \citet{hsi18} which states that $\hat{\tau}$ will not be estimated consistently by common approaches for estimating interactive fixed effects models. Thus, even if we wanted to use an interactive fixed effects models, it will not estimate the quantity of interest consistently.

Thirdly, SC focuses on imputing the bottom-right panel of $\boldsymbol{\Theta}$ rather than predicting all of its entries as many matrix completion techniques would do. Given that often we are not interested in fitting the rest of $\boldsymbol{\Theta}$, SC can be more efficient as it exploits the structure of the matrix and focuses on a smaller prediction task. Of course, if we were faced with a more general matrix, e.g.\ different individuals are treated at different times, then matrix completion methods will be more appropriate.

\subsection{CSC Bound}
Let us consider how we can establish an upper bound on the estimation error of CSC. By estimation error, we mean the absolute value of the difference between the estimated ATT $\hat{\tau}^{CSC}$ and the true $\tau$, i.e. $|\hat{\tau}^{CSC} - \tau|$. Unfortunately, given how complicated the optimisation problem solved by CSC is, there is generally no obvious analytical solution for the weights and so we cannot derive an exact expression for the estimation error. For this reason, we will construct an upper bound for the estimation error, drawing on High-Dimensional Statistics. We will then compare this upper bound to the exact expression, derived for DiD. This comparison can then be interpreted as conservative: even under the worst possible estimation error for CSC, there are still certain conditions, under which CSC should be chosen over DiD.

Before establishing the bound, we need to make several additional assumptions. Firstly, we will make the \underline{Exact Fit} assumption which states that CSC will find a weighted average of donor units that exactly match treated individual $i$ in terms of both outcomes and covariates. Versions of this assumption for the case of a single treated unit are common in the SC literature, e.g., see \citet{aba10} and \citet{bot19}. 

\begin{assumption}[\underline{Exact Fit}]
Consider a $( n_0 \times n_1)$ matrix $\boldsymbol{W}$ such that each weight $w_{ij}$ follows a (correlated) random coefficient model $w_{ij} = \omega_j + \sum_{k=1}^K \alpha_{kj}x^{(k)}_i$ and each column of $\boldsymbol{W}$ sums up to 1 with every entry being non-negative. Matrix $\boldsymbol{W}$ satisfies exact fit in \textbf{pre-treatment period} if:
\begin{enumerate}[(i)]
    \item For the outcome variable: $ \forall t \in \{1,2, \dots, T_0 \}: \quad y_{it} = \sum_{j=1}^{n_0} w_{ij} y_{jt}$
    \item For the covariates: $ \forall k \in \{1,2, \dots, K \}: \quad x^k_{i} = \sum_{j=1}^{n_0} w_{ij} x^{(k)}_j $
\end{enumerate}
\label{ex_fit_as}
\end{assumption}
Note, however, that our statement of the assumption is not identical to Abadie et al.'s original assumption \citeyearpar{aba10}. The reason is that when generalised to the many treated unit settings, they allow for the existence of \textit{a set} of weights satisfying exact fit for each $i$ and not for a unique combination of weights. This might be problematic, due to multiplicity of solutions. Nevertheless, if Abadie's assumption holds, then our \underline{Exact Fit} assumption will hold, as we can simply make $x_i^k$ a vector of dummies for each observation. Therefore, one may see our assumption as an extension of Abadie's. 

Another key difference is that our assumption also ruled out multiplicity of solutions by not assuming a set of weights but just one unique value of the weights that satisfies exact fit. Theoretically, this may seem restrictive and in applications it is unlikely to hold. Suppose that we have two treated units with the same observables: it is implausible that for both we can create exact SCs unless their time series match exactly. We may instead interpret \underline{Exact Fit} as suggesting that there is a unique set of correlated random weights that balances the fits of both treated units as much as possible, albeit not exactly. In that case, \underline{Exact Fit} would only hold approximately but we would have rules out multiplicity of solutions. In future work, we hope to formalise \underline{Approximate Fit} assumption and explore if it can be used to derive a bound on CSC. Intuitively, as a relaxation of \underline{Exact Fit}, it should make the bound on CSC's estimation error less tight.

Returning to the \underline{Exact Fit} assumption, we can prove an extremely useful lemma that generalises a result from \citet{aba10} for the case of many treated units and one post-treatment period. This is Lemma \ref{lem_aba} \underline{Abadie's representation} which allows us to rewrite the estimation error without the unobserved component $\boldsymbol{\mu}$. Note that we need to assume invertability of matrix $\boldsymbol{\lambda}'_{pre} \boldsymbol{\lambda}_{pre}$ and a necessary condition for this is $T_0 > F$ which suggests that our DGP should not be too complicated relative to how many time periods we observe. 

\begin{restatable}[\underline{Abadie's representation}]{lemma}{lemaba}
     Assume that:
    \begin{enumerate}[(i)]
        \item The true DGP is the interactive fixed effects model in (\ref{dgp_sim_1})
        \item Assumption \ref{ex_fit_as}. \underline{Exact Fit} holds
        \item The $F \times F$ matrix $\boldsymbol{\lambda}'_{pre} \boldsymbol{\lambda}_{pre}$ is invertible 
    \end{enumerate}
    Then, we can write the estimation error in the post-treatment period $T=T_0+1$ as:
    \begin{align*}
        \hat{\tau}^{CSC} - \tau = &  \frac{1}{n_1}   \left[ \boldsymbol{\lambda}_T (\boldsymbol{\lambda}'_{pre}\boldsymbol{\lambda}_{pre})^{-1} \boldsymbol{\lambda}'_{pre} \sum_{i=n_0+1}^N   \left(  \sum_{j=1}^{n_1} \hat{w}_{ij}  \boldsymbol{\epsilon}_{j,pre} - \boldsymbol{\epsilon}_{i,pre} \right)\right] \\  
        + &\frac{1}{n_1} \left[ \sum_{i=n_0+1}^N  \left( \epsilon_{iT} - \sum_{j=1}^{n_0} \hat{w}_{ij} \epsilon_{jT} \right)  \right]
   \end{align*}
   where $\hat{w}_{ij}$ is the random coefficient weight of donor $j$ on individual $i$, $\boldsymbol{\lambda}_T$ are the common factors at post-treatment time $T$, $\boldsymbol{\lambda}'_{pre}\boldsymbol{\lambda}_{pre}$ is the $F\times F$ matrix of interacted pre-treatment common factors, $\boldsymbol{\epsilon}_{j,pre}$ is a ($T_0 \times 1$) vector of error terms for observation $i$ in the pretreatment period and $\epsilon_{jT}$ is the error term for individual $j$ at time $T$.
   \label{lem_aba}
\end{restatable}
\begin{proof}
    The main proofs are relegated to Appendix \ref{sec_proof}. See Appendix \ref{sec_proof_lem_aba} for the proof of this Proposition.
\end{proof}
The proof of the lemma involves rewriting the estimation error in terms of the DGP for the outcome variables $y_{it}$ and substituting out the $\boldsymbol{\mu}_i$, using the \underline{Exact Fit} assumption. While this may seem complicated, it only consists of several tedious algebraic steps. 

Next, we also assume that the errors $\epsilon_{it}$ are iid \texttt{subGaussian}$(\sigma^2)$.\footnote{See \citet{rig15} for a more detailed discussion of \texttt{subGaussian} variables} While this assumption may seem restrictive, it only constraints the distribution not to have fat tails relative to a Normal distribution with variance $\sigma^2$. Moreover, it allows us to draw on techniques from High-Dimensional Statistics which are extremely useful for bounding different quantities such as the estimation error in our case \citep{rig15}. 
\vspace{-25pt}
\begin{assumption}[\underline{SubG Errors}]
$\epsilon_{it}$ are iid \texttt{subG}($\sigma^2$) 
\label{subg_err}
\end{assumption}
Lastly, we make two further assumptions: the matrix $\boldsymbol{\lambda}'_{pre} \boldsymbol{\lambda}_{pre}$ in the pre-treatment period is invertible (so that we can use Lemma \ref{lem_aba} \underline{Abadie's Representation}) and we work on an Euclidean space (which is useful, as in the proof we need to calculate operator norms of matrices). We can find an upper bound on the estimation error for $\hat{\tau}^{CSC}$:

\begin{restatable}[\underline{CSC Bound Estimation Error}]{proposition}{propCSCesterr}
    Suppose that:
    \begin{enumerate}[(i)]
        \item DGP is given by interactive fixed effects model in (\ref{dgp_sim_1}) with Assumption \ref{ass_th}.
        \item Assumption \ref{ex_fit_as} \underline{Exact Fit} holds.
        \item Assumption \ref{subg_err} \underline{SubG Errors} holds.
        \item $F \times F$ matrix $\boldsymbol{\lambda}'_{pre} \boldsymbol{\lambda}_{pre}$ is invertible 
        \item We work on an Euclidean space
    \end{enumerate}
    Then, with probability at least $1 - \frac{3}{\exp(0.25h^2)}$ in the single post-treatment period $T$ the estimation error for the ATT estimated by CSC satisfies for $h>0$:
    \begin{gather}
    |\widehat{\tau}^{CSC} - \tau| <  \left(
   \frac{F\lambda_{max}^2}{\phi_{min} } \sqrt{\frac{2\sigma^2n_0}{T_0} + \frac{h\sigma}{T_0^{1.5}} } \right)
    + \left( \frac{h\sigma F \lambda_{min}^2}{\phi_{max} ^2}\right)
    + h \sigma \sqrt{2}
    \label{CSC_bou}
    \end{gather}
    where $\lambda_{min}$ and $\lambda_{max}$ denote respectively the minimum and maximum common factor $\lambda_{fs}$ in absolute value for either pre-treatment or post-treatment period. Similarly, $\phi_{min}$ and $\phi_{max}$ are respectively the minimum and maximum eigenvalue of matrix $\frac{1}{T_0}\boldsymbol{\lambda}'\boldsymbol{\lambda}$
    \label{prop_CSC_est_err}
\end{restatable}

\begin{proof}
    See Appendix \ref{sec_proof_prop_CSC_est_err}.
\end{proof}
The strategy for the proof is to use \underline{Abadie's representation} for the estimation error. Then, we use various properties of \texttt{subG} variables to bound each quantity in the new expression for the estimation error. Note that when proving the theoretical result we focus on CSC \textit{without} an intercept in the objective function for simplicity.

Let us now interpret the expression for the upper bound of CSC (\ref{CSC_bou}). Firstly, the parameter $h$ controls the probability with which the bound obtains. The bigger $h$, the less precise our bound, but the higher the probability it holds. On the other hand, when we increase $T_0$ or decrease $F$, our estimate of $\tau$ becomes better. This makes intuitive sense, as adding more observations or decreasing the complexity of the DGP should allow us to estimate the weights more precisely. In particular, the first term in the bound will disappear as $T_0 \to \infty$. Furthermore, the upper bound increases with $\sigma^2$ which is approximately the variance of the strictly exogenous error terms in the DGP.\footnote{Approximately, because we assume the $\epsilon_{it}$ are \texttt{SubG}($\sigma^2$) and so their second moment can very well be smaller.} One reason why this might be happening is that with fixed $N$ and $T_0$ our SCs can match well the  part $ \boldsymbol{\theta}_t \boldsymbol{x}_i +  \boldsymbol{\lambda}_t \boldsymbol{\mu}_i$ of the outcomes $y_{it}$ but with bigger $\sigma^2$ it gets difficult to match the idiosyncratic shocks $\boldsymbol{\epsilon}_{it}$. 

The role of $n_0$ in the bound is ambiguous: one may expect that the more donors we have, the better our SCs will be. However, given how we have constructed the proof, the \textit{upper} bound increases in $n_0$, so that the more donors we have, the bigger the estimation error.  We inspect further the importance of $n_0$ in the simulation in the next section. Next, let us consider the importance of the common factors $\lambda_{tf}$. It is clear that we do not want the biggest common factor $\lambda_{max} = \max_{t \in (1,\dots, T_0, T), f \in (1, \dots, F)}|\lambda_{tf}|$ to be too large. This requires that no particular time period $t$ experienced a very large shock as reflected in the common factor. In applications, empirical researchers can leverage their domain knowledge to check if this is the case. 

Lastly, the second term in the bound $\left( \frac{h\sigma F \underaccent{\tilde}{\lambda}^2}{\phi_{max} ^2}\right)$ will usually be quite small, as we are dividing the smallest common factor by the largest eigenvalue. It is, thus, less of a concern than controlling the other two terms.

\subsection{Estimation Error of DiD} \label{sec_DiD_bou}
Although the previous section found an \textit{upper} bound for the estimation error of $\hat{\tau}^{CSC}$, in the case of DiD we can find an exact analytical solution for the estimation error. The reason is that we estimate via OLS the two-way fixed effects model $y_{it} = \rho + \gamma_i + \delta_t + D_{it} \tau + u_{it}$ and so the optimisation problem that is solved to obtain $\hat{\tau}^{DiD}$  has a closed-form solution. Proposition \ref{prop_DiD_est_err} gives an exact asymptotic expression for the estimation error of $\tau^{DiD}$, as the number of treated units $n_1$ goes to infinity:

\begin{restatable}[\underline{DiD Estimation Error}]{proposition}{propDiDesterr}
Suppose:
\begin{enumerate}[(i)]
    \item DGP is given by the interactive fixed effects model in (\ref{dgp_sim}) with Assumption \ref{ass_th}.
    \item We estimate via OLS the model $y_{it} = \rho + \gamma_i + \delta_t + D_{it} \tau + u_{it}$ to get $\hat{\tau}^{DiD}$
\end{enumerate}
Then, as $n_1 \to \infty$ the estimation error for DiD is:
\begin{gather}
    |\hat{\tau}^{DiD} - \tau | \stackrel{p}{
\to} \left| \frac{(\bar{\boldsymbol{\lambda}}_{pre} -  \boldsymbol{\lambda}_T) (\boldsymbol{\bar{\mu} }_{don} - E[\boldsymbol{\mu}_i|D_{it}=1])}{n_0} \right|
    \label{DiD_bou}
\end{gather}  
where $\bar{\boldsymbol{\lambda}}_{pre}$ is $(1\times F)$ vector of the average value of the common factors in the pre-treatment period, $\boldsymbol{\lambda}_T$ is $(1\times F)$ vector of the common factors in the single post-treatmnet period, the $(F \times 1)$ vector $\boldsymbol{\bar{\mu} }_{don}$ is the average of the factor loadings for the donors and $E[\boldsymbol{\mu}_i|D_{it}=1]$ is the $(F \times 1)$ vector with expected value of the factor loadings in the treatment group, .
\label{prop_DiD_est_err}
\end{restatable}

\begin{proof}
    The proof is in Appendix \ref{proof_bound_DiD}.
\end{proof}
The main strategy for the proof is to apply the Frisch-Waugh-Lovell Theorem to obtain an expression for $\hat{\tau}^{DiD}$. Essentially, (\ref{DiD_bou}) gives the expression for the asymptotic bias of $\hat{\tau}^{DiD}$ under an interactive fixed effects models. In a sense, the interactive fixed effects assumption generalises the parallel trends assumption which is the main assumption required for DiD to work.\footnote{We shall return to the parallel trends assumption in the next section.}  In any case, the main thing that (\ref{DiD_bou}) is telling us is that if there are difference in unobservables $\boldsymbol{\mu}_i$ between the treatment group and the control group, then $\hat{\tau}^{DiD}$ will be biased asymptotically. This would be even more problematic, if we suspect that the common factors experience a structural break between the pre-treatment and post-treatment period. Note that the reason why we have $\bar{\boldsymbol{\mu}}_{don}$ as an average and $E[\boldsymbol{\mu}|D_{it}=1] $ as an expectation is because we let $n_1 \to \infty$ and hold $n_0$ constant.

Expression (\ref{DiD_bou}) raises the question whether the estimation error would  disappear if we let $n_0 \to \infty$. If we assume $n_1, n_0 \to \infty$ but $\frac{n_1}{n_0} \to c$, then it can be shown that the estimation error will disappear.\footnote{See the proof to Proposition \ref{prop_DiD_est_err} and in particular equations (\ref{pr_lim_large_n0_n1}) and  (\ref{prob_lim_d_lam_mu}). The quantity in the first equation (\ref{pr_lim_large_n0_n1}) would go to infinity whereas (\ref{prob_lim_d_lam_mu}) would go to some constant. Then, the estimation error would go to 0. } In applications, however, justifying such asymptotics with respect to both $n_1$ and $n_0$ might often be unreasonable: in the German cities example from Section \ref{sec_mot} we have $n_1 > n_0$ but both are relatively small ($n_1=20$ and $n_0 = 99$) whereas in the Mariel Boatlift example we have $n_1 = 41$ and $n_0 = 1039$. Moreover, if we only let $n_0 \to \infty$ and hold $n_1$ fixed, we will get a similar result to (\ref{DiD_bou}), meaning that there will be some estimation error. Thus, in many applications, it is likely that the bias will not disappear.

\subsection{Discussion}

Let us now return to the main question of this section: when should CSC be used over DiD. Firstly, CSC would have a smaller estimation error for $\hat{\tau}$ when the treatment assignment $D_{it}$ is strongly correlated with the unobserved factor loadings $\boldsymbol{\mu}_i$. The reason is that even the conservative bound for CSC is independent of this correlation. This points to one of the biggest advantages of SC methods more generally. We construct synthetic individuals with similar characteristics to the treated units and we ``throw away" bad control units that are very different for our treated individuals. However, unless we apply some additional correction to DiD, these bad control units will bias our estimate of $\hat{\tau}^{DiD}$, as they will form a part of our control group. 
For instance, in the empirical application to the Mariel Boatlift (Section \ref{sec_emp}), we are interested in constructing SCs for low-skilled Miami workers. Since Florida is one of the bottom 20\% of US states in terms of median wages, we are probably looking at low-wage workers at a low-wage state. It is likely that without any further restrictions on the donor pool DiD will be biased due to a correlation between treatment-assignment and unobserved characteristics of Miami workers: we are including high-wage workers from high-wage states which are not relevant for the estimation but remain a part of the donor pool. In this particular case, CSC should be preferred over DiD. 

On the other hand, we can see certain similarities between the two expressions for the estimation errors. For example, if the common factors in $\boldsymbol{\lambda}$ are very volatile, both $\hat{\tau}^{DiD}$ and $\hat{\tau}^{CSC}$ can have a big estimation error. So, perhaps, if we suspect that the common factors have experienced structural break, it would be more appropriate to time-series models. There are situations, in which both DiD and CSC will do badly.

\section{Simulations} \label{sec_sim}
In this section, we perform a simulation exercise to continue our comparison between CSC and DiD. In addition, we also contrast their properties against the PSC of \citet{lho20} and we propose the \textit{infeasible} DiD (iDiD) estimator which is a version of DiD that estimates the parameters of interest consistently but requires knowledge of unobserved components. Using the iDiD as a benchmark, we shall see under what conditions CSC approximates its performance. We also provide evidence that the framework in Section \ref{subsec_setup} which turns the causal inference problem into a prediction problem for the missing values in $\boldsymbol{\Theta}$ is indeed a sensible approach for constructing estimators. 

To that aim, we begin by specifying a DGP for the outcomes in $\boldsymbol{\Theta}$ follows an interactive fixed effects models:
\begin{equation}
    \begin{aligned}
      y_{it} =  \boldsymbol{\beta} \boldsymbol{x'_i} + D_{it} \tau + \boldsymbol{\lambda_t} \boldsymbol{\mu_i} + \epsilon_{it}
    \end{aligned}
    \label{dgp_sim}
\end{equation}
that was described in the previous section. The only substantial difference that we make relative to Section \ref{sec_the} is that now the coefficients on $\boldsymbol{x_i}$ are \textit{not} time-varying but are time-invariant, i.e.\ $\forall t: \ \boldsymbol{\beta} = \boldsymbol{\beta}_t$. This modification is helpful, as it allows us to easily introduce iDiD. We also make the same \textbf{Assumption \ref{ass_th}}, as in Section \ref{sec_the}, except for having just one post-treatment period ($T=T_0+1$): the simulation allows for many post-treatment periods.

\subsection{Bias of DiD}\label{subsec_DiD_con}
Before proceeding with the simulation, we can show why DiD will provide a biased estimate of $\tau$ under the DGP, supplementing Proposition \ref{prop_DiD_est_err} that proves the inconsistency of $\hat{\tau}^{DiD}$. The reason for DiD's bias is that it fails the parallel trends assumption which is necessary for the identification of the ATT. Clarifying DiD's assumptions is of vital importance for applied researchers because DiD is often preferred in empirical applications with many treated units \citep[p.2]{ark20}.\footnote{I would like to thank Barbara Petrongolo and Frank DiTraglia for convincing me of the importance to discuss this questions.} More formally, the parallel trends assumption states  \citep{woo21}:
\begin{gather*}
    E[y_{it}(0) - y_{i1}(0)|D_{is}=1] = E[y_{it}(0) - y_{i1}(0)|D_{is}=0] 
\end{gather*}
where $y_{it}(0)$ is a potential outcome of individual $i$ at time $t$ without treatment, $t$ indicates either a pre-treatment or post-treatment period ($t \in \{1,2, \dots, T_0, \dots, T\}$) and $D_{is}$ indicates treatment assignment in some post-treatment period $s$ where $s \in \{T_0+1, T_0 +1, T\}$. Intuitively, we need the trend in the treatment group to be the same as in the control group. Next, we plug-in the two-way fixed effects model $y_{it}(0) = \rho + \gamma_i + \delta_t  + u_{it}$ which we assume when estimating the ATT with DiD into the assumption:
\begin{gather}
    \cancel{E[\delta_t - \delta_0|D_{is} = 1]} + E[u_{it} - u_{i0}|D_{is} = 1] = \cancel{E[\delta_t - \delta_0|D_{is} = 0]} + E[u_{it} - u_{i0}|D_{is} = 0]
    \label{par_tre_as_2}
\end{gather}
where in the first row the expressions with $\delta$ cancel as they are made up of fixed expressions independent of treatment assignment. Thus, we need the last expression to hold true, in order to identify  $\hat{\tau}^{DiD}$.

We can now check if (\ref{par_tre_as_2}) holds with DGP from an interactive fixed effect model (as given by (\ref{dgp_sim_2})) and Assumption \ref{ass_th}. We begin by adding and subtracting the mean of factor loadings $\boldsymbol{\mu_i}$ and common factors $\boldsymbol{\lambda_t}$, i.e., demeaning them:
\begin{align}
        y_{it} =& \boldsymbol{\beta x_i } + D_{it} \tau + (\boldsymbol{\lambda_t} - \bar{\boldsymbol{\lambda}} + \bar{\boldsymbol{\lambda}} )(\boldsymbol{\mu_i} - E[\boldsymbol{\mu_i}] + E[\boldsymbol{\mu_i}])  + \epsilon_{it} \quad \quad \iff \nonumber \\
        y_{it} =& \underbrace{(\bar{\boldsymbol{\lambda}} E[\boldsymbol{\mu_i}])}_{\rho} + \underbrace{(\boldsymbol{\beta x_i} + \boldsymbol{\bar{\lambda} \mu_i })}_{\gamma_i} + \underbrace{(\boldsymbol{\lambda}_t E[\boldsymbol{\mu_i}])}_{\delta_t} +  \tau D_{it} + \underbrace{ \tilde{\boldsymbol{\mu}}_i   \tilde{\boldsymbol{\lambda} }_t+  \epsilon_{it}}_{u_{it}} \label{DiD_int_fixef}
\end{align}
where we define the $(F \times 1)$ vector $\tilde{\boldsymbol{\mu}}_i = (\boldsymbol{\mu_i} - E[\boldsymbol{\mu_i}] )$ as the demeaned factor loadings and $\boldsymbol{\tilde{\lambda}}_t = (\boldsymbol{\lambda}_t - \boldsymbol{\bar{\lambda}} )$ are the demeaned common factors with $ \boldsymbol{\bar{\lambda}}$ being the $(1 \times F)$ vector of means and $\boldsymbol{\tilde{\mu}_i} - E[\boldsymbol{\mu_i}]$ are the demeaned factor loadings with $E[\boldsymbol{\mu_i}]$ being the $(F \times 1)$ vector of means. Note that the second expression illustrates how the potential outcome $y_{it}$ can be represented as the additive fixed effects structure assumed by the parallel trends assumption of DiD. 

After substituting (\ref{DiD_int_fixef}) in (\ref{par_tre_as_2}) and applying Assumption \ref{ass_th}, the Parallel Trends condition reduces to:\footnote{For more theoretical results on DiD with an interactive fixed effects model, please refer to Proposition \ref{prop_DiD_est_err} for an explicit expression for the asymptotic bias of $\hat{\tau}^{DiD}$ relative to the true $\tau$.}
$$ E[\boldsymbol{\mu}_i |D_{is} = 0](\bar{\boldsymbol{\lambda}}_t - \bar{\boldsymbol{\lambda}}_0 ) = E[\boldsymbol{\mu}_i |D_{is} = 1](\bar{\boldsymbol{\lambda}}_t - \bar{\boldsymbol{\lambda}}_0 ) $$
However, since the demeaned factor loadings $\tilde{\mu}_i$ are not independent from treatment assignment, the last equality will not hold in general. Thus, the parallel trends assumption fails and $\tau^{DiD}$ is biased.

\subsection{Infeasible DiD}
DiD is both inconsistent\footnote{This follows from Proposition \ref{prop_DiD_est_err}. We can also show this more informally in the following way. Estimating $y_{it} = \rho + \gamma_i + \delta_t + \tau D_{it} + u_{it}$ will yield an inconsistent estimate of $\tau$, since $Cov(D_{it}, u_{it}) \neq 0$, meaning that $D_{it}$ is endogenous.  Appendix \ref{subsec_DiD_incons} discusses this point further.} and biased because $\boldsymbol{\mu}_i$ enters the error term $u_{it}$, generating a non-zero correlation between the error term $u_{it}$ and $D_{it}$. However, what if we can control for this non-zero correlation? Then $\hat{\tau}^{DiD}$ would be consistent and unbiased. In particular, if we substract the demeaned interactive fixed effects $\tilde{\boldsymbol{\mu}}_i   \tilde{\boldsymbol{\lambda} }_t$ from both sides  of (\ref{DiD_int_fixef}), then DiD should be consistent. This is clearly infeasible, as we do not observe $\tilde{\boldsymbol{\mu}}_i   \tilde{\boldsymbol{\lambda} }_t$. Nevertheless, since we are running a simulation, we can define:

\begin{definition}[\underline{iDiD}]
The infeasible $\hat{\tau}^{iDiD}$ is given by the OLS estimate of $\tau$ in model:
$$ y_{it} - \boldsymbol{\tilde{\mu_i}\tilde{\lambda}_t} = \alpha + \delta_t + \gamma_i + D_{it} \tau + u_{it}  $$ after subtracting the demeaned interactive fixed effects $\boldsymbol{\tilde{\mu_i}\tilde{\lambda}_t}$ from the outcome variables. \label{def_iDiD}
\end{definition}
\vspace{-20pt}
Given this definition, it is possible to show that $\hat{\tau}^{iDiD}$ will estimate consistently the true quantity of interest $\tau$.

\begin{restatable}[\underline{Consistency of iDiD}]{proposition}{propiDiDcons}
Suppose that:
\begin{enumerate}[(i)]
    \item Data is generated by $y_{it} = \boldsymbol{\beta} \boldsymbol{x'_i}  + \boldsymbol{\lambda_t}  \boldsymbol{\mu_i}  + D_{it} \tau + \epsilon_{it} $ under Assumption \ref{ass_th}.
    \item $\hat{\tau}^{iDiD}$ is given by Definition \ref{def_iDiD}. \underline{Infeasible DiD}
\end{enumerate}
Then, DiD will estimate ATT consistently: $\hat{\tau}^{iDiD} \stackrel{p}{\to} \tau$ 
\label{prop_iDiD_cons}
\end{restatable}
\begin{proof}
    The proof which can be found in Appendix \ref{subsec_prop_iDiD_cons} is analogical to the more involved proof for the estimation error of \textit{feasible} DiD in Proposition \ref{prop_DiD_est_err}.
\end{proof}

Our key motivation for including iDiD in the simulation  is that it provides a strong benchmark, against which we can compare other estimators. For instance, if fDiD performs similarly to iDiD, this would be evidence that the bias induced by non-zero $Cov(\boldsymbol{\mu}_i, D_{it})$ is not too serious and so there is not much benefit to using SC techniques. 

\subsection{Set-up of Simulation}
In order to generate data from the interactive fixed effects model under Assumption \ref{ass_th}, we need to make some further assumptions on how the stochastic parameter are determined. For brevity, most details are left in Appendix \ref{subsec_fix_param}. Here we focus on the trickiest parameter to generate: the correlation between treatment assignment $D_{it}$ and unobservables, using the  approach of \citet[p.14]{ark20}. Given the set-up in Section \ref{subsec_setup}, the treatment can only be 1 in the last $T_0$ periods. The treatment indicator is drawn from a Bernoulli distribution with mean $\pi_i$. We introduce a dependence of treatment assignment $D_{it}$ on $\boldsymbol{\mu_i}$ via $\pi_i$. While $\pi_i$ is independent from the observables  $\boldsymbol{x}_i$, it is related to the factor loadings $\boldsymbol{\mu_i}$ via a hierarchical model:
\begin{gather}
    D_{iT}| \boldsymbol{\mu_i}, v_i \sim Ber(\pi_i) \quad \quad  \quad \pi_i = \frac{\exp( \boldsymbol{\mu_i \phi}  + \varepsilon_i)}{1+\exp( \boldsymbol{\mu_i \phi} + \varepsilon_i)}
    \label{tr_not_at_rand}
\end{gather}
where $\varepsilon_i$ are some iid $N(0,1)$ shocks. If $\boldsymbol{\phi} \neq 0$, then $D_{it}$ does not occur at random, as it is correlated with the unobserved $\boldsymbol{\mu_i}$. Our choice of $\boldsymbol{\phi}$ allow us to control the correlation.  This matters because studies assuming random treatment assignment tend to overestimate how well causal inference methods such as DiD estimate the true treatment effect \citep{ark20}, e.g., in terms of estimation error. Moreover, since in practice with observational data treatment is unlikely to be assigned completely at random, it seems reasonable to assume that $D_{it}$ is not independent of  unobserved characteristics, as reflected in $\boldsymbol{\mu_i}$.

After generating the data, the simulation will compare four estimators: the feasible difference-in-difference (fDiD), iDiD, CSC and PSC of Abadie and L`Hour \citeyearpar{lho20}. The motivation for including PSC is that it is the only other estimator known to us that addresses exactly the same set-up as ours by developing a SC estimator. Appendix \ref{subsec_psc_det} provides the formal details on how PSC calculates counterfactuals. Lastly, \textbf{Algorithm 1} summarises the flow of the simulation. 

\RestyleAlgo{boxruled}
\begin{algorithm}[h]
\DontPrintSemicolon
\SetAlgoLined
\For{ $j$ in $1:1000$ }{
    Generate each $y_{it}$(j) in the $(N \times T)$ matrix $\boldsymbol{\Theta}$(j) and $K$ covariates $\boldsymbol{x}_i$(j) \;
    \uIf{treat at random(j) $==$ TRUE}{
        Set $\boldsymbol{\phi} = 0$  \;
    }\Else{
        Set $\boldsymbol{\phi} \neq 0$  \;
    }
    Divide $\boldsymbol{\Theta}$(j) into pre-treatment data and post-treatment data \;
    Fit \textit{CSC}, \textit{feasible DiD}, \textit{infeasible DiD} and \textit{PSC} in pre-treatment dataset \;
    Estimate performance metrics in post-treatment period \;
}
\caption{Pseudo-code for the flow of our simulation}
\end{algorithm} 

Each of the 1000 replications begins by simulating the matrix $\boldsymbol{\Theta}$. Next, the treatment is assigned either at random or not at random via the choice of $\boldsymbol{\phi}$. After the data is divided into pre-treatment and post-treatment datasets, we implement the four estimators and estimate the ATT $\hat{\tau}$ for each of them. Afterwards, we calculate two performance metrics: 1) the estimation error of the estimated ATT $\hat{\tau}$, i.e., the difference $\hat{\tau}-\tau$; and 2) RMSE between the synthetic $\hat{y}_{it}(0)$ and the true unobserved potential outcome $y_{it}(0)$ for the treatment group, i.e. $\forall i \in \{n_0 + 1, n_0 + 2, \dots, n_0 + n_1 \}$ in the post-treatment period $\forall t \in \{T_0+1, \dots, T\}$. Note that this will not be possible with actual data, as we do not observe the two potential outcomes simultaneously due to the Fundamental Problem of Causal Inference, nor do we know the true ATT. Nevertheless, if indeed the matrix completion framework for matrix $\boldsymbol{\Theta}$ discussed in Section \ref{subsec_setup} is appropriate, then we would expect the results for the point estimate of ATT and the RMSE of the SCs to be similar.

\subsection{Results}
The results for the average estimation error between the true $\tau$ and the estimated $\hat{\tau}$ are presented in Table \ref{tab_sim_att}. Since we report results for the same quantity of interest as in Section \ref{sec_the}, that is $avr. \ est. \ error = \frac{\sum_{r=1}^{1000} (\hat{\tau}_r - \tau) }{1000} $, we can interpret Table \ref{tab_sim_att} in light of the theoretical properties we derived.\footnote{Except that here we calculate the \textit{average} estimation error over 1000 simulated datasets  rather than the \textit{single} estimation error for one dataset.} 

Table \ref{tab_sim_att} begins with a baseline specification where we set $N=100$ and $T=5$ for the dimensions of the matrix $\boldsymbol{\Theta}$, a three-factor model $F=3$, treatment-not-at-random, i.e.\ $Cov(D_{it}, \boldsymbol{\mu}_{i}) \neq 0$, and the average probability of treatment is set at $0.15$, i.e.\ $E[\pi_i]=0.15$ in (\ref{tr_not_at_rand}), so that we expect with $N=100$ on average $n_1 = 15$ treated observations and $n_0=85$ donors. In Table \ref{tab_sim_att}, we progressively change the values of  the five parameters (random treatment allocation, $N$, $\pi_i$, $T$, and $F$) relative to the baseline specification. 

Several intriguing things can be seen from the results in Table \ref{tab_sim_att}.  Firstly, across all  specifications iDiD dominates the other estimators, which is unsurprising given that it is unbiased and consistent. However, across the feasible estimators, CSC tends to outperform fDiD and PSC with PSC usually performing a bit worse than CSC. In particular, an (average) estimation error of $-0.11$ in the baseline case\footnote{For reasons of time, I was unable to rerun the same analysis using the absolute value of the estimation error which would have made the results from the Simulation and the Theoretical Section exactly comparable.} for CSC means that with the true $\tau$ being 1, CSC estimates $\hat{\tau}^{CSC}=0.89$ on average across 1000 simulations. So, CSC is doing quite well and even with a relatiely small sample of just 100 observations it performs similarly to iDiD. Secondly, let us consider what happens once we make treatment less correlated with the unobserved factor loadings. As suggested by Proposition \ref{prop_DiD_est_err}, fDiD indeed does better under random treatment assignment than under non-random treatment assignment and outperforms CSC. However, even under a week correlation between $D_{it}$ and $\boldsymbol{\mu}_i$, fDiD still does not do too well relative to CSC. If we suspect the treatment is allocated even slightly not at random, CSC should be used, confirming our key theoretical result. 

\begin{table}
    \centering
    \caption{\textit{Average estimation error for ATT ($\frac{\sum_{r=1}^{1000} (\hat{\tau}_r - \tau)}{1000}$) under different DGPs} }
    \label{tab_sim_att}
    \begin{threeparttable}
    \begingroup
    \setlength\extrarowheight{1.2pt} 
    \[
    \begin{array}{lccc|c}
    \\[-1.8ex]\hline 
    \hline 
    \mathrm{Change} & \mathrm{CSC} & \mathrm{fDiD} & \mathrm{PSC} & \mathrm{iDiD} \\ 
    \hline 
    \mathrm{Baseline^*} & \textbf{-0.11} & -0.94 & -0.25 & \textbf{-0.01} \\ 
    \mathrm{Less \ Correlation} & \textbf{-0.05} & -1.05 & -0.22 & \textbf{-0.02} \\ 
    \mathrm{Random \ Assignment \ (\phi = 0)} & 0.05 & \textbf{-0.01} & -0.03 & \textbf{-0.01} \\ 
    \hline 
   \mathrm{ Increase} \ N = 150 & \textbf{-0.04} & -0.95 & -0.24 & \textbf{-0.01} \\ 
    \mathrm{ Increase} \ N = 200 & \textbf{-0.03} & -0.96 & -0.18 & \textbf{-0.01} \\ 
    \hline 
    \mathrm{ Increase} \ E[\pi_i] = 0.25 & \textbf{-0.07 }& -0.86 & -0.23 &  \textbf{0.00} \\ 
    \mathrm{ Increase} \ E[\pi_i] = 0.40 & \textbf{-0.12} & -0.79 & -0.36 & \textbf{-0.01} \\ 
    \hline 
    \mathrm{ Decrease} \ T = 4 &  \textbf{2.43} &  3.53 &  4.71 &  \textbf{0.00} \\ 
    \mathrm{ Increase} \ T = 8 &  \textbf{0.09} &  0.50 &  0.31 & \textbf{-0.01} \\ 
    \hline 
    \mathrm{ Decrease} \ F = 2 & \textbf{-0.31} & -2.18 & -0.78 & \textbf{-0.01} \\ 
    \mathrm{ Increase} \ F = 4 & \textbf{-0.05} & -0.94 & -0.30 & \textbf{-0.01} \\ 
    \hline 
    \end{array}
    \]
    \endgroup
    \begin{tablenotes}
      \small
      \item $^{*}$Baseline refers to specifying $\{ \mathrm{Non-Random \ Assign.}, N=100, T = 6, F = 3, E[\pi_i] = 0.15 \}$. All tables were created with \texttt{R} package \texttt{stargazer} \citep{hla18}
    \end{tablenotes}
    \end{threeparttable}
\end{table}

Thirdly, increasing $N$ has little effect on the estimation errors for the four estimators. Perhaps we need to experiment with bigger values of $N$ to see any effects or consider changing the share of treated observations, as done in the next panel of the table. However, increasing the probability of treatment $\pi_i$ also has no substantial effect on our results. This is seemingly surprising in light of Proposition \ref{prop_CSC_est_err} which implies that the bias of CSC is positively related to the number of donors $n_0$. We may expect the estimation error of CSC to fall, as we decrease the proportion of donors, holding total $N$ fixed.  Nevertheless, we can explain this by the fact that what we found in Proposition \ref{prop_CSC_est_err} is an upper bound. So, perhaps, changing the number of donors $n_0$ is not such a bad thing for CSC.

Fourthly, the results for changing $T$ in the case of CSC are as expected: it improves its performance, even though we mostly gain from increasing $T=4$ to $T=6$. The next set of results from the simulation in Table \ref{tab_sim_out} paints a clearer picture of what happens if we increase $T$. Lastly, the results for changing the number of factors $F$ are puzzling. The theoretical properties of fDiD and CSC in Section \ref{sec_the} suggest that the estimation error of both is increasing in $F$ but here we do not find any significant differences when we change $F$.

While the point estimate for the estimation error of $\hat{\tau}^{CSC}$ looks close to 0, its variance might be very big. The reason why the variance is not accounted for in the estimation error is a consequence of our decision to work with the unadjusted difference between the true and estimated ATT. For this reason, Appendix \ref{subsec_rmse_att} reports the results for the $RMSE$ between $\tau$ and $\hat{\tau}$ rather than the pure estimation error. Specifically, RMSE is $\frac{\sum_{r=1}^{1000}(\hat{\tau} - \tau)^2 }{1000}$ and the estimation error is just $\frac{\sum_{r=1}^{1000}(\hat{\tau} - \tau) }{1000}$.  Overall, the results appear largely unchanged.

Besides studying our estimates of the ATT, we may also be interested in how well the SCs resemble the unobserved potential outcomes for the treated units without treatment in the post-treatment period (Table \ref{tab_sim_out}). As in the previous table, we begin with the baseline specification $\{ \mathrm{Non-Random \ Assign.}, N=100, T = 6, F = 3, E[\pi_i] = 0.15 \}$  and progressively change some of the key parameters. The results are very similar to those for the estimation error. For example, while iDiD creates the best possible SC, the best performing estimator among those that are feasible is CSC in all cases when treatment does not occur at random. This suggests that the framework in \ref{subsec_setup} is meaningful: we can indeed turn the causal inference problem into a prediction problem. One other important result from Table \ref{tab_sim_out} is the fact that we can observe more clearly what happens as we increase $T$: all estimators improve significantly and start to approximate the RMSE for iDiD. With $T=8$, CSC and fDiD perform very similarly. While in practice we may often observe shorter time periods, this result confirms that SC methods work well with long panels \citep{aba10}.

\begin{table}[hbtp!]
    \centering
    \caption{ \textit{RMSE between $\hat{y}^0_{it}$ and true $y^0_{it}$ under different DGPs across 1000 simulations}  }
    \label{tab_sim_out}
    \begin{threeparttable}
    \begingroup
    \setlength\extrarowheight{1.2pt} 
    \[
    \begin{array}{lccc|c}
   \\[-1.8ex]\hline 
    \hline 
    \mathrm{ Change} & \mathrm{CSC} & \mathrm{fDiD} & \mathrm{PSC} & \mathrm{iDiD}  \\ 
    \hline 
    \mathrm{ Baseline}^{*} & \textbf{1.80} & 2.28 & 2.51 & \textbf{0.92} \\ 
    \mathrm{ Weak \ Correlation} & \textbf{1.75} & 2.24 & 2.41 & \textbf{0.92} \\ 
     \mathrm{ Random \ Assignment} & \textbf{1.75} & 2.03 & 2.35 & \textbf{0.92} \\ 
     \hline 
    \mathrm{ Increase} \ N = 150 & \textbf{1.80} & 2.28 & 2.36 & \textbf{0.92} \\ 
    \mathrm{ Increase} \ N = 200 & \textbf{1.80} & 2.29 & 2.30 & \textbf{0.92} \\ 
     \hline 
    \mathrm{ Increase} \ E[\pi_i] = 0.25  & \textbf{1.84} & 2.26 & 2.63 & \textbf{0.92} \\ 
    \mathrm{ Increase} \ E[\pi_i] = 0.40 & \textbf{1.90} & 2.24 & 2.75 & \textbf{0.92} \\ 
    \hline
    \mathrm{ Decrease} \ T = 4 & \textbf{5.40} & 5.93 & 9.86 & \textbf{0.88} \\ 
    \mathrm{ Increase} \ T = 8 & \textbf{1.33} & \textbf{1.34} & 2.12 & \textbf{0.94} \\ 
     \hline 
    \mathrm{ Decrease} \ F = 2 & \textbf{1.86} & 2.97 & 2.71 & \textbf{0.92} \\ 
    \mathrm{ Increase} \ F = 4 & \textbf{1.95} & 2.43 & 2.98 & \textbf{0.92} \\
     \hline 
    \end{array}
    \]
    \endgroup
    \begin{tablenotes}
      \small
      \item $^{*}$Baseline refers to specifying $\{ \mathrm{Non-Random \ Assign.}, N=100, T = 6, F = 3, E[\pi_i] = 0.15 \}$
    \end{tablenotes}
    \end{threeparttable}
\end{table}

\section{Empirical Application} \label{sec_emp}
This section illustrates how CSC can be used in practice: we study the effect of immigration on natives' labour market outcomes. We  hope to gain some new insight on this question by revisiting the Mariel Boatlift \citep{car90} and analysing it with a new dataset: PSID.\footnote{Given the small sample size of PSID, we should be careful, when interpreting the results, as their external validity might be questioned.}

\subsection{Mariel Boatlift}
The textbook model of labour markets with perfect competition \citep{car14} postulates that immigration increases labour supply and brings down the wages of natives. This can be seen in Figure \ref{fig_imm_shock}: wages fell from $w$ to $w'$ following the shift of Labour Supply ($LS$). However, testing the hypothesis for a negative effect on wages raises many econometric challenges. Suppose we specify the model:
\vspace{-8pt}
\begin{gather}
    y_{it} = \eta + \tau m_{it} + \boldsymbol{\beta} \boldsymbol{x}_{it} + u_{it}
    \label{imm_reg}
\end{gather}
where $i$ indexes a local labour market,\footnote{We could base $i$ on geography or national education-experience cells, for example.} $y_{it}$ are log wages (although it could be another labour market outcome), $\eta$ is an intercept, $m_{it}$ is the number of immigrants, $\boldsymbol{x}_{it}$ are time varying characteristics and $u_{it}$ are error terms. The parameter of interest is $\tau$ which is the effect of immigration on wages and theory predicts that it should be negative. Arguably the biggest econometric challenge is the fact that immigrants self-select into areas with greater economic opportunities, implying that $Corr(m_{it}, u_{it}) >0$ which makes $m_{it}$ endogenous and so the OLS estimate $\hat{\tau}$ inconsistent.\footnote{Other challenges stem from how we define local labour markets. If we define $i$ to be geographical local labour markets such as counties, we can be concerned about displacement effects \citep{bor17}. On the other hand, if we use national education-experience cells, we should still think about what instrument to use to tackle the selection effects \citep{dus16}.}

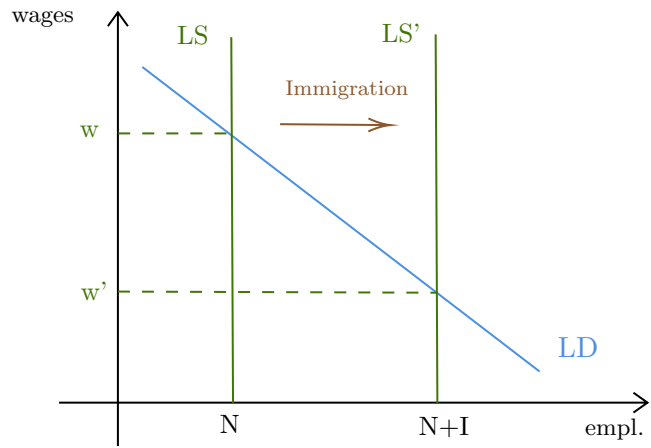
\begin{figure}
\centering
\tikzset{every picture/.style={line width=0.75pt}} 
\begin{tikzpicture}[x=0.75pt,y=0.75pt,yscale=-1,xscale=1]

\draw  (169.28,215.7) -- (463,215.7)(198.65,19.34) -- (198.65,237.51) (456,210.7) -- (463,215.7) -- (456,220.7) (193.65,26.34) -- (198.65,19.34) -- (203.65,26.34)  ;
\draw [color={rgb, 255:red, 74; green, 144; blue, 226 }  ,draw opacity=1 ]   (210.83,46.85) -- (409.13,200.04) ;
\draw [color={rgb, 255:red, 65; green, 117; blue, 5 }  ,draw opacity=1 ]   (256.11,215.83) -- (255.25,31.88) ;
\draw [color={rgb, 255:red, 65; green, 117; blue, 5 }  ,draw opacity=1 ]   (358.12,215.83) -- (357.26,30.49) ;
\draw [color={rgb, 255:red, 139; green, 87; blue, 42 }  ,draw opacity=1 ]   (279.6,75.64) -- (332.62,76.08) ;
\draw [shift={(334.62,76.1)}, rotate = 180.48] [color={rgb, 255:red, 139; green, 87; blue, 42 }  ,draw opacity=1 ][line width=0.75]    (10.93,-3.29) .. controls (6.95,-1.4) and (3.31,-0.3) .. (0,0) .. controls (3.31,0.3) and (6.95,1.4) .. (10.93,3.29)   ;
\draw [color={rgb, 255:red, 65; green, 117; blue, 5 }  ,draw opacity=1 ] [dash pattern={on 4.5pt off 4.5pt}]  (198.8,80.28) -- (254.96,80.28) ;
\draw [color={rgb, 255:red, 65; green, 117; blue, 5 }  ,draw opacity=1 ] [dash pattern={on 4.5pt off 4.5pt}]  (198.65,159.83) -- (358.41,160.41) ;

\draw (416.82,181.21) node [anchor=north west][inner sep=0.75pt]  [color={rgb, 255:red, 74; green, 144; blue, 226 }  ,opacity=1 ] [align=left] {LD};
\draw (226.69,24.54) node [anchor=north west][inner sep=0.75pt]  [font=\small,color={rgb, 255:red, 65; green, 117; blue, 5 }  ,opacity=1 ] [align=left] {LS};
\draw (328.85,23.61) node [anchor=north west][inner sep=0.75pt]  [font=\small,color={rgb, 255:red, 65; green, 117; blue, 5 }  ,opacity=1 ] [align=left] {LS'};
\draw (280.76,52.4) node [anchor=north west][inner sep=0.75pt]   [align=left] {{\scriptsize \textcolor[rgb]{0.55,0.34,0.16}{Immigration}}};
\draw (430.45, 222.03) node [anchor=north west][inner sep=0.75pt]   [align=left] {{\footnotesize empl.}};
\draw (144.05,17.64) node [anchor=north west][inner sep=0.75pt]  [font=\footnotesize] [align=left] {wages};
\draw (178.18,74.67) node [anchor=north west][inner sep=0.75pt]  [font=\small,color={rgb, 255:red, 65; green, 117; blue, 5 }  ,opacity=1 ] [align=left] {w};
\draw (178.61,153.12) node [anchor=north west][inner sep=0.75pt]  [font=\small,color={rgb, 255:red, 65; green, 117; blue, 5 }  ,opacity=1 ] [align=left] {w'};
\draw (248.1,220.43) node [anchor=north west][inner sep=0.75pt]  [font=\small] [align=left] {N};
\draw (347.34,221.59) node [anchor=north west][inner sep=0.75pt]  [font=\small] [align=left] {N+I};

   ;
\end{tikzpicture}
\caption{ \footnotesize{The effect of immigration in the classical labour market model with inelastic labour supply and flexible wages. \textit{Note:} We have inelastic LS, due to perfect competition between firms: if a firm decreases its wages, then all of its workers can leave and immediately find a job, meaning that wage elasticity of labour supply is infinite \citep{man06}.} }
\label{fig_imm_shock}
\end{figure}

Card's Mariel Boatlift study \citeyearpar{car90} overcomes this identification challenge by exploiting the effect of a large exogenous shock to labour supply. On 20th April 1980, Fidel Castro announced that Cubans would be allowed to leave the country and migrate to the US from the Mariel port. As a result, almost 125,000 Cubans fled their home and arrived in the US before the policy was reversed in late September. Many of them settled in Florida (and Miami in particular), because previous Cuban immigrants have established themselves there and not because of the economic opportunities there. Thus, the Mariel Boatlift provides an exogenous increase of 8\% to labour supply in Miami and so an excellent opportunity to test the model in Figure \ref{fig_imm_shock}. 

Card uses DiD to estimate the effects of immigration. He explores how wages and unemployment of natives evolved in Miami compared to a control group of cities in the pre-treatment period and post-treatment periods. He finds that while wages in Miami fell for white, black and Hispanic native workers between 1979 and 1983, they also fell in the control group by a similar amount. Thus, he concludes that there was no negative effect of the increase in labour supply to wages, contrary to the prediction of the model in Figure \ref{fig_imm_shock}.

Over the past five years, there has been a revival\footnote{For an excellent review (much clearer than mine), click \href{https://twitter.com/EGirlMonetarism/status/1376683452662185985?fbclid=IwAR103_8CTVBRA7JeZL0q6TxsvgxovGbfEZIDn2vCWLgjEuqzg-ZkNM8ActA}{here}. } in interest in evaluating the Mariel Boatlift \citep{bor17, per18}. These recent papers have found markedly different results: \citet{bor17} claims  that the increase in immigration had a negative effect on natives' wages whereas \citet{per18} find no effects. The recent literature has identified several important limitations of Card's study: 
\begin{enumerate}
    \item Firstly, the choice of control group consisting of Atalanta, Los Angeles, Houston and Tampa-St.\ Peterbourg  is not well justified \citep{bor17}. Since the publication of the paper, better techniques for constructing counterfactuals such as SC have been developed.
    \item Secondly, and perhaps most fundamentally, the Miami Current Population Survey (CPS) datasets, versions of which are used in \citet{car90}, \citet{bor17}, and \citet{per18}, have undergone significant compositional changes that were not related to the Mariel Boatlift \citep{cle19}. In particular, the number of low-skilled black workers in Miami with very low wages surveyed have been increased substantially in the 1980s, thereby giving a greater survey weight to this group in early 1980s relative to late 1970s. This renders CPS inappropriate for comparison across different years, especially given that the treatment occurs at 1980, given the important changes in survey weights.
    \item Thirdly, since the CPS data is a repeated cross-section, it has to be aggregate on a city-level in order to find the trend in wages in Miami. This can mask important individual-level differences in labour supply and wages that are lost when aggregating. 
\end{enumerate}

\subsection{PSID}
Based on Problem 2., it seems that using CPS data should be avoided and this is the reason why we revisit the Mariel Boatlift using PSID. When we combine PSID with CSC, we can also automate the selection of a control group (thus, tackling Problem 1.) and exploit PSID's panel structure to capture individual-level effects (thus, tackling Problem 3.). 

The PSID is the longest running panel survey of households, stemming from late 1960s to the present day. It contains data on roughly 5000 families, living in the US. PSID's composition did not change considerably which allows us to tackle Problem 2.\ above. Moreover, the fact that it is panel study opens the door to using CSC which allows us to construct good counterfacturals for every treated worker and which tackles Problem 1. The PSID sample that we use consists of 11 waves (1974-1984) and is restricted to heads of households who are male. The justification for these restrictions and the variable selected can be found in Appendix \ref{psid_var}. Most importantly, we focus on two outcomes of interest: wages and hours worked. While the classical labour markets model in Figure \ref{fig_imm_shock} does not allow any labour supply effects,\footnote{as the LS is assumed to be inelastic} some labour supply effects can be expected \citep{dus16}, given the strong evidence for sticky wages. 

Unfortunately, using PSID comes at a cost: our final sample contains only 42 treated workers. However, while our results may not be generalisable for the overall effect of the Mariel Boatlift, we can still interpret them as yielding useful insights for our particular sample and worry about their external validity separately. Moreover, the sample sizes used in \citet{bor17} and \citet{per18} also contain only respectively around 20 and 65 people per year, suggesting that small sample sizes is a persistent issue when estimating the effects of immigration.\footnote{However, note that they define low-skilled workers as workers without a high school degree whereas we include workers with a high school degree into low-skilled workers, implying that their sample sizes could have been greater if they included the latter group. }

Another issue with PSID is that it does not include metropolitan area indicators but only state indicators in its public release. So, we use as treatment group not residents of Miami but those of the state of Florida. This is problematic, as the impact of the Mariel Boatlift was most strongly felt in Miami. Nevertheless, there is some evidence that some Mariel Cubans also settled in other areas in Florida such as Tempa and West Palm Beach\footnote{At the time, Palm Beach was not considered a part of Miami Metropolitan Area \citep[p.30]{usdc82}. Since this is the definition that Card's data uses \citeyearpar[footnote]{car90}, he did not include Mariel immigrants located there.}, as shown by \citet{sko01}. Previous Cuban immigrants (e.g. Golden Exiles in 1960-64 and Freedom Flight refugees in 1965-1974) were mostly white and had higher socio-economic status than immigrants in 1980 \citep[p.457]{mcc85}. In contrast, Mariel Cubans were more likely to be non-white and more generally had different locational preferences, regarding where to settle \citep{sko01}. Quantitatively, Miami did host the majority of Cuban immigrants, but there were spillovers of migrants to other cities in Florida. Specifically, 1990 US Census data shows that more than 10\% of Mariel Cubans based in Florida were not located in Miami but in the early 1980s the number was higher due to  within-Florida migration to Miami in the 1980s \citep[p.464]{sko01}. Lastly, if we are able to find a negative effect for low-skilled workers in Florida, then it is likely that this effect would be even stronger for the same workers in Miami. Overall, then, the lack of metropolitan area identifiers does not invalidate our analysis. 

\subsection{PSC vs CSC}\label{subsec_psc_vs_CSC}

The first set of results with PSID that we present compares CSC and PSC using a cross-validation approach \citep{aba15}. We begin by dividing the data into post-treatment data (1980-1984) and pre-treatment data (1975-1979) which is further divided into (pre-treatment) training data and (pre-treatment) testing data. In all specifications, the testing period consists of the last two pre-treatment years, namely 1978 and 1979, but we vary the length of the training period. Then, PSC and CSC are fit on the training data and we evaluate their performance by calculating the Root Mean Squared Error (RMSE) between the predicted labour market outcome (either wages or hours worked) and the true outcomes. Lastly, we calculate the individual treatment effects and the overall ATT.\footnote{When fitting the models, we use the covariates for occupation, industry, education, being white, being married and having being ill a lot that were coded as discrete variables, so that CSC can accommodate them, as discussed by the restriction in Section \ref{subsec_est}. }

The results from this exercise are reported in Table \ref{tab_comp}. Overall, CSC performs slightly better than PSC for $T_{train}=3$ and $T_{train}=2$, as it predicts labour supply more accurately in the training period (1979-1980). In light of the simulation results in Section \ref{sec_sim}, this could reflect the fact treatment assignment is not too strongly correlated with unobserved components.\footnote{However, note that \textit{some} correlation is expected, given that Florida is in the bottom quantile in terms of median wage and that we are interested primarily in the effect on wages of the low-skilled workers. See the discussion in the end of Section \ref{sec_DiD_bou} for more details} 

\begin{table}[!bhtp] \centering 
  \caption{RMSE in testing period for different lengths of training period }  \label{tab_comp} 
  \begin{threeparttable} 
  \[ 
  \begin{array}{@{\extracolsep{3pt}} l c cc | cc} 
    \\ \hline \hline 
          &  & \multicolumn{2}{c}{Lab. \ Supply} & \multicolumn{2}{c}{Log \ wages^*} \\
          \cline{2-6} 
          & Year & CSC & PSC & CSC & PSC \\ 
        \hline 
        T_{train}=4 & 1978 & \textbf{569.92} & 666.43 & 0.63 & \textbf{0.57} \\ 
         & 1979 & \textbf{592.62} & 670.27 & 0.86 & \textbf{0.76} \\ 
         & Total & \textbf{581.38} & 668.35 & 0.75 & \textbf{0.67} \\ 
        \hline 
        T_{train}=3 & 1978 & \textbf{616.17} & 663.19 &\textbf{0.56} & \textbf{0.57} \\ 
         & 1979 & \textbf{638.05} & 687.91 & 0.77 & \textbf{0.73} \\ 
         & Total & \textbf{627.21} & 675.67 & \textbf{0.67} & \textbf{0.65} \\ 
        \hline 
        T_{train}=2 & 1978 & \textbf{637.04} & 655.14 & \textbf{0.49} & 0.60 \\ 
         & 1979 & \textbf{677.68} & 725.55 & \textbf{0.69} & \textbf{0.68} \\ 
         & Total & \textbf{657.67} & 691.24 & \textbf{0.60} & 0.64 \\
        \hline 
        T_{train}=1 & 1978 & 739.45 & \textbf{671.06} & \textbf{0.40} & 0.60 \\ 
         & 1979 & 851.64 & \textbf{687.11} & \textbf{0.59} & 0.68 \\ 
         & Total & 797.52 & \textbf{679.13} & \textbf{0.50} & 0.64 \\ 
        \hline \\
    \end{array} 
\]
\begin{tablenotes}
      \small
      \item $^*$ We applied transformation $\log(1+w_{it})$ to wages because some individuals had $w_{it} < 1$ so  $\log(w_{it})$ yields a non-positive value. In future work, we hope to use instead the inverse hyperbolic sine transformation.
    \end{tablenotes}
\end{threeparttable}
\end{table}

Interestingly, as we increase $T_{train}$, CSC becomes better at predicting total hours worked but worse at predicting (log) wages. This could perhaps be explained by the fact that wages are more volatile than hours worked which exhibits stronger autocorrelation. Note that our best prediction for log wages involves using $T_{train}=1$ and for hours worked using $T_{train}=4$ and that both use CSC rather than PSC. Moreover, this result points to the fact that using more lags of the outcome variable is helpful only if it is informative for future values. Note, however, that this is not the case for PSC where the results seem largely invariant to the number of lags we include. Overall, we can conclude that CSC definitely does not do worse than PSC and our best performers in terms of predicting wages and hours worked both utilise CSC but use a different number of lags.

\subsection{Heterogeneous Treatment Effects}\label{subsec_hte}
In Figure \ref{fig_haat}, we present the treatment effects of the Mariel Boatlift on labour supply and log wages for low-skilled and high-skilled native workers. To get the individual treatment effects, we use the SCs created by CSC in the previous section with $T_{train} = 4$ for labour supply and with $T_{train} = 1$ for wages.\footnote{These choices of $T_{train}$ seem the best option based on the results in Table \ref{tab_comp}} Results for other choices of $T_{train}$ are provided in the Appendix \ref{subsec_htt_diff_spec} and are quite similar.

Figure \ref{fig_haat} plots the evolution of the 95\% Confidence Interval for the heterogeneous treatment effects on low-skilled workers and high-skilled workers. Following \citet{ace11}, we define low-skilled workers as those individuals who do not have a college degree.  We can see that there are no effects on high-skilled workers' labour market outcomes. This is not surprising, given that many Mariels were classified as being low-skilled and so were not competing with high-skilled natives for jobs \citep{sko01}. The lack of labour supply effects also implies that the $LS$ curve in Figure \ref{fig_imm_shock} is probably quite inelastic, albeit not perfectly inelastic, for both type of workers. 

However, there is some decrease in wages for low-skilled workers. While this interpretation makes sense in light of the competitive labour market model, our sample sizes are very small. We observe just 25 low-skilled and 17 high-skilled workers. At most, what we can say is that for the selected sample we observe a negative effect on wages of low-skilled workers. If PSID is to be used for future evaluations of the Mariel Boatlift, the researchers should consider how to increase sample size, e.g., include women. 

\begin{figure}
    \centering
    \includegraphics{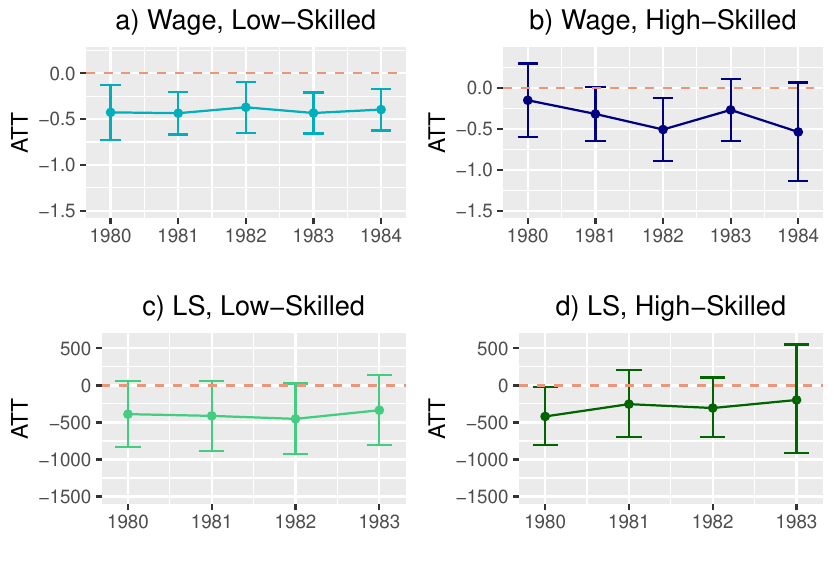}
    \caption{95\% CI for Mariel Boatlift's effect on low-skilled and high-skilled workers.}
    \label{fig_haat}
\end{figure}

Lastly, the confidence intervals in Figure \ref{fig_haat} could in reality be wider. The reason is that the treatment effects do not incorporate the uncertainty connected with the weights of the SC. In general, when doing causal inference with any SC method, there are (at least) two sources of uncertainty: from the weights of the SCs and from the estimates of the individual treatment effects $\tau_i$. The confidence intervals above only account for the second source of uncertainty and not for the uncertainty connected to the weights themselves. So, overall inference is likely to be unreliable in this case. At present, accounting for both sources of uncertainty is an open question.

\section{Conclusion}\label{sec_concl}
The main purpose of this paper was to develop an estimator from the family of SC for the case of many treated units observed for a short time period. CSC creates synthetic counterfactuals that are correlated across treated individuals similar in terms of observables. We show that CSC has good theoretical properties when treatment assignment is correlated with unobservable characteristics. In such cases, it even dominates the DiD estimator which is the most widely used technique in the setting we consider. We confirmed this theoretical result in our simulation: making treatment correlated with unobservables has little effect for estimates of treatment effects from CSC but estimates from DiD become much less precise. Moreover, we discussed two applications taken from economic geography and labour economics, in which CSC should be preferred over other causal inference techniques.

With that said, CSC currently has some important limitations. At present, it only allows for time-invariant discrete characteristics. This key limitation stems partly from the fact that we have constrained the (correlated) random coefficients model of the weights to allow only for time-invariant predictors. What happens if we relax this restriction? Theoretically, the optimisation problem solved by CSC to find the weights can accommodate time-varying discrete covariates in the weights (e.g., changing occupation of treated workers).\footnote{See restriction (\ref{anders_res}). With time-varying covariates we still need for it to hold. However, it does not preclude the possibility for $x_{it}^k$ to be time-varying.} This is a fascinating possibility, as it will allow the weights on donors to be evolving over time, according to the time-varying covariates of a given treated individual. For instance, if an individual loses her job, we can use different donors for her SC before and after this event. 

On the other hand, the conventional $(SC \ Constraints)$ that the weights sum up to 1 and are non-negative is sometimes useful, as we can interpret the weights as probabilities. However, our results from the simulation show good support for the framework introduced in Section \ref{sec_est_cons} that turns the causal inference problem into a prediction problem. So, perhaps, if prediction is the end game, this restriction can be relaxed. Maybe we can impose other constraints that find sparse solutions but do not restrict the set of weights as much, e.g., Lasso.

Further limitations include the fact that we have not used donors' covariates. In particular, when specifying the correlated random coefficients structure of the weights, we only used treated individuals' covariates. This is not entirely satisfactory, as we are losing an important source of information. In any case, there are many open questions related to CSC and other SC estimators that this paper could not address. However, as noted by Guido Imbens in the Sargan lecture last month \citeyearpar{imb21}, SC-like methods are growing into an important part of the toolkit of the empirical researchers. Therefore, pursuing such open questions is a fruitful area of research. 

\newpage 
\bibliography{main} 

\newpage 
\appendix
\addcontentsline{toc}{section}{Appendices}
\section*{Appendix}

\section{Supplementary results and examples} \label{ap_oth_res}
\addtocontents{toc}{\protect\setcounter{tocdepth}{1}}

\subsection{ Matching covariates versus matching outcomes}\label{ap_cov_vs_out}
As discussed in the main text, the model specified by SC methods for the weights is:
$$ y_{it} = \sum_{i=1}^{n_0} \hat{w}_j y_{jt} + e_{jt}  $$
which looks similar to a linear regression model, except for the constraint that the weights sum up to 1 and are greater than 0. Based on this intuition, there have emerged two contrasting approaches in the literature emphasising the matching of SCs on different things: 1) matching primarily on \textit{covariates} advocated by Alberto Abadie and coauthors \citep{aba10, lho20, aba20} versus 2) matching primarily on \textit{outcomes} advocated by \citet{ ath19, dou18, ben20}. The main reason why this distinction matters is that it leaves open the question of how to use covariates in SC, as illustrated by \citep{bot19}.

To understand the differences, suppose that $n_1 = 1$. It is useful to consider a linear transformation of the stacked row vector of outcomes and covariates $(\boldsymbol{y}^{pre}_{n_0+1}, \boldsymbol{x}_{n_0+1})'$  via the $(H \times (T_0 + k))$ matrix $\boldsymbol{M}$, given by the multiplication $\boldsymbol{M}(\boldsymbol{y}^{pre}_{n_0+1}, \boldsymbol{x}_{n_0+1})'$. The two approaches differ is in the choice of $\boldsymbol{M}$ and its dimension $H$. As common in the literature \citep[e.g.][]{fer19}, 
suppose that $\Delta$ denotes the $n_0$-dimensional standard simplex $\Delta \equiv \{ (w_1, \dots, w_{n_0})' \in \R^{n_0}| \sum_{j=1}^{n_0} w_j = 1 \ and \ w_j \geq 0 \}$. Simplifying slightly,\footnote{We are simplifying because in practice Alberto Abadie \citeyearpar{aba20} advocates that the weights $\boldsymbol{V}$ and $\boldsymbol{W}$ should be estimated in two optimisation problems.} both approaches can be reduced to solving a version of the constrained quadratic optimisation problem in matrix form:
\begin{gather}
    \min_{\boldsymbol{w} \in \Delta} \left( \boldsymbol{M} 
    \begin{pmatrix}
    \boldsymbol{y}^{pre}_{n_0+1} \\
    \boldsymbol{x}_{n_0+1}
    \end{pmatrix} - 
    \boldsymbol{M}  \begin{pmatrix}
    (\boldsymbol{Y}^{pre}_{n_0})' \\
    \boldsymbol{X}_{n_0}
    \end{pmatrix} \boldsymbol{w}\right)' \boldsymbol{V}
     \left(\boldsymbol{M} \begin{pmatrix}
    \boldsymbol{y}^{pre}_{n_0+1} \\
    \boldsymbol{x}_{n_0+1}
    \end{pmatrix} - \boldsymbol{M}  \begin{pmatrix}
    (\boldsymbol{Y}^{pre}_{n_0})' \\
    \boldsymbol{X}_{n_0}
    \end{pmatrix} \boldsymbol{w}\right)
    \label{gen_scm_pro}
\end{gather}
where $\boldsymbol{w} = (w_1, \dots, w_{n_0})'$ is the $(n_0 \times 1)$ vector of weights and $\boldsymbol{X}_{n_0}$ is the $n_0 \times K$ matrix of time-invariant covariates for the donors.\footnote{We need to take the transpose  of $\boldsymbol{Y}^{pre}_{n_0}$  to achieve consistency with the matrix $\boldsymbol{\Theta}$ above} The $(H \times H)$ matrix $\boldsymbol{V}$ is a diagonal matrix with different \textit{predictors'} weights (that are different from \textit{donors'} weights $\boldsymbol{w}$) on the linear combination of outcomes and covariates that we use. 

\textit{Matching on outcomes} approaches consider predicting covariates a secondary concern and instead aims at imputing the missing elements of $\boldsymbol{\Theta}$. As a result, the matrix $\boldsymbol{M}$ is chosen so that the covariates are excluded and we solve the problem:
\begin{gather}
    \min_{\boldsymbol{w} \in \Delta} \left( 
    \boldsymbol{y}^{pre}_{n_0+1} -
    \boldsymbol{Y}^{pre}_{n_0}  \boldsymbol{w}\right)' \left( 
    \boldsymbol{y}^{pre}_{n_0+1}  -  
    \boldsymbol{Y}^{pre}_{n_0}  \boldsymbol{w}\right)
    \label{sta_scm_pro}
\end{gather}
For this approach, imputing the matrix $\boldsymbol{\Theta}$ is the most important factor when constructing SCs. If our SCs match well the time-series of the outcome in the pre-treatment period, then we should not be too worried about balance on covariates. In this case,  $\boldsymbol{V}$ is often chosen to be the identity matrix. 

Instead of matching on the full set of outcomes separately $\boldsymbol{y}^{pre}_{n_0+1}$, \textit{matching on covariates} takes an average of the pretreatment outcomes $\bar{y}_{n_0+1, pre} = \sum_{t=1}^{T_0} y_{n_0+1, t} / T_0$ and stack this average on top of the covariates. By choosing an appropriate $\boldsymbol{M}$, the general optimisation problem in  (\ref{gen_scm_pro}) reduces to:
\begin{gather}
    \min_{\boldsymbol{w} \in \Delta} \left( \begin{pmatrix}
    \bar{y}^{pre}_{n_0+1} \\
    \boldsymbol{x}_{n_0+1}
    \end{pmatrix} - \begin{pmatrix}
    \boldsymbol{\bar{y}}^{pre}_{n_0} \\
    \boldsymbol{X}_{n_0}
    \end{pmatrix} \boldsymbol{w} \right)' \boldsymbol{V} \left( \begin{pmatrix}
    \bar{y}^{pre}_{n_0+1} \\
    \boldsymbol{x}_{n_0+1}
    \end{pmatrix} - \begin{pmatrix}
    \boldsymbol{\bar{y}}^{pre}_{n_0} \\
    \boldsymbol{X}_{n_0}
    \end{pmatrix} \boldsymbol{w} \right)
\end{gather}
where $\boldsymbol{\bar{y}}^{pre}_{n_0}$ is a row vector of averages for the donors. The diagonal $(K+1 \times K+1)$ matrix $\boldsymbol{V}$ can be chosen to ensure we have the same scale on all covariates and the average outcome or can optimise some other criteria, e.g. predictors are proportional to their importance in explaining the outcome variable \citep{aba20}. 

So, should we be matching outcomes or covariates?  \citet{kau18} warns that if we decide to ignore covariates, we may get misleading results in certain cases even though our SC follows the true time series better in the pre-treatment period. Nevertheless, \citet{bot19} show that SC methods have good theoretical properties in cases when we fail to match on covariates. Furthermore, unless $T_0$ is very long,  there seem to be ample possibilities for specification-searching of counterfactuals under both matching on covariates and outcomes \citep{pos20}. Therefore, the place of covariates in SC methods is ambiguous at present. 

\newpage

\subsection{Details on Estimation}\label{sec_est_det}
This Appendix provides details on estimating  the weights given by the correlated random coefficients model:
\begin{gather}
  y_{it}(0) = \eta_i + \sum_{j=1}^{n_0} w_{ij} y_{jt}(0) + e_{it} \quad s.t.  \label{CSC_mod_2} \\
  w_{ij} = \underbrace{\omega_j}_{ind.-invariant} + \underbrace{\boldsymbol{x}_i \boldsymbol{\alpha}^{j}}_{ind.-specific} \quad \quad (Random \ Coef.) \nonumber \\
    \sum_{j=1}^{n_0} w_{ij} = 1 \quad \quad  \forall j: w_{ij} \geq 0 \quad \quad (SC \ Constraints)  \nonumber 
\end{gather}
where we also allow for an intercept $\eta_i$, following suggestions in \citet{fer19} and \citet{dou18}. Proposition \ref{prop_optim_prob_full} gives us two equivalent formulations of the same constrained quadratic optimisation problem from (\ref{CSC_mod_2}): in scalar form (\ref{opt_pr_scal}) that was discussed in the main  body of the paper and in matrix form (\ref{opt_pr_matr}).

To understand the latter form in Proposition \ref{prop_optim_prob_full}, let us define $\otimes$ to be the Kronecker product of two matrices and $\boldsymbol{\iota}_{n_1}$ to be a ($n_1 \times 1$) column  vector of $1$-s. We also introduce the operator $Vec(.)$ for vectorising some matrix $\boldsymbol{A}$. In particular, $Vec(A)$ takes the columns of $\boldsymbol{A}$ and stack them on top of each other.\footnote{For example, $Vec\begin{pmatrix} 1 & 2 \\ 3 & 4 \end{pmatrix} = (1,3,2,4)' $}
\begin{restatable}{proposition}{propoptimprobfull}
    The parameters $\boldsymbol{\alpha}$ and $\boldsymbol{\omega}$ can be estimated by solving the optimisation problem: 
     \begin{equation} 
        \begin{gathered}
            \max_{\alpha_j^k, \omega_j} \sum_{i=1}^{n_1} \sum_{t=1}^{T_0} \left( y_{it} - \eta_i - \sum_{j=1}^{n_0} y_{jt} \omega_j - \sum_{j=1}^{n_0} \sum_{k=1}^K \alpha^{k}_j y_{jt} x_i^k  \right)^2 \quad \quad s.t. \\
     \quad  \forall i: \  \sum_{j=1}^{n_0} \left(\omega_j + \sum_{k=1} x_{i}^k \alpha^k_j\right) = 1 \quad \quad \quad  \forall (i,j): \ \omega_j + \sum_{k=1}^K x_{i}^k \alpha^k_j \geq 0
            \label{opt_pr_scal}
        \end{gathered}
    \end{equation}  
    where $\boldsymbol{x}_i$ and $\boldsymbol{\alpha}_j$ are $(1 \times K)$ \underline{row} vectors.  Equivalently, in matrix notation, we have: 
\begin{equation}
    \begin{gathered}
    \max_{\boldsymbol{\omega}, \boldsymbol{\alpha}} \quad  Vec\left((\boldsymbol{Y}^{pre}_{n1})' -     \boldsymbol{\iota}_{T_0} \otimes \boldsymbol{\eta}'  - (\boldsymbol{Y}^{pre}_{n_0})' \boldsymbol{\omega} \otimes \boldsymbol{\iota}'_{n_1} -   (\boldsymbol{Y}^{pre}_{n_0})' \boldsymbol{\alpha} \otimes \boldsymbol{X}'_{n_1}\right)'\\ 
    Vec\left((\boldsymbol{Y}^{pre}_{n1})' -     \boldsymbol{\iota}_{T_0} \otimes \boldsymbol{\eta}'  - (\boldsymbol{Y}^{pre}_{n_0})' \boldsymbol{\omega} \otimes \boldsymbol{\iota}'_{n_1} -   (\boldsymbol{Y}^{pre}_{n_0})' \boldsymbol{\alpha} \otimes \boldsymbol{X}'_{n_1} \right) \quad \quad s.t. \\ 
    (\boldsymbol{\iota}_{n_1}  \otimes \boldsymbol{\omega'} + \boldsymbol{\boldsymbol{X}}_{n_1} \boldsymbol{\alpha}' ) \boldsymbol{\iota}_{n_0}  = \boldsymbol{\iota}_{n_1} \quad \quad 
    \boldsymbol{\iota}_{n_1} \otimes \boldsymbol{\omega'} + \boldsymbol{\boldsymbol{X}}_{n_1} \boldsymbol{\alpha}' \geq \boldsymbol{O}_{(n_1 \times n_0)}
    \label{opt_pr_matr}
    \end{gathered}
\end{equation}
where as discussed $\boldsymbol{Y}^{pre}_{n_1}$ is the observed $(n_1 \times T_0)$ matrix of pre-treatment outcomes for the treated group, $\boldsymbol \eta$ is a $(n_1 \times 1)$ column vector of intercepts, $\boldsymbol{Y}^{pre}_{n_0}$ is the observed $(n_0 \times T_0)$ matrix of pre-treatment outcomes for the donors, $\boldsymbol{\omega}$ is the $(n_0 \times 1)$ vector of individual invariant weights, $\boldsymbol{\alpha}$ is $(n_0 \times K)$ matrix of coefficients and $ \boldsymbol{O}_{(n_1 \times n_0)}$ is a $(n_1 \times n_0)$ matrix of zeros.
\label{prop_optim_prob_full}
\end{restatable}

\begin{proof}
 See Appendix \ref{sec_proof_prop_optim_prob} for the proof of this Proposition.
\end{proof}

The main intuition for the proof is that we are stacking (\ref{CSC_mod_2}) horizontally over $i$ and vertically over $t$. This allows us to obtain the matrix inside the $Vec(.)$ operator. Note also that even though the constraints in the two formulations look different at first, they represent the same expressions.

Let us now explore further the matrix representation in (\ref{opt_pr_matr}).  While (\ref{opt_pr_matr}) may look daunting at first given these definitions, it is extremely useful when we would like to implement the estimator in practice via some optimisation software. The reason is that it is written in matrices and vectors that are directly observed in the data. For example, $\boldsymbol{Y}^{pre}_{n1}$ and $\boldsymbol{Y}^{pre}_{n_0}$ can be directly taken from $\boldsymbol{\Theta}$ - the matrix of outcomes. As a result, we can apply some relatively simply transformations on the data such as finding transposes and vectorising matrices in order to achieve the desired formulation. 

Lastly, regarding the implementation of CSC, we code the optimisation problem (\ref{opt_pr_matr}) in \texttt{R} using the package \texttt{CVXR}. The key advantage of \texttt{CVXR} relative to other quadratic programming solvers is that it allows us to formulate the problem in a natural mathematical language rather than the restrictive formulations that are required by other solvers \citep{fu19}. Moreover, since \texttt{CVXR} is a wrapper around other solvers, we can directly compare the performance of different solvers in terms of how accurately and how quickly they find the weights.

\newpage
\subsection{Example of multiplicity of solutions} \label{ap_multiplicity}
\begin{example}[Multiplicity of solutions]
Suppose that we have the dataset in Table \ref{ex_tab_multip} and we would like to construct a SC for Tobi. Normal SC that matches on outcomes only will be infeasible, as we have two weighted averages that match perfectly Tobi's pretreatment wage in 1978 and 1979 and his education:
\begin{gather*}
    Tobi = \frac{1}{2} Max + \frac{1}{2} Micol \\
    Tobi = \frac{2}{3} Dirk + \frac{1}{3} Yi
\end{gather*}
Moreover, for any $\lambda \in [0,1]$ we have that $\lambda(\frac{1}{2} Max + \frac{1}{2} Micol ) + (1-\lambda)(\frac{2}{3} Dirk + \frac{1}{3} Yi)$ will also work. Thus, we have an infinite number of solutions.  The question then become how to choose one SC among the infinite number of potential SC. \citet{lho20} suggest selecting the SC with the most similar outcomes and covariates. In this case, this will be Max and Micol, as the difference between their outcomes and covariates and Tobi's outcomes and covariates is smaller in absolute terms than the same difference but for Dirk and Yi. 

\begin{table}[!htbp] \centering 
\begin{threeparttable}
  \caption{ Pseudo Individual-level sample for the Mariel Boatlift} 
    \label{ex_tab_multip}
\begin{tabular}{@{\extracolsep{5pt}} lccccc} 
\\[-3ex]\hline 
\hline \\[-3ex] 
Person & City  & wage 1978 & wage 1979 & education \\ 
\hline \\[-3ex] 
Tobi & Miami   & $12$ & $15$ & $11$ \\ 
Mobarak & Miami   & $21$ & $23$ & $12$ \\ 
\dots &   \dots &  & \dots \\ 
Max & Nashville  & $11$ & $13$ & $12$ \\
Dirk & Nashville  & $9$ & $10$ & $9$ \\ 
\dots &   \dots &  & \dots \\
Micol & Chicago  & $13$ & $17$ & $10$ \\
Yi & Chicago  & $18$ & $30$ & $15$ \\ 
\hline \\[-3ex] 
\end{tabular} 
\end{threeparttable}

\end{table} 
\end{example}
\newpage 

\subsection{City-level vs Pooled SC}\label{ap_equiv}
This section is aimed at giving some intuition towards why the pooled SC nests the city-level SC as a special case. We can prove the following proposition:

\begin{restatable}{proposition}{propcityscm}
Suppose that:
\begin{enumerate}[(i)]
    \item We have a panel which is balanced across cities, e.g., for every treated and untreated city we observe the same number of individuals $n$.
    \item We have no covariates, so that we only observe $\boldsymbol{\Theta}$ and city membership
    \item We have a total of $C$ cities, of which the first $c_0$ cities are untreated and the rest $c_1$  cities are treated
\end{enumerate}
Then, if we impose the additional requirement $w_{dj} = \frac{w_d}{n}$, the pooled SC estimator and the city-level SC estimator solve the same optimisation problem.
\label{prop_city_scm}
\end{restatable}
\begin{proof}
        The proof is in Appendix \ref{subsec_prop_city_scm}.
\end{proof}

While the balance of the panel across cities may seem like a restrictive condition, we should be able to relax it by imposing in the pooled SC that each city has an equal weight overall, even though different cities contain a different number of observations. For instance, if we have many individuals within two cities Miami and Atlanta with Miami having 100 individuals and Atlanta just 50, then we can impose that each individual in Miami gets a weight $\frac{50}{150}$ and each individual in Atlanta gets a weight of $\frac{100}{150}$. For reasons of time, however, I was unable to show this more general case more formally. Nevertheless, the main result still stands: the pooled SC nests the city-level SC.

\newpage

\newpage

\subsection{DiD Inconsistency}\label{subsec_DiD_incons}
Here we sketch another argument in addition to Proposition \ref{DiD_bou} for showing that $\hat{\tau}^{DiD}$ will be inconsistent in light of the discussion in Section \ref{subsec_DiD_con}. Note that estimating $y_{it} = \rho + \gamma_i + \delta_t + \tau D_{it} + u_{it}$ will yield an inconsistent estimate of the true $\tau$, since $Cov(D_{it}, u_{it}) \neq 0$, meaning that $D_{it}$ is endogenous. Naturally, in the spirit of the Arellano-Bond estimator, the question arises if we can apply some form of transformation to the data (e.g. first differences) and instrument $D_{it}$ with some of its lagged values  $D_{is}$ in order to deal with the endogeneity issue. Unfortunately, this will not work, due to the interactive fixed effects structure. Note that we also have a further problem with consistency: $Cov(\gamma_i, u_{it}) \neq 0$ as factor loadings $\boldsymbol{\mu}_i$ enter into both (reduced-form) terms, complicating further consistency of $\hat{\tau}^{DiD}$ in (\ref{DiD_int_fixef}).

\newpage 

\subsection{Fixed parameters for simulation}\label{subsec_fix_param}

The full DGP for the simulation is specified as:
\begin{equation}
    \begin{aligned}
        (Outcome)& \quad \quad \quad y_{it} = \boldsymbol{\beta x_i'} + \boldsymbol{ \lambda_t\mu_i} + D_{it}\tau + \epsilon_{it} \\
        (Common \ Factor)& \quad \quad \quad  \boldsymbol{ \mu_i} = \boldsymbol{ \gamma}  \boldsymbol{x}'_i + \boldsymbol{v}_i \\
        (Shock \ to \ \mu_i)& \quad \quad \quad  \boldsymbol{v}_i \stackrel{iid}{\sim} N_F(\boldsymbol{\mu},\boldsymbol{I}_F) \\
        (Idiosyncratic \ shock)& \quad \quad \quad \epsilon_{it} \stackrel{iid}{\sim} N(0,1) \\
        (Categorical \ Predictor)& \quad \quad \quad x^{(1)}_i \stackrel{iid}{\sim} U[{1,2,\dots, K}] \\
        (Continuous \ Predictor)& \quad \quad \quad x^{(2)}_i \stackrel{iid}{\sim} N(0,1)
    \end{aligned}
    \label{full_dgp}
\end{equation}
Let us now detail what this DGP suggests. For $\boldsymbol{\Theta}$, we assume that every outcome follows an interactive fixed effects model, given by $(Outcome)$. In addition, the common factors $\boldsymbol{\mu}_i$ are allowed to be correlated with the time-invariant covariates, as required by Assumption \ref{ass_th}.(iii). Moreover, $\boldsymbol{ \gamma}$ is a fixed $1\times K$ vector controlling the correlation between the factor loadings and the time-invariant covariates, e.g., if we set $\boldsymbol{\gamma} = 0$ then there is no correlation. We allow $\boldsymbol{\mu}_i$ to have an expectation different from $0$ since the mean of $v$ is not 0, as $\boldsymbol{v}$'s mean is actually $\boldsymbol{\mu}$ in expectation. The idiosyncratic shocks $\epsilon_{it}$ are drawn from a standard normal distribution. On the other hand, we assume that there are two covariates in the DGP: a categorical covariate  $x^{(1)}_i$ and a continuous predictor $x^{(2)}_i$.\footnote{Both of them are time-invariant in order to avoid the need to make transformations such as taking the average of $$x^{(2)}_i$$ which is required by SC-like methods. Integrating time-varying predictors in SC methods and more generally the place of covariates in SC is still work in progress \citep{bot19}.} The motivation for including both continuous and discrete covariates stemmed from the restriction on CSC that only allows for discrete predictors.\footnote{See Section \ref{subsec_est}} So, including $x^{(2)}_i$ allows us to explore what happens to CSC when we simply recode a continuous predictor as a binary one, given that CSC cannot handle categorical predictors (see Section \ref{subsec_est}). On the other hand, we assume that $x^{(1)}_i$ is drawn from a discrete uniform distribution with $K$ categories, meaning that every observation is equally likely to belong to any category in $\{1, 2, \dots, K\}$.  Lastly, we assume $v_i$, $u_{it}$, $x^{1}_i$ and $x^{2}_i$ are all uncorrelated with each other.

The parameters $\boldsymbol{\beta}$, $\boldsymbol{\lambda}$, $\boldsymbol{\gamma}$ and $\boldsymbol{\phi}$ are treated as fixed in the assumptions of the model. So, they are constant across all simulations. Firstly for $\boldsymbol{\beta}$, we set the vector $\boldsymbol{\beta}$ with coefficients for $x^{(1)}_i$ and the $H=5$ categories of $x^{(2)}_i$ as:
$$ \boldsymbol{\beta} = (1, 0.4, 0.6, 0.8, 1, 1.2) $$
We need to fix the parameters for three matrices: the common factors $\boldsymbol{\lambda}$, the partial correlation between unobesrvables and covariates $\boldsymbol{\gamma}$ and the partial correlation between unobservables and treatment assignment $\boldsymbol{\phi}$.

Secondly, the $(T \times F)$ matrix of common factors $\boldsymbol{\lambda}$ was independently drawn once from $N(3,2)$ in (\ref{tab_lambda}). The rows are the time periods and the columns are the factors. 
\begin{gather}
    \begin{pmatrix}
        1.79 & 2.44 & 2.49 & 2.31 \\ 
        3.27 & 2.11 & 2.09 & 1.55 \\ 
        4.08 & 2.52 & 2.16 & 3.58 \\ 
        0.65 & 2.00 & 5.41 & 1.98 \\ 
        3.43 & 2.22 & 3.13 & 2.99 \\ 
        3.51 & 3.06 & 2.51 & 2.06 \\ 
        2.43 & 3.96 & 2.56 & 4.10 \\ 
        2.45 & 2.89 & 3.46 & 2.52 \\ 
    \end{pmatrix}
    \label{tab_lambda}
\end{gather}

Thirdly, we give the fixed values of the $K \times F$ matrix $\boldsymbol{\gamma}$ in (\ref{tab_gamma}). These are the values, controlling the (partial) correlation between the time-invariant covariates $\boldsymbol{x}_i$ and the unobserved factor loadings $\boldsymbol{\mu}$ from (\ref{full_dgp}). We draw each member of $\boldsymbol{\gamma}$ independently from a standard normal $N(1,0)$ distribution:

\begin{gather}
    \begin{pmatrix}
            -1.48 & -0.32 & -0.78 & 0.51 \\ 
            1.58 & -0.63 & 0.01 & -0.29 \\ 
            -0.96 & -0.11 & -0.15 & 0.22 \\ 
            -0.92 & 0.43 & -0.70 & 2.01 \\ 
            -2.00 & -0.78 & 1.19 & 1.01 \\ 
            -0.27 & -1.29 & 0.34 & -0.30 \\ 
    \end{pmatrix}
    \label{tab_gamma}
\end{gather}

Lastly, the parameter$\boldsymbol{\phi}$  determines the degree of correlation between the factor loadings and treatment assignment (See  \ref{tr_not_at_rand}). In particular, we set the parameter as 
$$  \boldsymbol{\phi} = (-1.12, -0.46,  3.12,  0.14)' $$
Note that we hold it constant across all of our simulations, even though we generated it initially from a $N(0,1)$ distribution.

After generating the initial values for these parameters, we hold them to be the same across all of our simulations and for all estimators that we compare in order to avoid making the performance of the different estimators dependent on the particular values being generated in each case.\footnote{For instance, suppose we decide to generate different values these parameters when simulating the data for different estimators. However, we may get very unlucky with our draws of values for iDiD but very lucky with our values of CSC, so that in the end it may seem that CSC performs better but in reality this is simply due to the randomness.}

\newpage
\subsection{More details on PSC}\label{subsec_psc_det}
How does PSC \citep{lho20} work? As mentioned in Section \ref{subsec_man}, it creates separate SCs for every treated observation in a way that tackles multiplicity of solutions. Out of all possible combinations of donors that exactly fit the time-series for treated observation $i$, PSC choses the one combination, in which the donors are most similar to $i$ in terms of observables in addition to matching the weighted average of unit $i$. We can show this more formally via the optimisation problem solved by PSC. For every treated unit $i$ in $\{n_0+1, n_0+2, \dots, N\}$, we have in scalar form\footnote{For the purpose of clarity, we diverge from the widely used notation in Section \ref{sec_est_cons}}:
\begin{gather*}
    \min_{\boldsymbol{w}_i \in \Delta} \underbrace{\sum_{t=1}^{T_0} \left(y_{it} - \sum_{j=1}^{n_0} y_{jt}w_{ij} \right)^2 + \sum_{k=1}^{K} \left(x^{(k)}_{i} - \sum_{j=1}^{n_0} x^{(k)}_{j}w_{ij} \right)^2}_{Synthetic \ Control} + \\ 
    \underbrace{\lambda \sum_{j=1}^{n_0} w_{ij} \left( \sum_{t=1}^{T_0} (y_{it} - y_{jt})^2 + \sum_{k=1}^K (x^{(k)}_{i} - x^{(k)}_{j})^2 \right)}_{Penalty}
\end{gather*}
where $\lambda$ is a tuning parameter that can be chosen via cross-validation.\footnote{When we have data on outcomes $y_{it}$ for $T_0 \geq 2$, we can calculate $n_1$ (penalised) SCs using data on the time-invariant covariates $x_i^k$ and the first $T_0-1$ outcomes $y_1$. We can then evaluate the performance of the $n_1$ different SC at time $T_0$.} The SC part of the objective function ensures that the weighted average matches the treated unit's outcomes and covariates. However, on its own, it is not sufficient to select a unique solution. So, the penalty part ensures that we select the set of donors that are also similar to the treated unit in terms of observable characteristics (and not just with a similar weighted average), i.e., PSC prefers similar donors rather than different donors. A numerical example of how PSC works is presented in Appendix \ref{ap_multiplicity}.\footnote{As in the case of CSC, we coded the optimisation problem in \texttt{R} using \texttt{CVXR} package \citep{fu19}.} 

\newpage

\subsection{Further Simulation Results } \label{subsec_rmse_att}
Table \ref{tab_sim_mse_att} gives the results from the same simulation in \textbf{Algorithm 1} that we used in the main body of the test but not for the estimation error of $\hat{\tau}$ but rather its RMSE. The general features of the results are the same, as discussed in Section \ref{sec_sim}. One interesting feature of the results, however, is that once we take account of the variance we can see even more clearly how iDiD does a better job at predicting the ATT relative to the case where we study the pure estimation error.

\begin{table}[hbtp!]
    \centering
    \caption{ \textit{RMSE between $\hat{\tau}$ and true $\tau$ under different DGPs across 1000 simulations} }
    \label{tab_sim_mse_att}
    \begingroup
    \setlength\extrarowheight{1.2pt} 
    \[
    \begin{array}{lccc|c}
   \\[-1.8ex]\hline 
    \hline 
    \mathrm{Change} & \mathrm{CSC} & \mathrm{fDiD} & \mathrm{PSC} & \mathrm{iDiD}  \\ 
    \hline 
     \mathrm{Baseline} & \textbf{1.91} & 2.40 &  2.68 & \textbf{0.98} \\ 
    \mathrm{Less \ Correlation} & \textbf{1.87} & 2.35 &  2.56 & \textbf{0.99} \\ 
   \mathrm{Random \ Assignment} & \textbf{1.88} & 2.15 &  2.51 & \textbf{1.00} \\ 
    \hline 
    N = 150 & \textbf{1.88} & 2.37 &  2.45 & \textbf{0.98} \\ 
    N = 200 & \textbf{1.88} & 2.37 &  2.37 & \textbf{0.97} \\ 
    \hline 
    T = 4 & \textbf{5.70} & 6.35 & 10.29 & \textbf{0.94} \\ 
    T = 8 & \textbf{1.42} & \textbf{1.40} &  2.28 & \textbf{1.01} \\
    \hline 
    F = 2 & \textbf{1.98} & 3.08 &  3.05 & \textbf{0.99}\\ 
    F = 4 & \textbf{2.06} & 2.55 &  3.15 & \textbf{0.99} \\ 
    \hline 
    E[\pi_i] = 0.25 & \textbf{1.92} & 2.35 &  2.76 & \textbf{0.98} \\ 
    E[\pi_i] = 0.40 & \textbf{1.97} & 2.30 &  2.86 & \textbf{0.97}\\ 
    \hline
    \end{array}
    \]
    \endgroup
\end{table}

\newpage

\subsection{Variables from PSID} \label{psid_var}
\subsubsection{Sample restrictions}
The PSID sample that we use consists of 11 waves (1974-1984) and is restricted to male heads of households.  Until recently, the heads of households in PSID were considered to be the husband in a traditional heterosexual marriage. This formulation has been controversial and was justifiably described as ``anachronistic" \citep{car18} and so it has been changed in recent years but not for the time period studied here. As a result, our sample is restricted to males of working age (19-64), for whom we have data on all 11 years. 

Excluding women can be justified in two ways: i) the data for them in PSID is very sparse for the period we studied and  ii) controlling for the general increase in women's participation in the labour force at the time is challenging given that we work with time-invariant covariates \citep{bor17}. However, further work on the Mariel Boatlift with PSID should strive to include women in the analysis. 

\subsubsection{Variable Selection}
Turning to the variables selected in the sample, the main outcomes of interest are total hours worked per year as an indicator of labour supply and hourly (log) wages. While previous studies have found no effect on aggregate unemployment, \citet{dus16} has recently stated the importance of relaxing the assumption of perfectly inelastic labour supply and studying labour supply effect in response to immigration. This would amount to making $LS$ in Figure \ref{fig_imm_shock} downward sloping rather than a vertical line. In addition, we also include various covariates on occupation, industry, health, marriage status, age and education. 

Table \ref{tab_var_psid} below presents the variables that we include in the analysis of the empirical application in the paper. It also discusses any recoding that we have used to ensure that the variables can be meaningfully included in the analysis. For example, days being ill has been recoded to a binary variable, indicating being one of the top 25\% people who spent the longest period of time being ill. More importantly,  we turn the covariates below into time-invariant variables by taking their average over the pre-treatment period. In addition, we give the name of the variable only in 1975. The reason is that variable names differ across years but from PSID's Variable Search engine one can identify which are the comparable variables across years.

Next, it should be remarked that the original sample we used was from 1973 to 1985. However, we need to impose restrictions, due to the fact that one of the outcomes (hours worked) was retrospective for the past year and the other outcome (wages) was recorded for the week of interview. Lastly, since PSID's website is a bit hard to navigate, we used \texttt{R} package \texttt{psidR} to download the data \citep{osw20}.

\begin{table}[htbp]
    \centering
    \caption{ Variables used in the analysis for the empirical application}
    \label{tab_var_psid}
    \begin{tabular}{llcc}
       \\[-1.8ex]\hline 
        \hline 
        1975 Var. & Interpretation & Use & Recoding and Notes \\
        \hline 
        V3823 & Tot. Hours Worked & Outcome & Hours worked in \textit{previous} year \\
        V4093 & Hourly Wage & Outcome & Hourly wage in  \textit{current} year \\
        \hline
        V3803 & State Residence & Treatment & Used to find who's in Florida \\
        V4204 & Race & Covariate & Coded as white vs non-white \\
        V4093 & Education & Covariate & Coded as college vs no college \\
        V4194 & Married  & Covariate & Married or not \\
        V3969A & Head Industry & Covariate & Coded as dummy variables \\
        V3968A & Head Occupation & Covariate & Coded as dummy variables \\
        V3921 & Age & Covariate & Used to restrict sample \\
        V3825 & Hours Ill & Covariate & Coded as binary variable \\
        \hline
    \end{tabular}
\end{table}

\newpage

\subsection{HTT for other Specifications}\label{subsec_htt_diff_spec}
As discussed in Section \ref{subsec_hte}, we also calculated heterogeneous treatment effects for other lengths of the training period. The overall impression remains, however, that the only variable for which we consistently find a significant negative effects is wages of low-skilled workers. For other combinations (e.g. labour supply for high-skilled workers), we cannot reject the null hypothesis of no effect. The results can be found in  Figure \ref{fig_hte_2}, Figure \ref{fig_hte_3} and  Figure \ref{fig_hte_4}.

\begin{figure}[hbtp!]
    \centering
    \includegraphics{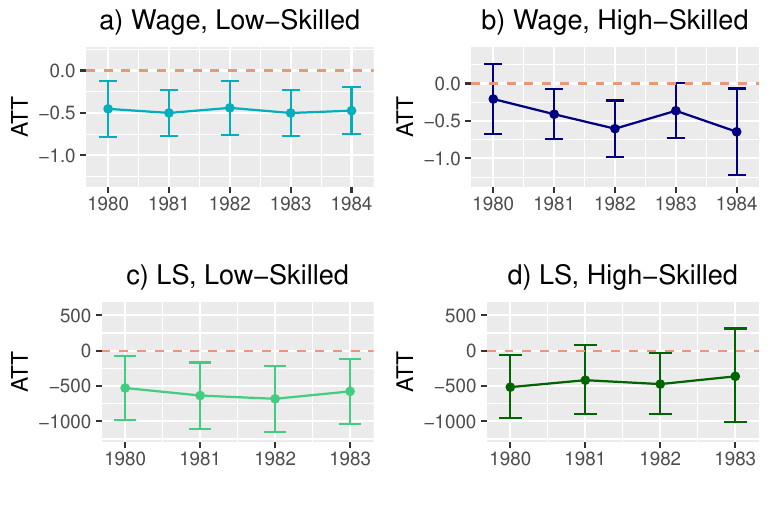}
    \caption{Treatment Effects when $T_{train} = 2$}
    \label{fig_hte_2}
\end{figure}

\begin{figure}
    \centering
    \includegraphics{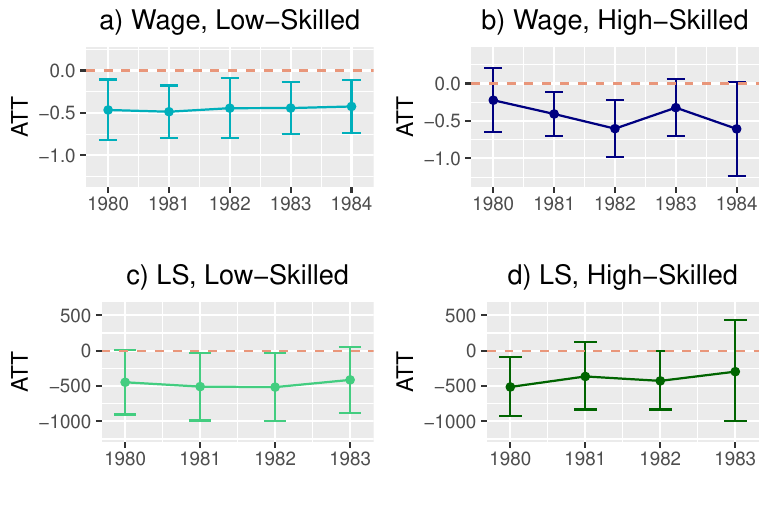}
    \caption{Treatment Effects when $T_{train} = 3$}
    \label{fig_hte_3}
\end{figure}

\begin{figure}
    \centering
    \includegraphics{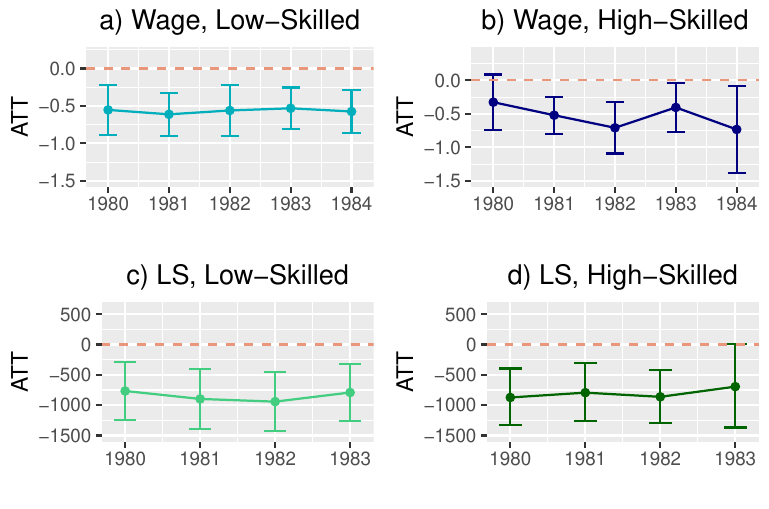}
    \caption{Treatment Effects when $T_{train} = 4$ for wages and $T_{train} = 1$ for LS }
    \label{fig_hte_4}
\end{figure}

\newpage
\section{Proofs of main results}\label{sec_proof}

\subsection{Proof of Lemma \ref{lem_aba} }\label{sec_proof_lem_aba}
\lemaba*

\begin{proof}
The proof is a generalisation for $n_1 > 1$ of a result in \citet[][503-504]{aba10}. Define first the unobserved true ATT for all pre-treatment and post-treatment $t \in \{1,2, \dots, T_0,  T \}$ as:
$$ \tau_t := \frac{\sum_{i=n_0+1}^N [y_{it}(1) - y_{it}(0) ]}{n_1} $$
Next, we can define the estimated ATT for $\forall t \in \{1,2, \dots, T_0,  T \}$ as:
$$ \hat{\tau}_t := \frac{\sum_{i=n_0+1}^N [y_{it}(1) - \hat{y}_{it}(0) ]}{n_1} $$
where we note that $\hat{\tau}_T = \hat{\tau}^{CSC}$ and $\tau_T = \tau$ which form the estimation error in the proposition. We begin by substituting the two quantities for all pre-treatment and post-treatment $t \in \{1,2, \dots, T_0,  T \}$:
\begin{align*}
    \widehat{\tau}_t - \tau_t  =& \frac{\sum_{i=n_0+1}^N [y_{it}(1) - \hat{y}_{it}(0) ]}{n_1} -  \frac{\sum_{i=n_0+1}^N [y_{it}(1) - y_{it}(0) ]}{n_1} = \\ 
    =& \frac{\sum_{i=n_0+1}^N [y_{it}(0) - \hat{y}_{it}(0)] }{n_1} = \\
    =& \frac{1}{n_1} \sum_{i=n_0+1}^{n_1} \left(y_{it}(0) - \sum_{j=1}^{n_0} \hat{w}_{ij}y_{jt}(0) \right) 
\end{align*} 
where the first row substitutes the definitions of $\hat{\tau}$ and $\tau$, the second row cancels out the sums of $y_{it}(0)$ over $i$ and $t$ and in the last row we use the \underline{Exact Fit} assumption that $\hat{y}_{it}(0) = \sum_{j=1}^{n_0} \hat{w}_{ij}y_{jt}(0)$. We have for all $t \in \{1,2, \dots, T_0, \dots, T \}$ that:
\begin{gather*}
    \frac{1}{n_1} \sum_{i=n_0+1}^N (y_{it}(0) - \sum_{j=1}^{n_0} \hat{w}_{ij} y_{jt}(0) ) = \\
    \frac{1}{n_1} \sum_{i=n_0+1}^N \left( \boldsymbol{\theta_t} \boldsymbol{x}_i + \boldsymbol{\lambda}_t \boldsymbol{\mu}_i + \epsilon_{it} - \left(\sum_{j=1}^{n_0} \hat{w}_{ij}(\boldsymbol{\theta}_t \boldsymbol{x}_j +  \boldsymbol{\lambda}_t \boldsymbol{\mu}_j + \epsilon_{it} ) \right) \right) = \\
    = \frac{1}{n_1} \sum_{i=n_0+1}^N\left( \boldsymbol{\theta}_t \left( \boldsymbol x_i - \sum_{j=1}^{n_0} \widehat{w}_{ij}\boldsymbol x_j \right) + \boldsymbol{\lambda}_t \left( \boldsymbol{\mu}_i - \sum_{j=1}^{n_0} \hat{w}_{ij} \boldsymbol{\mu}_j \right) + \left( \epsilon_{it} - \sum_{j=1}^{n_0} \hat{w}_{ij} \epsilon_{jt} \right) \right) 
 \end{gather*}
where the second follows from the assumption for the DGP (namely, $y_{it}(0) = \boldsymbol{\theta}_t \boldsymbol{x}_j +  \boldsymbol{\lambda}_t \boldsymbol{\mu}_j + \epsilon_{it}$) and the third row simply rearranges the second row. We can then apply \textbf{Assumption \ref{ex_fit_as}} \underline{Exact Fit} for the covariates $\boldsymbol x_i $ in the last expression, so that it can be rewritten as:
\begin{gather*}
    \frac{1}{n_1} \sum_{i=n_0+1}^N (y_{it}(0) - \sum_{j=1}^{n_0} \hat{w}_{ij}y_{jt}(0))  =  \frac{1}{n_1} \sum_{i=n_0+1}^N\left( \boldsymbol{\lambda}_t \left( \boldsymbol{\mu}_i - \sum_{j=1}^{n_0} \hat{w}_{ij} \boldsymbol{\mu}_j \right) + \left( \epsilon_{it} - \sum_{j=1}^{n_0} \hat{w}_{ij} \epsilon_{jt} \right) \right) \quad \quad (*)
 \end{gather*}
Next, following \citet{aba10}, we will eliminate from $(*)$ the expression involving $\boldsymbol{\mu}$-s. So, we stack over $t \in \{1,2, \dots, T_0\}$ for the pretreatment period:
\begin{gather*}
    \sum_{i=n_0+1}^N \left( \boldsymbol{y}_{i,pre}(0) - \sum_{j=1}^{n_0}  \widehat{w}_{ij} \boldsymbol{y}_{j,pre}(0) \right) =  \sum_{i=n_0+1}^N \left[ \boldsymbol{\lambda}_{pre}  \left( \boldsymbol{\mu}_i - \sum_{j=1}^{n_1} \hat{w}_{ij} \boldsymbol{\mu}_j \right) + \left( \boldsymbol{\epsilon}_{i,pre} - \sum_{j=1}^{n_1} \hat{w}_{ij} \boldsymbol{\epsilon}_{j,pre}  \right) \right] \\
    \sum_{i=n_0+1}^N \left(  \sum_{j=1}^{n_1} \hat{w}_{ij} \boldsymbol{\epsilon}_{j,pre} - \boldsymbol{\epsilon}_{i,pre} \right) = \boldsymbol{\lambda}_{pre} \sum_{i=n_0+1}^N    \left( \boldsymbol{\mu}_i - \sum_{j=1}^{n_1} \hat{w}_{ij} \boldsymbol{\mu}_j \right)  \\
   (\boldsymbol{\lambda}'_{pre}\boldsymbol{\lambda}_{pre})^{-1} \boldsymbol{\lambda}'_{pre}  \sum_{i=n_0+1}^N \left(  \sum_{j=1}^{n_1} \hat{w}_{ij}  \boldsymbol{\epsilon}_{j,pre} - \boldsymbol{\epsilon}_{i,pre} \right) = \sum_{i=n_0+1}^N    \left( \boldsymbol{\mu}_i - \sum_{j=1}^{n_1} \hat{w}_{ij} \boldsymbol{\mu}_j \right) \quad \quad (**)
\end{gather*}
where the second row follows from the first due to \textbf{Assumption \ref{ex_fit_as}} \underline{Exact Fit} for the outcomes in the pre-treatment period and the third row follows due to assumption of invertible $\boldsymbol{\lambda}'_{pre}\boldsymbol{\lambda}_{pre}$ and after multiplying both sides by $(\boldsymbol{\lambda}'_{pre}\boldsymbol{\lambda}_{pre})^{-1} \boldsymbol{\lambda}'_{pre}$. Next, we substitute the last expression back into $(*)$ to obtain for $T=T_0+1$, i.e. in the single post-treatment period: 
\begin{gather*}
     \hat{\tau}^{CSC} - \tau = \frac{1}{n_1} \sum_{i=n_0+1}^N (y_{iT}(0) - \sum_{j=1}^{n_0} \hat{w}_{ij} y_{jT}(0) ) = \\
     = \frac{1}{n_1} \boldsymbol{\lambda}_T \sum_{i=n_0+1}^N\left(  (\boldsymbol{\lambda}'_{pre}\boldsymbol{\lambda}_{pre})^{-1} \boldsymbol{\lambda}'_{pre}  \left(  \sum_{j=1}^{n_1} \hat{w}_{ij}  \boldsymbol{\epsilon}_{j,pre} - \boldsymbol{\epsilon}_{i,pre} \right) \right) + \frac{1}{n_1} \left[ \sum_{i=n_0+1}^N  \left( \epsilon_{iT} - \sum_{j=1}^{n_0} \hat{w}_{ij} \epsilon_{jT} \right)  \right] \\
     = \frac{1}{n_1}   \left[ \boldsymbol{\lambda}_t (\boldsymbol{\lambda}'_{pre}\boldsymbol{\lambda}_{pre})^{-1} \boldsymbol{\lambda}'_{pre} \sum_{i=n_0+1}^N   \left(  \sum_{j=1}^{n_1} \hat{w}_{ij}  \boldsymbol{\epsilon}_{j,pre} - \boldsymbol{\epsilon}_{i,pre} \right)\right] + \frac{1}{n_1} \left[ \sum_{i=n_0+1}^N  \left( \epsilon_{iT} - \sum_{j=1}^{n_0} \hat{w}_{ij} \epsilon_{jT} \right)  \right]
\end{gather*}
where the second row follows after substituting  $(**)$ into $(*)$ and after separating the summation over $i$. The third row follows by simply rearranging the second row.
\end{proof}

\newpage

\subsection{Proof of Lemma 2}

\begin{lemma}
    Suppose $Z_1, Z_2, \dots Z_n$ are all \texttt{subG}$(\sigma^2_i)$. Then, their average satisfies:
    $$ \frac{\sum_{i=1}^n Z_i}{n} \sim SubG\left(\frac{1}{n^2} \left(\sum_{i=1}^n \sigma_i \right)^2 \right) $$
    \label{lem_subg}
\end{lemma}
\begin{proof}
We will show this by induction. The expression holds trivially for $n=1$, as if $Z_1 \sim$ \texttt{SubG}$(\sigma^2)$, then  $aZ_1 \sim$ \texttt{SubG}$(a^2\sigma^2)$. Suppose for $k-1$, it holds that:
$$ \sum_{i=1}^{k-1} Z_i \sim SubG\left( \left(\sum_{i=1}^{k-1} \sigma_i \right)^2 \right) \equiv SubG(\sigma^2) $$
Then, for $k$ we have:
\begin{gather*}
    E\left[ e^{((h\sum_{i=1}^{k-1} Z_i) + Z_k)} \right] \leq \\
    \leq \left[ E\left( \exp\left(h \frac{\sum_{i}Z_i(\sigma + \sigma_{k+1})}{\sigma} \right) \right) \right]^{\frac{\sigma}{\sigma + \sigma_k}} \left[ E\left(h \exp\left( \frac{Z_k(\sigma + \sigma_{k+1})}{\sigma_{k}} \right) \right) \right]^{\frac{\sigma_k}{\sigma + \sigma_k}} \leq \\
    \leq \left[ \exp\left( \frac{h^2 (\sigma + \sigma_{k})^2}{2\sigma^2}\sigma^2 \right) \right]^{\frac{\sigma}{\sigma + \sigma_k}} \left[ \exp\left(  \frac{h^2(\sigma + \sigma_{k})^2 }{2\sigma^2_k} \sigma^2_k \right) \right]^{\frac{\sigma_k}{\sigma + \sigma_k}} =\\
    = \left[ \exp\left(\frac{h^2\sigma(\sigma+ \sigma_{k})}{2} \right) \right] \left[ \exp \left( \frac{h^2\sigma_k(\sigma+ \sigma_{k})}{2} \right) \right] = \exp\left( \frac{(h(\sigma+\sigma_{k}))^2}{2} \right)
\end{gather*}
where the second row follows due to Holder's inequality for $p=\frac{\sum_{i}Z_i(\sigma + \sigma_{k+1})}{\sigma}$ and $q = \frac{\sum_{k}Z_k(\sigma + \sigma_{k+1})}{\sigma}$. The third row follows due to the induction step for $k-1$ and the definition of \texttt{subG}. The last expression gives the necessary result for the induction: 
$$ \sum_{i=1}^{k} Z_i \sim SubG\left( \left(\sum_{i=1}^{k} \sigma_i \right)^2 \right)  $$
Next, we multiply by $1/k$ to obtain: 
$$ \frac{\sum_{i=1}^{k} Z_i}{k} \sim SubG\left( \left( \frac{1}{k}\sum_{i=1}^{k} \sigma_i \right)^2 \right)  $$
\end{proof}

\newpage
\subsection{Proof of Proposition \ref{prop_CSC_est_err}} \label{sec_proof_prop_CSC_est_err}
\propCSCesterr*

\begin{proof}
The strategy for proof is related to the result in \citet{aba10} and Theorem 2 in \citet{ben21}. However, the final bound is different. We proceed in four steps where we prove the claims:
\begin{itemize}
    \item We begin by proving that \begin{claim}
    The estimation error for $\widehat{\tau}_s$ for post-treatment period $s > T_0$ is:
    \begin{gather}
        \widehat{\tau}_s - \tau_s = \\  
        \frac{1}{n_1} \left[ \boldsymbol{\lambda}_t (\boldsymbol{\lambda}'_{pre}\boldsymbol{\lambda}_{pre})^{-1} \boldsymbol{\lambda}'_{pre} \sum_{i=n_0+1}^N   \left(  \sum_{j=1}^{n_1} \hat{w}_{ij}  \boldsymbol{\epsilon}_{j,pre}  \right)\right] - \label{att_cl_eq_1} \\ 
        \frac{1}{n_1} \left[ \boldsymbol{\lambda}_t (\boldsymbol{\lambda}'_{pre}\boldsymbol{\lambda}_{pre})^{-1} \boldsymbol{\lambda}'_{pre} \sum_{i=n_0+1}^N  \boldsymbol{\epsilon}_{i,pre} \right] +  \label{att_cl_eq_2} \\ 
        \frac{1}{n_1} \sum_{i=n_0+1}^N \left[ \epsilon_{is} - \sum_{j=1}^{n_0} \hat{w}_{ij} \epsilon_{js} \right]
        \label{att_cl_eq_3}
    \end{gather}
    \label{er_prop_claim_1}
    \end{claim}
   
    \item We will bound each  of (\ref{att_cl_eq_1}), (\ref{att_cl_eq_2})  and (\ref{att_cl_eq_3}) in the last expression, similarly to \citet{ben21}. Firstly, for \ref{att_cl_eq_1} we find an upper bound: 
    \begin{claim}
    With probability $1-\exp(-h^2/4)$, it holds that $$
        \frac{1}{n_1} \sum_{i=n_0+1}^{n_1} \boldsymbol{\lambda}_s (\boldsymbol{\lambda'_{pre} \lambda_{pre}})^{-1} \boldsymbol \lambda'_{pre} \left[ \sum_{j=1}^{n_0} \hat{w}_{ij}   \boldsymbol{e}_{j,pre} \right] \leq \frac{F\tilde{\lambda}^2}{\phi_{min}\left(\frac{1}{T_0}\boldsymbol{\lambda'_{pre} \lambda_{pre}}\right) } \sqrt{\frac{2\sigma^2n_0}{T_0} + \frac{h\sigma}{T_0^{1.5}}} $$
    \label{er_prop_claim_4}
    \end{claim}
    \item Secondly, for (\ref{att_cl_eq_2}) we can found a lower bound, as it is subtracted:
    \begin{claim}
    With probability $1-\exp(-h^2/2)$ we have:
    $$ \frac{1}{n_1} \left[ \boldsymbol{\lambda}_t (\boldsymbol{\lambda}'_{pre}\boldsymbol{\lambda}_{pre})^{-1} \boldsymbol{\lambda}'_{pre} \sum_{i=n_0+1}^N  \boldsymbol{\epsilon}_{i,pre} \right] > \frac{-hF^2 \underaccent{\tilde}{\lambda}^2 \sigma^2}{\phi_{max} \left(\frac{1}{T_0}\boldsymbol{\lambda'_{pre} \lambda_{pre}}\right)} $$
    \label{er_prop_claim_2}
    \end{claim}
    \item Thirdly, (\ref{att_cl_eq_3}) an upper bound can be found \begin{claim} With prob. $1-\exp(-h^2/2)$ we have
    $ \frac{1}{n_1} \sum_{i=n_0+1}^N \left[ \epsilon_{is} - \sum_{j=1}^{n_0} \hat{w}^j \epsilon_{js} \right] \leq h\sigma \sqrt{2}$ 
    \label{er_prop_claim_3}
    \end{claim}
    \item We finish the proof by combining the last three claims via Frechet’s inequality and obtaining the statement in the initial proposition.
\end{itemize}
Claim \ref{er_prop_claim_1} follows after a slight algebraic manipulation of Abadie's representation in Lemma \ref{lem_aba}.

Now, we will show Claim $\ref{er_prop_claim_4}$. Note that $\hat{w}_{ij}$ are not independent from $\epsilon_{jt}$ for some pretreatment $t < T_0$. Let us now proceed with the proof:
\begin{gather*}
     \frac{1}{n_1} \sum_{i=n_0+1}^{n_1} \boldsymbol{\lambda}_s (\boldsymbol{\lambda'_{pre} \lambda_{pre}})^{-1} \boldsymbol \lambda'_{pre} \left[ \sum_{j=1}^{n_0} \hat{w}_{ij}   \boldsymbol{e}_{j,pre} \right] \stackrel{min \ eigenvalue}{\leq} \\ 
     \frac{1}{n_1T_0} \sum_{i=n_0+1}^{N} \left( \frac{\sum_{j=1}^{n_0} \hat{w}_{ij}(\boldsymbol{\lambda}_s \boldsymbol{\lambda}'_{pre} \boldsymbol{\epsilon}_{j,pre})\ }{\phi_{min}\left(\frac{1}{T_0}\boldsymbol{\lambda'_{pre} \lambda_{pre}}\right)}  \right) \stackrel{Cauchy-Schwarz \ Ineq.}{\leq} \\
     \frac{1}{n_1T_0\phi_{min}\left(\frac{1}{T_0}\boldsymbol{\lambda'_{pre} \lambda_{pre}}\right)} \sum_{i=n_0+1}^{N} ||\hat{w}_i||_2 ||\boldsymbol{\lambda}_s \boldsymbol{\lambda}'_{pre} \boldsymbol{\epsilon}_{n_0, pre}  ||_2 \stackrel{def. \ \hat{\boldsymbol w}}{\leq} \\
     \frac{1}{n_1T_0\phi_{min}\left(\frac{1}{T_0}\boldsymbol{\lambda'_{pre} \lambda_{pre}}\right)} \sum_{i=n_0+1}^N \left|\left| \sum_{f=1}^F \sum_{t=1}^{T_0} \lambda_{sf} \lambda_{tf} \boldsymbol{\epsilon}_{n_0,t} \right|\right|_2 \stackrel{\max_{t,f} |\lambda_{tf}|}{\leq} \\
      \frac{n_1}{n_1T_0\phi_{min}\left(\frac{1}{T_0}\boldsymbol{\lambda'_{pre} \lambda_{pre}}\right)}  \left|\left| F\tilde{\lambda}^2 \sum_{t=1}^{T_0}  \boldsymbol{\epsilon}_{n_0,t} \right|\right|_2 = \\
      \frac{F\tilde{\lambda}^2}{T_0\phi_{min}\left(\frac{1}{T_0}\boldsymbol{\lambda'_{pre} \lambda_{pre}}\right)} \left|\left|  \sum_{t=1}^{T_0}  \boldsymbol{\epsilon}_{n_0,t} \right|\right|_2 \quad \quad \quad (*)
\end{gather*}
where the second inequality follows as we define $\phi_{min}\left(\frac{1}{T_0}\lambda'_{pre} \lambda_{pre} \right)$ as the smallest eigenvalues of the matrix $\frac{1}{T_0}\lambda'_{pre} \lambda_{pre}$. Now, note that each row $\sum_{t=1}^{T_0}  \boldsymbol{\epsilon}_{n_0,t}$ is a sum of $T_0$ iid variables, so that  $\sum_{t=1}^{T_0}  \epsilon_{j,t} \sim SubG(T^2_0\sigma^2)$. However, then $\left|\left|  \sum_{t=1}^{T_0}  \boldsymbol{\epsilon}_{n_0,t} \right|\right|_2$  is the Eucledian norm of $n_0$ $SubG$ variables. Thus, we can use a corollary of Theorem 1 in \citet{hsu11} which is a version of the more general Hanson-Wright inequality \citep{han71}: 
\begin{cor}[to \citet{hsu11}'s Theorem 1]
    Suppose $\boldsymbol{A}$ is a $n_0 \times n_0$ matrix. If $z \in \R^{n_0}$ is \texttt{subG}$_{n_0}(\xi^2)$ vector, then:
    $$ P(\boldsymbol{z'Az} < 2\xi^2 trace(\boldsymbol{A})+h) \geq 1 - \exp\left( - \frac{h^2}{4\xi^2 ||\boldsymbol{A}||_{op}}\right) $$
    where $||\boldsymbol{A}||_{op}$ is the operator norm of matrix $\boldsymbol{A}$.
\end{cor}
Let us apply this result with
$$ \boldsymbol{z} = \sum_{t=1}^{T_0}\boldsymbol{\epsilon}_{n_0,t} \quad \quad \xi^2 = T_0\sigma^2 \quad \quad \boldsymbol{A} = \boldsymbol{I}_{n_0} \quad \quad  $$
where we have $trace(\boldsymbol{I}_{n_0})=n_0$ and, since we are on Eucledian space, $||\boldsymbol{I}_{n_0}||_{op} = \phi_{max}(\boldsymbol{I}'_{n_0} \boldsymbol{I}_{n_0}) = 1$, i.e. the operator norm is the largest eigenvalue. Thus, choosing $h = h \sqrt{T_0} \sigma $, we can find:
$$ P\left(\left|\left|  \sum_{t=1}^{T_0}  \boldsymbol{\epsilon}_{n_0,t} \right|\right|_2^2 < 2 T_0 \sigma^2 n_0 + h\sqrt{T_0} \sigma  \right) \leq 1-\exp\left(  - \frac{h^2}{4}\right) $$
Although the last result is with respect to the \textit{squared} Euclidean norm, we have $\left|\left|  \sum_{t=1}^{T_0}  \boldsymbol{\epsilon}_{n_0,t} \right|\right|_2 \geq 0$ and so we can write: 
$$ P\left(\left|\left|  \sum_{t=1}^{T_0}  \boldsymbol{\epsilon}_{n_0,t} \right|\right|_2 < \sqrt{2 T_0 \sigma^2 n_0 + h\sqrt{T_0} \sigma } \right) \leq 1- \exp\left(  - \frac{h^2}{4}\right) $$

Plugging in the last result to $(*)$, we obtain with high probability that: 
\begin{gather*}
    \frac{F\tilde{\lambda}^2}{T_0\phi_{min}\left(\frac{1}{T_0}\boldsymbol{\lambda'_{pre} \lambda_{pre}}\right)} \left|\left|  \sum_{t=1}^{T_0}  \boldsymbol{\epsilon}_{n_0,t} \right|\right|_2 \leq \\ \frac{F\tilde{\lambda}^2}{T_0\phi_{min}\left(\frac{1}{T_0}\boldsymbol{\lambda'_{pre} \lambda_{pre}}\right)} \sqrt{2 T_0 \sigma^2 n_0 + h\sqrt{T_0} \sigma }  = \\
   \frac{F\tilde{\lambda}^2}{\phi_{min}\left(\frac{1}{T_0}\boldsymbol{\lambda'_{pre} \lambda_{pre}}\right) } \sqrt{\frac{2\sigma^2n_0}{T_0} + \frac{h\sigma}{T_0^{1.5}}}
\end{gather*}
as required by Claim \ref{er_prop_claim_4}.

We will proceed to show Claim \ref{er_prop_claim_2}. Note that we are seeking a lower bound for expression (\ref{att_cl_eq_2}). Analogically to the last derivation for an \textit{upper} bound, here we can take the \textit{largest} eigenvalue of $\left(\frac{1}{T_0}\boldsymbol{\lambda}'_{pre} \boldsymbol{\lambda}_{pre} \right)$ and the \textit{smallest} $\underaccent{\tilde}{\lambda}$:
\begin{gather*}
    \frac{1}{n_1} \sum_{i=n_0+1}^{N} \boldsymbol{\lambda}_s \left( \frac{\boldsymbol{\lambda}'_{pre}\boldsymbol{\lambda}_{pre} }{T_0}  \right)^{-1} \boldsymbol{\lambda}'_{pre} \boldsymbol{e}_{i,pre} \geq \\
    \frac{1}{T_0n_1} \sum_{i=n_0+1}^N \sum_{t=1}^{T_0} \frac{F \underaccent{\tilde}{\lambda}^2 e_{it} }{\phi_{max}\left(\frac{1}{T_0}\boldsymbol{\lambda}'_{pre} \boldsymbol{\lambda}_{pre} \right)}= \frac{F \underaccent{\tilde}{\lambda}^2}{T_0n_1\phi_{max}\left(\frac{1}{T_0}\boldsymbol{\lambda}'_{pre} \boldsymbol{\lambda}_{pre} \right)} \underbrace{\sum_{i=n_0+1}^N \sum_{t=1}^{T_0} e_{it}}_{n_1T_0subG(\sigma^2)} \equiv Q
\end{gather*}
where we define the maximum eigenvalue of $\left(\frac{1}{T_0}\boldsymbol{\lambda}'_{pre} \boldsymbol{\lambda}_{pre} \right)$ as $\phi_{max}\left(\frac{1}{T_0}\boldsymbol{\lambda}'_{pre} \boldsymbol{\lambda}_{pre} \right)$. The $subG$-ness of the sum of error follows from their independence. Note that we define $Q$ as the last expression we found. Since $Q$ is a sum of $n_1T$ independent RV with $\sigma^2_Q = \left(\frac{\sigma F \underaccent{\tilde}{\lambda}^2}
{\phi_{max} ^2}\right)^2$, we find a Chernoff bound for any $h$:
\begin{gather*}
    P(Q < -h) = P(-Q > h) = P(e^{-Qv} > e^{hv}) \leq \\ \frac{E[e^{-Qv}]}{e^{hv}} \leq \frac{\exp\left(\frac{\sigma_Q^2v^2}{2}\right)}{\exp(hv)} = \exp\left(\frac{\sigma_Q^2v^2}{2} - hv\right)
\end{gather*}
where the first inequality follows from Markov's inequality and the second inequality is due to the definition of $SubG$ random variables. Then, we find the tightest bound by maximising (????) $\frac{\sigma_Q^2v^2}{2} - hv$. FOC is $\sigma^2_Qv =h $. So, the Chernoff bound is: 
$$
    P(Q < -h)  \leq \exp\left( -\frac{h^2}{2\sigma^2} \right)
$$
Since the above holds for any $h$, take $h\sigma_Q$: 
$$P\left( Q> -\frac{h\sigma F \underaccent{\tilde}{\lambda}^2}
{\phi_{max} ^2}\right) = 1- \exp\left(-h^2/2 \right) $$
as suggested by Claim \ref{er_prop_claim_2}, with probability at least  $1-\exp\left( -\frac{h^2}{2}\right)$ we have:
\begin{gather*}
    \frac{1}{n_1} \sum_{i=n_0+1}^{N} \boldsymbol{\lambda}_s \left( \frac{\boldsymbol{\lambda}'_{pre}\boldsymbol{\lambda}_{pre} }{T_0}  \right)^{-1} \boldsymbol{\lambda}'_{pre} \boldsymbol{e}_{i,pre} \geq Q > -\frac{h\sigma F \underaccent{\tilde}{\lambda}^2}{\phi_{max} ^2}
\end{gather*}


Now, we will show the last Claim $\ref{er_prop_claim_3}$ for (\ref{att_cl_eq_3}). For some given $i$ in the treated group $\{n_0+1, \dots, n_0+n_1\}$, we have that $\epsilon_{is}$ is independent from $\epsilon_{js}$ for the control group, $j \in \{1,2, \dots, n_0\}$. So, we can write for some $h$:
\begin{align*}
    E\left[ e^{h(\epsilon_{is} - \sum_{j=1}^{n_0} \hat{w}_i^j \epsilon_{js})}\right] \stackrel{ind.}{=} E[e^{h\epsilon_{is}}] \prod_{j=1}^{n_0} E\left[e^{-s\hat{w}_{ij} \epsilon_{js}} \right] \\
    \stackrel{subG}{\leq} e^{\frac{\sigma^2h^2}{2}} \prod_{j=1}^{n_0} e^{\frac{\sigma^2(-h\hat{w}_{ij})^2}{2}} = 
    \exp\left(\frac{h^2\sigma^2(1+\sum_{j}\hat{w}_{ij} )}{2}\right)
\end{align*}
which means by definition of \texttt{subG} that:
$$ \epsilon_{is} - \sum_{j=1}^{n_0} \hat{w}_i^j \epsilon_{js} \sim SubG\left(\sigma^2\left(1+\sum_{j=1}^{n_0}\hat{w}^2_{ij} \right) \right) $$
 
However, if we further aggregate over $i$ in the treated group, we have a problem. Define $z_i \equiv \epsilon_{is} - \sum_{j=1}^{n_0} \hat{w}_i^j\epsilon_{js}$. For some $k$ and $i$, we do not have independence of $z_i$ and $z_k$ as for $s>T_0$:
\begin{align*}
    E[z_iz_k] =& E\left[\left(\epsilon_{is} - \sum_{j=1}^{n_0} \hat{w}_{ij}\epsilon_{js}\right) \left(\epsilon_{ks} - \sum_{j=1}^{n_0} \hat{w}_{kj}\epsilon_{js}\right)  \right] = \\
    =& E[\epsilon_{is}\epsilon_{ks}] - E[\epsilon_{is}\sum_{j=1}^{n_0} \hat{w}_{kj}\epsilon_{js}] - E[\epsilon_{ks}\hat{w}_{ij}\epsilon_{js}] + E\left[ \left( \sum_{j=1}^{n_0} \hat{w}_{ij}\epsilon_{js} \right) \left( \sum_{j=1}^{n_0} \hat{w}_{kj}\epsilon_{js} \right) \right] = \\
    = & 0 - 0 - 0 + \sum_{i=1}^{n_0}\hat{w}^2_{ij} E[\epsilon^2_{is}] > 0
\end{align*}
due to independence of $\epsilon_{is}$, $E[\epsilon^2_{is}] \neq 0$ and the fact that every $\hat{w}_{kj}$ is independent of all post-treatment errors $\epsilon_{is}$.

We can still apply Lemma \ref{lem_subg} for the sum of not necessarily independent \texttt{subG} variables, as $z_i=\epsilon_{is} - \sum_{j=1}^{n_0} \hat{w}_i^j \epsilon_{js} $ are  \texttt{subG}.  Thus, we obtain: 
\begin{gather*}
    \frac{1}{n_1}\sum_{i=1}^{n_1}\epsilon_{is} - \sum_{j=1}^{n_0} \hat{w}_i^j \epsilon_{js} \sim SubG\left(\frac{\sigma^2}{n_1^2} \left( \sum_{i=n_0+1}^N\sqrt{1+\sum_{j=1}^{n_0} \hat{w}^2_{ij}} \right)^2 \right)
\end{gather*}
As in the case of Claim \ref{er_prop_claim_2}, we can take a Chernoff bound and use the same argument. So, with probability $1-\exp(-h/2)$:
\begin{gather*}
    \frac{1}{n_1}\sum_{i=1}^{n_1}\epsilon_{is} - \sum_{j=1}^{n_0} \hat{w}_{ij} \epsilon_{js} \leq \frac{h\sigma}{n_1} \left(\sum_{i=n_0+1}^N \sqrt{ 1 + \sum_{j=1}^{n_0} \hat{w}^2_{ij}} \right) \leq h \frac{\sigma}{n_1} n_1 \sqrt{2}
\end{gather*}
and the last inequality follows due to the properties of the weights.

Now, let us call the events as follows:
\begin{gather*}
    E_1 =  \left\{ \left( \frac{1}{n_1} \sum_{i=n_0+1}^{n_1} \boldsymbol{\lambda}_s (\boldsymbol{\lambda'_{pre} \lambda_{pre}})^{-1} \boldsymbol{\lambda}'_{pre} \left( \sum_{j=1}^{n_0} \hat{w}_{ij}   \boldsymbol{e}_{j,pre}\right)   \right) \leq \frac{F\tilde{\lambda}^2}{\phi_{min} } \sqrt{\frac{2\sigma^2n_0}{T_0} + \frac{h\sigma}{T_0^{1.5}} } \right\}  \\
    E_2 =  \left\{ \frac{1}{n_1} \sum_{i=n_0+1}^{N} \boldsymbol{\lambda}_s \left( \frac{\boldsymbol{\lambda}'_{pre}\boldsymbol{\lambda}_{pre} }{T_0}  \right)^{-1} \boldsymbol{\lambda}'_{pre} \boldsymbol{e}_{i,pre}   
    \geq - \frac{h\sigma F \underaccent{\tilde}{\lambda}^2}{\phi_{max} ^2}   \right\}  \\
    E_3 =\left\{ \frac{1}{n_1} \sum_{i=n_0+1}^N \left[ \epsilon_{is} - \sum_{j=1}^{n_0} \hat{w}^j \epsilon_{js} \right] \leq h \sigma \sqrt{2} \right\}
\end{gather*}
where we use the shorthand to denote the minimum eigenvalues $\phi_{min} = \phi_{min} \left(\frac{1}{T_0}\boldsymbol{\lambda'_{pre} \lambda_{pre}}\right)$ and maximum eigenvalue $\phi_{max} = \phi_{max} \left(\frac{1}{T_0}\boldsymbol{\lambda'_{pre} \lambda_{pre}}\right)$

Then, we can use our results from the four claims and apply Frechet's inequality to show that:
\begin{gather*}
    Pr\left[\widehat{\tau}_s - \tau_s < 
    \frac{F\tilde{\lambda}^2}{\phi_{min} } \sqrt{\frac{2\sigma^2n_0}{T_0} + \frac{h\sigma}{T_0^{1.5}} }
    - \left(- \frac{h\sigma F \underaccent{\tilde}{\lambda}^2}{\phi_{max} ^2}\right)
    + h \sigma \sqrt{2} \right] \\
    \geq Pr\left( E_1 \cap E_2 \cap E_3 \right) \geq \\
     \max\left(0, 1 - \frac{1}{\exp(0.25h^2)}  - \frac{2}{\exp(0.5h^2)} \right) \geq  \max\left(0, 1 - \frac{3}{\exp(0.25h^2)} \right)
\end{gather*}
where $h$ will be set such that the last expression is bigger than 0.
Lastly, we apply a union bound to the event that that  $E_1$, $E_2$ and $E_3$ hold simultaneously for the absolute value of the estimation error (and not just for $\widehat{\tau}_s - \tau_s$ as done so far), then it follows: 
\begin{gather*}
    |\widehat{\tau}_s - \tau_s| < 
     \frac{F\tilde{\lambda}^2}{\phi_{min} } \sqrt{\frac{2\sigma^2n_0}{T_0} + \frac{h\sigma}{T_0^{1.5}} }
    + \left( \frac{h\sigma F \underaccent{\tilde}{\lambda}^2}{\phi_{max} ^2}\right)
    + h \sigma \sqrt{2}
\end{gather*}
with probability $1 - \frac{3}{\exp(0.25h^2)}$, due to Frechet inequality for logical disjunction, i.e. $P(A\cup B) \geq \max(P(A), P(B))$. Note that in the statement of the theorem we amend the notation, so that for clarity $\lambda_{min}$ and $\lambda_{max}$ denote respectively the minimum and maximum common factor $\lambda_{fs}$ in absolute value. However, we have $\lambda_{min} \equiv \underaccent{\tilde}{\lambda}$ and $\lambda_{max} \equiv \tilde{\lambda}$
\end{proof}


\newpage

\subsection{Proof of Proposition \ref{prop_DiD_est_err}}\label{proof_bound_DiD}

\propDiDesterr*
\begin{proof}
    We follow an analogical approach to Proposition \ref{prop_iDiD_cons}. So, we begin by estimating $\hat{\tau}$ via DiD and so fit the model: $$ y_{it} = \rho  + \gamma_i + \delta_t + D_{it}\tau + e_{it} $$
    where $\rho$ is an intercept, $\gamma_i$ is an individual fixed effect and  $\delta_t$ is a time fixed effect. In matrix form, let us denote by $\boldsymbol{Z}$ the $NT \times (1+N-1+T-1)$ matrix of dummies, used for estimating $\rho  + \gamma_i + \delta_t$, so that 
    $$ \boldsymbol{y} = \boldsymbol{Z} \begin{pmatrix}
        \rho \\ 
        \boldsymbol{\gamma} \\
        \boldsymbol{\delta} 
    \end{pmatrix} + 
    \boldsymbol{D} \tau + \boldsymbol{e}  $$
    where $\boldsymbol{D}$ is a $(NT \times 1)$ vector of treatment indicators.
    By Frisch-Waugh-Lovell's Theorem, we have that:
    \begin{align*}
         \hat{\tau} =& \left(\boldsymbol{\hat{u}}'_{D \cdot \rho, \gamma, \delta} \boldsymbol{\hat{u}}_{D \cdot \rho, \gamma , \delta} \right)^{-1} \boldsymbol{\hat{u}}'_{D \cdot \rho, \gamma, \delta} \boldsymbol{u}_{y \cdot \rho, \gamma, \delta} = \\
         =& \left( \boldsymbol{D}'(\boldsymbol{I}_{NT} - \boldsymbol{Z}(\boldsymbol{Z}'\boldsymbol{Z})^{-1} \boldsymbol{Z}')\boldsymbol{D} \right)^{-1} \boldsymbol{D}'(\boldsymbol{I}_{NT} - \boldsymbol{Z}(\boldsymbol{Z}'\boldsymbol{Z})^{-1} \boldsymbol{Z}' ) \boldsymbol{y} = \\
         =&\left(\boldsymbol{\hat{u}}'_{D \cdot \rho, \gamma, \delta}\boldsymbol{D} \right)^{-1} \boldsymbol{\hat{u}}'_{D \cdot \rho, \gamma, \delta} \boldsymbol{y}
    \end{align*}
    where $\boldsymbol{\hat{u}}'_{D \cdot \rho, \gamma, \delta}$ are the residuals\footnote{We do not need to worry about predicted values outside of $[0,1]$ range, as we are regressing only on fixed effects and not on any continuous variable. I would like to thank Frank DiTraglia for pointing this out.} from regressing $D_{it}$ on $\rho$, $\gamma_i$ and $\delta_t$ and analogically for $\boldsymbol{\hat{u}}'_{y \cdot \rho, \gamma, \delta}$ being the residuals from regressing $y_{it}$ on $\rho$, $\gamma_i$ and $\delta_t$. The second row follows due to idempotence of matrix $\boldsymbol{M} \equiv (\boldsymbol{I}_{NT} - \boldsymbol{Z}(\boldsymbol{Z}'\boldsymbol{Z})^{-1} \boldsymbol{Z}' )$, i.e. $\boldsymbol{M}'\boldsymbol{M}=\boldsymbol{M}$. Next, we substitute the factor model generating $\boldsymbol{y}$ to obtain:  
    \begin{align*}
         \hat{\tau} =& \left( \boldsymbol{D}'\boldsymbol{M}\boldsymbol{D} \right)^{-1} \boldsymbol{D}'\boldsymbol{M}  \boldsymbol{y} = \\
         =& \left( \boldsymbol{D}'\boldsymbol{M}\boldsymbol{D} \right)^{-1} \boldsymbol{D}'\boldsymbol{M} (Vec(\boldsymbol{\theta X'}) + Vec(\boldsymbol{\lambda \mu'}) + \boldsymbol{D}\tau + \boldsymbol{e} ) = \\
         =& \tau  + \left(\frac{ \boldsymbol{D}'\boldsymbol{M}\boldsymbol{D}}{NT} \right)^{-1} \frac{\boldsymbol{D}'\boldsymbol{M} Vec(\boldsymbol{\theta X'})}{NT} + 
         \left(\frac{ \boldsymbol{D}'\boldsymbol{M}\boldsymbol{D}}{NT} \right)^{-1} \frac{\boldsymbol{D}'\boldsymbol{M} Vec(\boldsymbol{\lambda \mu'})}{NT} + 
        \left(\frac{ \boldsymbol{D}'\boldsymbol{M}\boldsymbol{D}}{NT} \right)^{-1} \frac{\boldsymbol{D}'\boldsymbol{M} \boldsymbol{e}}{NT} 
    \end{align*}
    where we need to use the $Vec(\boldsymbol{A})$ operator which stacks all columns of matrix $\boldsymbol{A}$. Given the assumption on the DGP, let us hold $T$ and $n_0$ fixed and let $n_1 \to \infty$, so that $N \to \infty$. Then, we have for all $(i,j)\in \{1,2, \dots, N\}$ and $(j,s) \in \{1,2, \dots, T\}$ by the Law of Large Numbers:  
    $$ \frac{\boldsymbol{D}'\boldsymbol{M} \boldsymbol{e}}{NT} \stackrel{p}{\to} E[g(D_{it})e_{js}] = 0 \quad \quad \frac{\boldsymbol{D}'\boldsymbol{M} Vec(\boldsymbol{\theta X'})}{NT}  \stackrel{p}{\to} E[f(D_{it})x_{j}] = 0 $$
    where $g(.)$ and $f(.)$ are some measurable functions. Crucially, we note that while the elements of $\boldsymbol{M}$ could be correlated with with the elements of $\boldsymbol{D}$, what matters here is the fact that both $\boldsymbol{x}$ and $\boldsymbol{\epsilon}$ are assume to be independent from   $\boldsymbol{M}$ and $\boldsymbol{D}$. Lastly, to get the probability limits we need either $n_1 \to \infty$ or $n_0 \to \infty$ or both.  Note that in the expression above $\frac{ \boldsymbol{D}'\boldsymbol{M}\boldsymbol{D}}{NT}$ is a scalar, as $\boldsymbol{D}$ is $(NT\times1)$ vector.  Assuming (as we shall show below) that $\left(\frac{ \boldsymbol{D}'\boldsymbol{M}\boldsymbol{D}}{NT} \right)^{-1} \stackrel{p}{\to} c \neq 0$, then we get as $n_1 \to \infty$:
    \begin{gather}
        \hat{\tau} - \tau \stackrel{p}{\to}  
    \left(\underset{N \to \infty}{\plim}\left(\frac{ \boldsymbol{D}'\boldsymbol{M}\boldsymbol{D}}{NT} \right)\right)^{-1} \underset{N \to \infty}{\plim} \left( \frac{\boldsymbol{D}'\boldsymbol{M} Vec(\boldsymbol{\lambda \mu'})}{NT}  \right)
    \label{pro_lim_est_err_DiD}
    \end{gather}
    with the result following from the Law of Large Numbers and Slutsky's Theorem.
    Next, we will obtain a more precise estimate of the probability limit. We note that $\boldsymbol{M}$ is a function of the matrix of dummies $\boldsymbol{Z}$ and so only of $n_1$, $n_0$, $T_0$ and $T$. To see this, consider the case for $N=3$ and $T=2$ where $\boldsymbol{Z}$ is $6\times4$ matrix:
    \begin{gather*}
        \begin{pmatrix}
            1 & 0 & 0 & 0 \\ 
            1 & 0 & 0 & 1 \\
            1 & 1 & 0 & 0 \\ 
            1 & 1 & 0 & 1 \\ 
            1 & 0 & 1 & 0 \\ 
            1 & 0 & 1 & 1 \\
        \end{pmatrix}
    \end{gather*}
    More generally, the first column contains $1$-s for the intercept, the next $N-1$ columns contain dummies for all individuals in the data and the last $T-1$ are with dummies for the treated periods. 
    
    With this insight in mind, we shall find an \textit{exact} expression for $\boldsymbol{M} = \left(\boldsymbol{I}_{NT} - \boldsymbol{Z}(\boldsymbol{Z}'\boldsymbol{Z})^{-1} \boldsymbol{Z}' \right)$. We begin by showing:
    \begin{gather}
        \boldsymbol{Z'Z} = 
        \begin{pmatrix}
            NT & \boldsymbol{\iota}'_{N-1} T & \boldsymbol{\iota}'_{T-1} N \\
            \boldsymbol{\iota}_{N-1} T & T \boldsymbol{I}_{N-1} & \boldsymbol{1}_{(N-1)\times(T-1)} \\
            \boldsymbol{\iota}_{T-1} N & \boldsymbol{1}_{(T-1)\times(N-1)} & N \boldsymbol{I}_{T-1}  \\
        \end{pmatrix}
        \label{tz_z}
    \end{gather}
    where $\boldsymbol{\iota}_K$ is $1\times K$ column vector of ones and $\boldsymbol{1}_{K\times Q}$ is a $(K\times Q)$ matrix of ones. This matrix be found by guess and verify. Next, we define by $\boldsymbol{H}$ the matrix from (\ref{tz_z}): 
    \begin{gather*}
        \boldsymbol{H} = 
        \begin{pmatrix}
            T \boldsymbol{I}_{N-1} & \boldsymbol{1}_{(N-1)\times(T-1)} \\
             \boldsymbol{1}_{(T-1)\times(N-1)} & N \boldsymbol{I}_{T-1} 
        \end{pmatrix}
    \end{gather*}
    We can then apply the matrix inversion lemma \citep[][p.108]{ber09}:
    \begin{gather}
        \boldsymbol{B}^{-1} = 
        \begin{pmatrix}
            \boldsymbol{P} & \boldsymbol{Q} \\
            \boldsymbol{R} & \boldsymbol{S}
        \end{pmatrix}^{-1} =
        \begin{pmatrix}
            \boldsymbol{P}^{-1} + \boldsymbol{P}^{-1}\boldsymbol{Q}(\boldsymbol{S} -\boldsymbol{R}\boldsymbol{P}^{-1}\boldsymbol{Q} )^{-1} \boldsymbol{R} \boldsymbol{P}^{-1} & 
            -\boldsymbol{P}^{-1} \boldsymbol{Q} (\boldsymbol{S} -\boldsymbol{R}\boldsymbol{P}^{-1}\boldsymbol{Q} )^{-1} \\
            -(\boldsymbol{S} -\boldsymbol{R}\boldsymbol{P}^{-1}\boldsymbol{Q} )^{-1} \boldsymbol{R} \boldsymbol{P}^{-1} & 
            (\boldsymbol{S} -\boldsymbol{R}\boldsymbol{P}^{-1}\boldsymbol{Q} )^{-1} 
        \end{pmatrix}
        \label{minv_lem}
    \end{gather}
    where all the inverted matrices are assumed to be non-singular and also use the Sherman-Morrison-Woodbury lemma \citep[][p.141]{ber09}:
    \begin{gather}
        (\boldsymbol{C} + \boldsymbol{a}\boldsymbol{b}')^{-1} = \boldsymbol{C}^{-1} - \frac{\boldsymbol{C}^{-1}\boldsymbol{a}\boldsymbol{b}' \boldsymbol{C}^{-1}}{1+\boldsymbol{b}'\boldsymbol{C}^{-1} \boldsymbol{a}}
    \label{smw_lem}
    \end{gather}
    where $\boldsymbol{a}$ and $\boldsymbol{b}$ are column vectors. If we use these two results, we can find:
    \begin{gather*}
        \boldsymbol{H}^{-1} = 
        \begin{pmatrix}
            \frac{1}{T}\left(\boldsymbol{I}_{n-1} +\frac{T-1}{T+N-1} \boldsymbol{1}_{(N-1) \times (N-1)} \right) & 
            -\frac{1}{T+N-1} \boldsymbol{1}_{(N-1)\times(T-1)} \\
            -\frac{1}{T+N-1} \boldsymbol{1}_{(T-1) \times (N-1)} & 
            \frac{1}{N} \left(\boldsymbol{I}_{T-1} +\frac{N-1}{T+N-1} \boldsymbol{1}_{(T-1) \times (T-1)} \right)
        \end{pmatrix}
    \end{gather*}
    The calculations are tedious, but not so challenging mentally. Next, we can use $\boldsymbol{H}^{-1}$ in finding the inverse of $\boldsymbol{Z'Z}$. We achieve this in three steps. Firstly, use the matrix inversion lemma (\ref{minv_lem}) to the expression for $\boldsymbol{Z'Z}$ in (\ref{tz_z}) where we use $\boldsymbol{H}$ as our bottom-right block $\boldsymbol{S}$ in (\ref{minv_lem}). Secondly, we use $\boldsymbol{H}^{-1}$  and the Sherman-Morrison-Woodburry lemma to find the bottom right panel of the inverted lemma, that is  $(\boldsymbol{S} -\boldsymbol{R}\boldsymbol{P}^{-1}\boldsymbol{Q} )^{-1} $. Thirdly, we plug-in this expression in the other three panels, noting that it appears in all of them. The derivations are significantly simplified by the fact that in $\boldsymbol{Z'Z}$ the top-left panel is a scalar and so the top-right and bottom-left blocks are actually vectors.
    
    Using this procedure, we can find:
    \begin{gather*}
        (\boldsymbol{Z'Z}) = 
        \begin{pmatrix}
            \frac{N+T-1}{NT} & -\frac{1}{T} \boldsymbol{\iota}'_{N-1} & -\frac{1}{N} \boldsymbol{\iota}'_{T-1} \\
            -\frac{1}{T} \boldsymbol{\iota}_{N-1} & \frac{1}{T} \left( \boldsymbol{I}_{N-1} + \boldsymbol{1}_{(N-1)\times (N-1)} \right) & \boldsymbol{O}_{(N-1)\times(T-1)} \\
            -\frac{1}{N} \boldsymbol{\iota}_{T-1} &  \boldsymbol{O}_{(T-1)\times(N-1)} & \frac{1}{N} \left( \boldsymbol{I}_{T-1} + \boldsymbol{1}_{(T-1)\times (T-1)} \right)
        \end{pmatrix}
    \end{gather*}
    The next step in the derivation is to guess and verify the $(NT \times NT)$ matrix $\boldsymbol{Z(Z'Z)^{-1}Z'}$ has a block symmetric structure:
    \begin{gather*}
        \begin{pmatrix}
            \boldsymbol{A} & \boldsymbol{B} & \boldsymbol{B} & \dots & \boldsymbol{B} & \boldsymbol{B} \\
            \boldsymbol{B}  & \boldsymbol{A}& \boldsymbol{B} & \dots & \boldsymbol{B} & \boldsymbol{B} \\
            \boldsymbol{B}  & \boldsymbol{B} &  \boldsymbol{A} & \dots & \boldsymbol{B} & \boldsymbol{B} \\
            \vdots & \vdots &  \vdots & \ddots & \vdots & \vdots \\
            \boldsymbol{B}  & \boldsymbol{B} & \boldsymbol{B} & \dots & \boldsymbol{A} & \boldsymbol{B}\\
            \boldsymbol{B}  & \boldsymbol{B} & \boldsymbol{B} & \dots & \boldsymbol{B} & \boldsymbol{A}\\
        \end{pmatrix}
    \end{gather*}
    where $\boldsymbol{A} = \frac{1}{NT} \left( T \boldsymbol{I}_T + (N-1)\boldsymbol{1}_T \right)$ and $\boldsymbol{B} = \frac{1}{NT} \left( T \boldsymbol{I}_T - \boldsymbol{1}_{T\times T} \right)$. On each block row, we have $N-1$ matrices $\boldsymbol{B}$  and one matrix $\boldsymbol{A}$ on the block diagonal for some treated observation.  In specific, it is given by:
    \begin{gather*}
        \boldsymbol{Z(Z'Z)^{-1}Z'} =  \frac{1}{NT} \times \\
        \left.
        \begin{matrix}
        \vphantom{s_{m+1}}\\
        \vphantom{\vdots}\\
        \vphantom{s_{m+k}} \\
        \vphantom{s_{m+k}} \\
        \vphantom{s_{m+k}}
        \end{matrix}
        NT \right \lbrace  \hfill
        \overbrace{\begin{pmatrix}
           T \boldsymbol{I}_T + (N-1)\boldsymbol{1}_T & \smash[b]{\block{T(N-1)} } \\
           T \boldsymbol{I}_T - \boldsymbol{1}_{T\times T} & T \boldsymbol{I}_T + (N-1)\boldsymbol{1}_T \quad  \smash[b]{\blockk{T(N-2)} } \\
           \vdots & \vdots \\
            T \boldsymbol{I}_T - \boldsymbol{1}_{T\times T}  & \quad \quad \cdots \quad \quad  \cdots \quad  \quad \quad \quad T \boldsymbol{I}_T + (N-1)\boldsymbol{1}_T
        \end{pmatrix}}^{NT}
        \begin{matrix}
    \end{matrix}
    \end{gather*}
    It is then easy to calculate $\boldsymbol{M} = \left(\boldsymbol{I}_{NT} - \boldsymbol{Z(Z'Z)^{-1}Z'}  \right)$: 
    \begin{gather}
    \frac{1}{NT}\begin{pmatrix}
        T(N-1) \boldsymbol{I}_T - (N-1)\boldsymbol{1}_{T\times T} & -T \boldsymbol{I}_T + \boldsymbol{1}_{T\times T} & \cdots & -T \boldsymbol{I}_T + \boldsymbol{1}_{T\times T} \\
        -T \boldsymbol{I}_T + \boldsymbol{1}_{T\times T} &  T(N-1) \boldsymbol{I}_T - (N-1)\boldsymbol{1}_{T\times T} & \cdots &  -T \boldsymbol{I}_T + \boldsymbol{1}_{T\times T} \\
        \vdots & \vdots & \ddots  & \vdots \\
        -T \boldsymbol{I}_T + \boldsymbol{1}_{T\times T} & -T \boldsymbol{I}_T + \boldsymbol{1}_{T\times T}   & \cdots  & T(N-1) \boldsymbol{I}_T - (N-1)\boldsymbol{1}_{T\times T} \\
    \end{pmatrix}
    \label{M_sol}
     \end{gather}
     Given this expression, we will now find the probability limit of:
     $$  \left(\frac{ \boldsymbol{D}'\boldsymbol{M}\boldsymbol{D}}{NT} \right)^{-1}  \left( \frac{\boldsymbol{D}'\boldsymbol{M} Vec(\boldsymbol{\lambda \mu'})}{NT}  \right) $$
     Let us first find an expression for the scalar $\frac{\boldsymbol{D}'\boldsymbol{M}\boldsymbol{D}}{NT}$. At this point, we realise that the vector $\boldsymbol{D}$ contains only zeros and ones. Since $T_0 = T-1$, we have:
     $$ \boldsymbol{D} = (\overbrace{\underbrace{0, \dots, 0}_{T}, \underbrace{0, \dots, 0}_{T} \dots, \underbrace{0, \dots, 0}_{T}}^{n_0T}, \overbrace{\underbrace{0, \dots, 0}_{T_0}, 1, 0 \underbrace{0, \dots, 0}_{T_0}, 1, \dots, \underbrace{0, \dots, 0}_{T_0}, 1}^{n_1T}  )' $$
     So, using the structure of $\boldsymbol{D}$, we can find:
     $$
         \frac{\boldsymbol{D}'\boldsymbol{M}\boldsymbol{D}}{NT} = \sum_{p=1}^{NT}\sum_{r=1}^{NT} \mathbbm{1}\{(D_p=1) \cap (D_r=1) \}m_{pr} = n_1 m_{kk} + n_1(n_1-1) m_{qk} 
     $$
     where $m_{kk} $ is any diagonal entry of $\boldsymbol{M}$ and $m_{qk}$ is any diagonal entry of the off-diagonal matrix $(-T \boldsymbol{I}_T + \boldsymbol{1}_{T\times T})$ in the expression for $\boldsymbol{M}$ in (\ref{M_sol}). Next, the last expression can be written as $n_1 \to \infty$ in the following way: 
     \begin{align}
         \frac{\boldsymbol{D}'\boldsymbol{M}\boldsymbol{D}}{NT} =& \quad n_1 \frac{[T(N-1) - (N-1)]}{NT} + n_1(n_1 - 1) \frac{(-T + 1)}{NT} \label{pr_lim_large_n0_n1} \\
         = & \quad \frac{(T-1)n_0n_1}{(n_1 + n_0)T} \quad = \quad \frac{(Tn_0 - n_0)}{(1+ \frac{n_0}{n_1})T} \nonumber \\ 
         \stackrel{p}{\to}& \quad \frac{(T-1)n_0}{T} 
         \label{pr_lim_den}
     \end{align}
     Note that in addition if we let $n_0 \to \infty$ we have $ \frac{\boldsymbol{D}'\boldsymbol{M}
     \boldsymbol{D}}{NT} \stackrel{p}{\to}\frac{T-1}{T}n_1 $ and if we have both 
     $n_0 \to \infty$  and $n_1 \to \infty$ then $ \frac{\boldsymbol{D}' 
     \boldsymbol{M}\boldsymbol{D}}{NT} 
     \stackrel{p}{\to} \frac{T-1}{T} \frac{1}{0} \approx \infty $.\footnote{Interestingly, see what happens when $n_0 \to \infty$, $n_1 \to \infty$ and $\frac{n_0}{n_1} \to c$} By Slutsky's Theorem, we have $\left(  \frac{\boldsymbol{D}'\boldsymbol{M}\boldsymbol{D}}{NT} \right) \stackrel{p}{\to} \frac{T}{T_0n_0}$.
     
     Next, assuming that $n_1 > n_0$, we can write:
     \begin{align*}
         \frac{1}{NT} \boldsymbol{D'M} Vec(\boldsymbol{\lambda \mu'}) =& \frac{1}{NT} \sum_{t=1}^{T_0} \left[ \left(\frac{n_1}{n_0}\sum_{j=1}^{n_0} \boldsymbol{\mu}_j \right) - \left(\sum_{i=n_0+1}^N\boldsymbol{\mu}_i \right)   \right] \boldsymbol{\lambda}_t \\
         -& \frac{T_0}{NT}  \left[ \left(\frac{n_1}{n_0}\sum_{j=1}^{n_0} \boldsymbol{\mu}_j \right) - \left(\sum_{i=n_0+1}^N\boldsymbol{\mu}_i \right)   \right] \boldsymbol{\lambda}_T
     \end{align*}
    given our expression for $\boldsymbol{M}$ and recalling that $\boldsymbol{D}$ contains of ones and zeros. Then, by the Law of Large Numbers and our assumptions, we can find as $n_1 \to \infty$:
    \begin{gather*}
        \frac{1}{NT} \sum_{t=1}^{T_0} \left[ \left(\frac{n_1}{n_0}\sum_{j=1}^{n_0} \boldsymbol{\mu}_j \right)   \boldsymbol{\lambda}_t \right] \stackrel{p}{\to} \boldsymbol{\bar{\mu}}_{don}  \boldsymbol{\bar{\lambda}}_{pre} \frac{T_0}{T}  \\
        \frac{1}{NT} \sum_{t=1}^{T_0} \left[ \left(\sum_{i=n_0+1}^N\boldsymbol{\mu}_i \right)   \boldsymbol{\lambda}_t \right]\stackrel{p}{\to}  E[\boldsymbol{\mu}_{tr}|D_{it}=1] \boldsymbol{\bar{\lambda}}_{pre} \frac{T_0}{T} \\
        \frac{T_0}{NT}  \left[ \left(\frac{n_1}{n_0}\sum_{j=1}^{n_0} \boldsymbol{\mu}_j \right) \boldsymbol{\lambda}_T \right]  \stackrel{p}{\to} \boldsymbol{\bar{\mu}}_{don} \boldsymbol{\lambda}_T \frac{T_0}{T} \\
        \frac{T_0}{NT}  \left[  \left(\sum_{i=n_0+1}^N\boldsymbol{\mu}_i \right)   \right] \boldsymbol{\lambda}_T \stackrel{p}{\to} E[\boldsymbol{\mu}_{tr}|D_{it}=1]\boldsymbol{\lambda}_T \frac{T_0}{T} 
    \end{gather*}
    We can then combine by Slutsky's Theorem to obtain:
    \begin{gather}
         \boldsymbol{D'M} Vec(\boldsymbol{\lambda \mu'}) \stackrel{p}{\to} \frac{T_0}{T}(\boldsymbol{\hat{\mu}}_{don} - E[\hat{\boldsymbol{\mu}}|D_{it}=1])(\bar{\boldsymbol{\lambda}}_{pre} -  \boldsymbol{\lambda}_T)
         \label{prob_lim_d_lam_mu}
    \end{gather}
    Lastly, via Slutsky's Theorem we can combine with (\ref{pr_lim_den}) to obtain the final result in the proposition, using the representation (\ref{pro_lim_est_err_DiD}):
    \begin{gather*}
         \hat{\tau} - \tau \stackrel{p}{\to}  
       \frac{(\boldsymbol{\bar{\mu} }_{don} - E[\boldsymbol{\mu}|D_{it}=1])(\bar{\boldsymbol{\lambda}}_{pre} -  \boldsymbol{\lambda}_T)}{n_0} 
    \end{gather*}
    
\end{proof}

\newpage

\subsection{Proof of Proposition \ref{prop_iDiD_cons} } \label{subsec_prop_iDiD_cons}
\propiDiDcons*

\begin{proof}
    We follow an analogical approach to Proposition \ref{prop_DiD_est_err}. So, we begin by estimating $\hat{\tau}$ via iDiD and so fit the model: $$ \tilde{y}_{it} \equiv  y_{it}- \boldsymbol{\tilde{\mu}}_i \boldsymbol{\tilde{\lambda} }_t = \rho  + \gamma_i + \delta_t + D_{it}\tau + e_{it} $$
    where $\rho$ is an intercept, $\gamma_i$ is an individual fixed effect and  $\delta_t$ is a time fixed effect and where we define as $\tilde{y}_{it}$ the values of the outcomes after subtracting the demeaned interactive fixed effects. In matrix form, let us denote by $\boldsymbol{Z}$ the $NT \times (1+N-1+T-1)$ matrix of dummies, used for estimating $\rho  + \gamma_i + \delta_t$, so that 
    $$ \boldsymbol{\tilde{y}} = \boldsymbol{Z} \begin{pmatrix}
        \rho \\ 
        \boldsymbol{\gamma} \\
        \boldsymbol{\delta} 
    \end{pmatrix} + 
    \boldsymbol{D} \tau + \boldsymbol{e}  $$
    By Frisch-Waugh-Lovell's Theorem, we have that:
    \begin{align*}
         \hat{\tau} =& \left(\boldsymbol{\hat{u}}'_{D \cdot \rho, \gamma, \delta} \boldsymbol{\hat{u}}_{D \cdot \rho, \gamma , \delta} \right)^{-1} \boldsymbol{\hat{u}}'_{D \cdot \rho, \gamma, \delta} \boldsymbol{u}_{\tilde{y} \cdot \rho, \gamma, \delta} = \\
         =& \left( \boldsymbol{D}'(\boldsymbol{I}_{NT} - \boldsymbol{Z}(\boldsymbol{Z}'\boldsymbol{Z})^{-1} \boldsymbol{Z}')\boldsymbol{D} \right)^{-1} \boldsymbol{D}'(\boldsymbol{I}_{NT} - \boldsymbol{Z}(\boldsymbol{Z}'\boldsymbol{Z})^{-1} \boldsymbol{Z}' ) \boldsymbol{\tilde{y}} = \\
         =&\left(\boldsymbol{\hat{u}}'_{D \cdot \rho, \gamma, \delta}\boldsymbol{D} \right)^{-1} \boldsymbol{\hat{u}}'_{D \cdot \rho, \gamma, \delta} \boldsymbol{\tilde{y}}
    \end{align*}
    where $\boldsymbol{\hat{u}}'_{D \cdot \rho, \gamma, \delta}$ are the residuals\footnote{We do not need to worry about predicted values outside of $[0,1]$ range, as we are regressing only on fixed effects and not on any continuous variable. I would like to thank Frank DiTraglia for pointing this out.} from regressing $D_{it}$ on $\rho$, $\gamma_i$ and $\delta_t$ and analogically for $\boldsymbol{\hat{u}}'_{\tilde{y} \cdot \rho, \gamma, \delta}$ being the residuals from regressing $\tilde{y}_{it}$ on $\rho$, $\gamma_i$ and $\delta_t$. The second row follows due to idempotence of matrix $\boldsymbol{M} \equiv (\boldsymbol{I}_{NT} - \boldsymbol{Z}(\boldsymbol{Z}'\boldsymbol{Z})^{-1} \boldsymbol{Z}' )$, i.e. $\boldsymbol{M}'\boldsymbol{M}=\boldsymbol{M}$. Next,  we substitute the interactive fixed effects model generating $\boldsymbol{\tilde{y} }$ to obtain:  
    \begin{align*}
         \hat{\tau} =& \left( \boldsymbol{D}'\boldsymbol{M}\boldsymbol{D} \right)^{-1} \boldsymbol{D}'\boldsymbol{M}  \boldsymbol{\tilde{y}} = \\
         =& \left( \boldsymbol{D}'\boldsymbol{M}\boldsymbol{D} \right)^{-1} \boldsymbol{D}'\boldsymbol{M} (Vec(\boldsymbol{\beta X'}) + \boldsymbol{D}\tau + \boldsymbol{\epsilon} ) = \\
         =& \tau  + \left(\frac{ \boldsymbol{D}'\boldsymbol{M}\boldsymbol{D}}{NT} \right)^{-1} \frac{\boldsymbol{D}'\boldsymbol{M} Vec(\boldsymbol{\beta X'})}{NT}  + 
        \left(\frac{ \boldsymbol{D}'\boldsymbol{M}\boldsymbol{D}}{NT} \right)^{-1} \frac{\boldsymbol{D}'\boldsymbol{M} \boldsymbol{\epsilon}}{NT} 
    \end{align*}
    where we need to use the $Vec(\boldsymbol{A})$ operator which stacks all columns of matrix $\boldsymbol{A}$ \citep{cre08}. Given the assumption on the DGP, let us hold $T$ and let $n_0 \to \infty$ or  $n_1 \to \infty$, so that $N \to \infty$. Then, we have for all $(i,j)\in \{1,2, \dots, N\}$ and $(j,s) \in \{1,2, \dots, T\}$ by the Law of Large Numbers:  
    $$ \frac{\boldsymbol{D}'\boldsymbol{M} \boldsymbol{e}}{NT} \stackrel{p}{\to} E[g(D_{it})e_{js}] = 0 \quad \quad \frac{\boldsymbol{D}'\boldsymbol{M} Vec(\boldsymbol{\theta X'})}{NT}  \stackrel{p}{\to} E[f(D_{it})x_{j}] = 0 $$
    where $g(\cdot)$ and $f(\cdot)$ are some measurable functions. Note that in the expression above $\frac{ \boldsymbol{D}'\boldsymbol{M}\boldsymbol{D}}{NT}$ is a scalar, as $\boldsymbol{D}$ is $(NT\times1)$ vector.  Assuming (as we shall show below in Proposition \ref{prop_DiD_est_err}) that $\left(\frac{ \boldsymbol{D}'\boldsymbol{M}\boldsymbol{D}}{NT} \right)^{-1} \stackrel{p}{\to} c \neq 0$, then we get:
    \begin{gather*}
        N \to \infty: \quad \quad \hat{\tau}^{iDiD} \stackrel{p}{\to}   \tau 
    \end{gather*}
    with the result following from the Law of Large Numbers and Slutsky's Theorem.
\end{proof}

\newpage

\subsection{Proof of Proposition \ref{prop_optim_prob_full}}
\label{sec_proof_prop_optim_prob}

Define $\otimes$ to be the Kronecker product of two matrices and $\boldsymbol{\iota}_{n_1}$ to be a ($n_1 \times 1$) column  vector of $1$-s. We also introduce the operator $Vec(.)$ for vectorising some matrix $\boldsymbol{A}$. In particular, $Vec(A)$ takes the columns of $\boldsymbol{A}$ and stack them on top of each other.\footnote{For example, $Vec\begin{pmatrix} 1 & 2 \\ 3 & 4 \end{pmatrix} = (1,3,2,4)' $}

\propoptimprobfull*

\begin{proof}
The strategy for the proof involves three steps, in which we are turning model (\ref{CSC_mod}) into a constrained quadratic optimisation problem by stacking equations over $i$ and $t$.

Firstly, we substitute the constraint $(Random. \ Coef)$ for some given $i$ and $t$: 
\begin{gather*}
  y_{it} = \eta_i + \sum_{j=1}^{n_0} \underbrace{\omega_j}_{ind.-invariant} y_{jt} +  \sum_{j=1}^{n_0} \underbrace{\boldsymbol{x}_i \boldsymbol{\alpha}^{j}}_{ind.-specific} y_{jt} + \epsilon_{it} \quad s.t. \\
  \quad \quad \sum_{j=1}^{n_0} \omega_j + \boldsymbol{x}_i \boldsymbol{\alpha}^{j} = 1 \quad \quad  \forall j: \omega_j + \boldsymbol{x}_i \boldsymbol{\alpha}^{j} \geq 0
\end{gather*}
We can rewrite this expression as the first constrained quadratic optimisation problem in the proposition:
\begin{gather*}
     \max_{\alpha_j^k, \omega_j} \sum_{i=1}^{n_1} \sum_{t=1}^{T_0} \left( y_{it} - \eta_i - \sum_{j=1}^{n_0} y_{jt} \omega_j - \sum_{j=1}^{n_0} \sum_{k=1}^K y_{jt} \alpha^{k}_j x_i^k  \right)^2 \quad \quad s.t. \\
     \quad \quad \sum_{j=1}^{n_0} \left(\omega_j + \sum_{k=1} x_{i}^k \alpha^k_j\right) = 1 \quad \quad  \forall j: \omega_j + \sum_{k=1}^K x_{i}^k \alpha^k_j \geq 0
\end{gather*}
where every symbol represents a scalar.

Secondly, we will stack $y_{it}$ over time periods $t$ and treated units index $i$. Let us leave the constraints in their present form and focus on the expression for $y_{it}^0$ rewritten in vector form:
 \begin{gather*}
 y_{it} = \eta_i +  \boldsymbol{y'_{t, n_0}} \boldsymbol{\omega} +   \boldsymbol{y'_{t, n_0}} \boldsymbol{\alpha} \boldsymbol{x}'_i  + e_{it}
 \end{gather*}
where  $ \boldsymbol{y_{t, n_0}} $ is a $(n_0 \times 1)$ vector of outcomes for the donor pool at time $t$, $\boldsymbol{\omega}$ is $ (n_0 \times 1)$ vector of ind.-invariant weights, $\boldsymbol{x}_i$ a $(1\times K)$ row vector,  $\boldsymbol{\alpha}$ is a $(n_0 \times K)$ vector of coefficients and $\boldsymbol{y_{t, n_0}}$ is a $(n_0 \times 1)$ vector of outcomes for  donor pool. The next step is to stack over pre-treatment time period $t \leq T_0$ for some given individual $i$:
\begin{gather*}
 \underbrace{(\boldsymbol{y}^{pre}_{i})'}_{ T_0 \times 1} =  \underbrace{\boldsymbol{\iota}_{T_0}\eta_i}_{T_0 \times 1}  + 
 \underbrace{(\boldsymbol{Y}^{pre}_{n_0})'}_{T_0 \times n_0} \underbrace{\boldsymbol{\omega}}_{n_0 \times 1}  +  \underbrace{(\boldsymbol{Y}^{pre}_{n_0})'}_{T_0 \times n_0} \underbrace{\boldsymbol{\alpha}}_{n_0 \times K} \underbrace{\boldsymbol{x}_i}_{K \times 1} + \underbrace{\boldsymbol{\epsilon}_{i}}_{T_0 \times 1}
 \end{gather*}
 where $ \boldsymbol{y}^{pre}_{i}$ is a row $ 1 \times T_0$ vector of outcomes for the given treated unit $i$ in the pre-treatment period, $\boldsymbol{Y}^{pre}_{n_0}$ is a $ n_0 \times T_0$ matrix of pre-treatment outcomes for the donors and $\boldsymbol{\epsilon}_{i}$ is a $T_0 \times 1$ vector of errors. Here we can remark that $\boldsymbol{Y}^{pre}_{n_0}$ is the same matrix that appears in the full matrix of observed outcomes $\boldsymbol{\Theta}$ in (\ref{full_theta_matrix}). The matrix $\boldsymbol{\Theta}$  is the reason we pick the specific dimensions of the stacked $\boldsymbol{y}$ vectors. Next, we will combine the information over all treated individuals $i$ to form the following mapping: 
 \begin{gather}
 \underbrace{(\boldsymbol{Y}^{pre}_{n1})'}_{T_0 \times n_1} =    \underbrace{\boldsymbol{\iota}_{T_0}}_{T_0\times 1} \otimes \underbrace{\boldsymbol{\eta}'}_{ 1 \times n_1}   + \underbrace{(\boldsymbol{Y}^{pre}_{n_0})' \boldsymbol{\omega}}_{T_0 \times 1} \otimes \underbrace{\boldsymbol{\iota}'_{n_1}}_{1\times n_1} +   \underbrace{(\boldsymbol{Y}^{pre}_{n_0})'}_{T_0 \times n_0} \underbrace{\boldsymbol{\alpha}}_{n_0 \times K}  \otimes \underbrace{(\boldsymbol{X}_{n_1})'}_{K \times n_1} + \underbrace{\boldsymbol{\epsilon}}_{T_0 \times n_1}
 \label{full_stacked}
 \end{gather}
where $\boldsymbol{Y}^{pre}_{n1}$ is the $(n_1 \times T_0)$ matrix of pre-treatment outcomes that can also be found in (\ref{full_theta_matrix}), $\boldsymbol{X}_{n_1}$ is a $(n_1 \times K)$ matrix of time-invariant covariates for all treated units and $\boldsymbol{\epsilon}$ is a $(T_0 \times n_0)$ matrix. Both sides of the last equation are actually $(T_0 \times n_1)$ matrices. 

Thirdly, to turn (\ref{full_stacked}) into a quadratic optimisation problem, we can introduce the operator $Vec(.)$ for vectorising some matrix $\boldsymbol{A}$. In particular, $Vec(A)$ takes the columns of $\boldsymbol{A}$ and stack them on top of each other, as discussed above. Thus, if we apply this to (\ref{full_stacked}), we will be able to identify the parameters via solving the constrained quadratic optimisation problem:
\begin{gather*}
    \max_{\boldsymbol{\omega}, \boldsymbol{\alpha}} \quad  Vec\left((\boldsymbol{Y}^{pre}_{n1})' -     \boldsymbol{\iota}_{T_0} \otimes \boldsymbol{\eta}'  - (\boldsymbol{Y}^{pre}_{n_0})' \boldsymbol{\omega} \otimes \boldsymbol{\iota}'_{n_1} -   (\boldsymbol{Y}^{pre}_{n_0})' \boldsymbol{\alpha} \otimes \boldsymbol{X}'_{n_1}\right)'\\ 
    Vec\left((\boldsymbol{Y}^{pre}_{n1})' -     \boldsymbol{\iota}_{T_0} \otimes \boldsymbol{\eta}'  - (\boldsymbol{Y}^{pre}_{n_0})' \boldsymbol{\omega} \otimes \boldsymbol{\iota}'_{n_1} -   (\boldsymbol{Y}^{pre}_{n_0})' \boldsymbol{\alpha} \otimes \boldsymbol{X}'_{n_1} \right) \quad \quad s.t. \\ 
    (\boldsymbol{\iota}_{n_1}  \otimes \boldsymbol{\omega'} + \boldsymbol{\boldsymbol{X}}_{n_1} \boldsymbol{\alpha}' ) \boldsymbol{\iota}_{n_0}  = \boldsymbol{\iota}_{n_1} \quad \quad 
    \boldsymbol{\iota}_{n_1} \otimes \boldsymbol{\omega'} + \boldsymbol{\boldsymbol{X}}_{n_1} \boldsymbol{\alpha}' \geq \boldsymbol{0}
\end{gather*}
where as discussed $\boldsymbol{Y}^{pre}_{n1}$ is the observed $(n_1 \times T_0)$ matrix of pre-treatment outcomes for the treated group, $\boldsymbol \eta$ is a $(n_1 \times 1)$ column vector of intercepts, $\boldsymbol{Y}^{pre}_{n1}$ is the observed $(n_0 \times T_0)$ matrix of pre-treatment outcomes for the donors, $\boldsymbol{\omega}$ is the $(n_0 \times 1)$ vector of individual invariant weights, $\boldsymbol{\alpha}$ is $(n_0 \times K)$ matrix of coefficients.
\end{proof}
\newpage

\subsection{Proof of Proposition \ref{prop_city_scm} } \label{subsec_prop_city_scm}
\propcityscm*

\begin{proof}
        The pooled SC solves the problem:
    \begin{gather*}
        \max_{w_{jc} } \sum_{c=c_1+1}^C  \sum_{i=1}^n \sum_{t=1}^T \left(y_{cit} - \sum_{d=1}^{c_0} \sum_{j=1}^{n} w_{dj}y_{djt} \right)^2 \quad  \quad s.t. \quad \quad \sum_{d=1}^{c_0} \sum_{j=1}^{n} w_{dj} = 1 \quad w_{dj} \geq 0
    \end{gather*}
    Let us ignore the constraints and impose the additional requirement that every individual within a city gets the same weight $w_{dj} = \frac{w_d}{n}$. Then, we can  rewrite the objective function as:
    \begin{align*}
        \min_{\boldsymbol{w}} \sum_{c=c_1+1}^{C} \sum_{i=1}^n\sum_{t=1}^T( y_{cit} - \underbrace{\sum_{d=1}^{c_0} w_d \bar{y}_{dt}}_{h_t(w)} )^2
    \end{align*}
    where we use $\bar{y}_{dt} = \frac{\sum_{j=1}^n y_{djt}}{n}$ and $h_t(w) = \sum_{d=1}^{c_0} w_d \bar{y}_{dt}$. If we open the squares and drop the terms that do not involve $\boldsymbol{w}$, we obtain:
    \begin{align}
        \min_{\boldsymbol{w}} \sum_{t=1}^T h_t(w) \left(Ch_t(w) - 2 \sum_{c=c_1+1}^C \bar{y}_{ct} \right)
        \label{obj_psc}
    \end{align}
    where we use the same definition of $\bar{y}_{ct}$
    
    In contrast, the city-level SC solves:
    \begin{gather*}
        \min_{\boldsymbol{w}} \sum_{c=c_1+1}^C \sum_{t=1}^T (\bar{y}_{ct} - \sum_{d=1}^{c_0}w_d \bar{y}_{dt} )^2
    \end{gather*}
    Once we upon the square and drop the terms that are irrelevant for optimisation, we can rewrite this as:
    \begin{gather}
         \min_{\boldsymbol{w}} \sum_{t=1}^T h_t(w) \left(Ch_t(w) - 2 \sum_{c=c_1+1}^C \bar{y}_{ct} \right)
         \label{obj_ssc}
    \end{gather}
    which is the same as the (\ref{obj_psc}), so that the objective function is the same in the two cases.\footnote{This is related to the regression equivariance property of the Least Squares Objective function.}
\end{proof}

\end{document}